\newif\ifabstract
\abstracttrue
 \abstractfalse %Uncommenting this line gives the full version
\newif\iffull
\ifabstract \fullfalse \else \fulltrue \fi
\documentclass[11pt,letterpaper]{article}
\usepackage{amsfonts}
\usepackage{amssymb}
\usepackage{amstext}
\usepackage{amsmath}
\usepackage{xspace}
\usepackage{ntheorem}
\usepackage{graphicx}

\usepackage{url}
\usepackage{graphics}
\usepackage{colordvi}
\usepackage[colorlinks,linkcolor=black,citecolor=black,filecolor=blue,urlcolor=blue]{hyperref}
\usepackage{hyperref}
\usepackage{subfigure}
\usepackage{xcolor}
\usepackage{cleveref}

\usepackage{thmtools}
\usepackage{thm-restate}
\usepackage{enumitem}

\usepackage{booktabs, multirow, threeparttable, tabularx, makecell, array}
\newcolumntype{Y}{>{\raggedright\arraybackslash}X}
\newcommand{\res}[2]{{#1}\,\cite{#2}}
\advance \oddsidemargin by -0.8in
\newcommand{\myparskip}{3pt}
\parskip \myparskip

\newtheorem{theorem}{Theorem}[section]

\newtheorem{lemma}[theorem]{Lemma}
\newtheorem{observation}[theorem]{Observation}
\newtheorem{corollary}[theorem]{Corollary}
\newtheorem{claim}[theorem]{Claim}

\newtheorem{assumption}{Assumption}[section]
\newtheorem{definition}[theorem]{Definition}

\newenvironment{proof}{\par \smallskip{\bf Proof:}}{\hfill\stopproof}
\def\stopproof{\square}
\def\square{\vbox{\hrule height.2pt\hbox{\vrule width.2pt height5pt \kern5pt
			\vrule width.2pt} \hrule height.2pt}}

{\hfill\stopproof}

		\newenvironment{properties}[2][0]
		{
			\begin{enumerate} \setcounter{enumi}{#1}}{\end{enumerate}}

\newcommand{\set}[1]{\left\{ #1 \right\}}

\newcommand{\aalg}{{\mathcal A}}

\newcommand{\bset}{{\mathcal{B}}}

\newcommand{\poly}{\operatorname{poly}}
\newcommand{\vol}{\operatorname{vol}}
\newcommand{\OPT}{\mbox{\sf OPT}}
\newcommand{\polylog}{\mathrm{polylog}}
\newcommand{\event}{{\cal{E}}}
\newcommand{\val}{\operatorname{val}}

\newcommand{\tcopy}{\mathrm{copy}}
\newcommand{\vcopy}{\mathrm{sink}}
\newcommand{\complete}{\mathrm{complete}}

\newcommand{\expect}[2][]{\mathbb{E}_{#1}\left[{#2}\right]}
\newcommand{\ceil}[1]{\ensuremath{\left\lceil#1\right\rceil}}
\newcommand{\floor}[1]{\ensuremath{\left\lfloor#1\right\rfloor}}

\newcommand{\algglobalvertexcut}{\ensuremath{\mathsf{AlgGlobalVertexCut}}\xspace}

\newcommand{\exactsparsify}{\ensuremath{\operatorname{BatchSparsify}}\xspace}
\newcommand{\approxsparsify}{\ensuremath{\operatorname{ApproxBatchSparsify}}\xspace}\newcommand{\vertexapproxsparsify}{\ensuremath{\operatorname{VertexApproxBatchSparsify}}\xspace}

\newcommand{\mynewline}{\\}

\begin{document}
	
\begin{titlepage}
		
	\title{Almost-Optimal Approximation Algorithms for Global Minimum Cut in Directed Graphs}
	
	\author{Ron Mosenzon\thanks{Toyota Technological Institute at Chicago. Email: {\tt ron.mosenzon@ttic.edu}. Supported in part by NSF grant CCF-2402283.}}
	
	\date{}
	\maketitle
	\pagenumbering{gobble}
	
	\begin{abstract}
We develop new $(1+\epsilon)$-approximation algorithms for finding the global minimum edge-cut in a directed edge-weighted graph, and for finding the global minimum vertex-cut in a directed vertex-weighted graph.
Our algorithms are randomized, and have a running time of $O\left(m^{1+o(1)}/\epsilon\right)$ on any $m$-edge $n$-vertex input graph, assuming all edge/vertex weights are polynomially-bounded.
In particular, for any constant $\epsilon>0$, our algorithms have an almost-optimal running time of $O\left(m^{1+o(1)}\right)$.
The fastest previously-known running time for this setting, due to (Cen et al., FOCS 2021), is $\Tilde{O}\left(\min\left\{n^2/\epsilon^2,m^{1+o(1)}\sqrt{n}\right\}\right)$ for Minimum Edge-Cut, and $\Tilde{O}\left(n^2/\epsilon^2\right)$ for Minimum Vertex-Cut.

Our results further extend to the rooted variants of the Minimum Edge-Cut and Minimum Vertex-Cut problems, where the algorithm is additionally given a \emph{root vertex} $r$, and the goal is to find a minimum-weight cut separating any vertex from the root $r$.

In terms of techniques, we build upon and extend a framework that was recently introduced by (Chuzhoy et al., SODA 2026) for solving the Minimum Vertex-Cut problem in unweighted directed graphs.
Additionally, in order to obtain our result for the Global Minimum Vertex-Cut problem, we develop a novel black-box reduction from this problem to its rooted variant. Prior to our work, such reductions were only known for more restricted settings, such as when all vertex-weights are unit.
\end{abstract}
	\newpage
\end{titlepage}

\tableofcontents

\newpage
\pagenumbering{arabic}

\section{Introduction}

In recent years, there has been widespread success in the design of almost-linear time algorithm for various graph-based optimization problems.
Following this trend, we consider the design of almost-linear-time approximation algorithms for two classical problems: the (directed) Global Minimum Edge-Cut problem, and the related problem of (directed) Global Minimum Vertex-Cut.
%In recent years, graphs in real-world applications have become increasingly large, raising the need for extremely fast algorithms -- namely, those that run in close to linear time.
%In this paper, we consider the design of almost-linear-time approximation algorithms for two classical graph optimization problems: the (directed) Global Minimum Edge-Cut problem, and the related problem of (directed) Global Minimum Vertex-Cut.
In the (directed) Global Minimum Edge-Cut problem, we are given a directed $n$-vertex $m$-edge graph $G=(V,E)$ with positive integral weights $w(e)$ on the edges $e \in E$, and must find a partition of the vertices into two non-empty sets $(X,Y)$ so as to minimize the total weight of edges going from $X$ to $Y$.
Meanwhile, in the problem of (directed) Global Minimum Vertex-Cut, we are given a directed $n$-vertex $m$-edge graph $G=(V,E)$ with positive integral weights $w(v)$ on the vertices $v \in V$, and must find a \emph{vertex-cut} of $G$ -- a partition of $V$ into three sets $(L,S,R)$ with $L,R \neq \emptyset$ and with no edges $e \in E$ from $L$ to $R$ -- which minimizes the total weight of vertices in $S$.

The Global Minimum Edge-- and Vertex-Cut problems have been extensively studied over many decades, resulting in almost-linear time algorithms for several restricted settings and edge-density regimes.
%The Global Minimum Edge-- and Vertex-Cut problems have seen extensive study over many decades, resulting in almost-linear time algorithms for several restricted settings and edge-density regimes.
For example, in the \emph{undirected} setting, the Global Minimum Edge-Cut problem has been known to admit near-linear time exact algorithms for over 25 years \cite{Kar00}, with a more recent \emph{deterministic} almost-linear time algorithm being achieved by \cite{LP20}, using the breakthrough Max-Flow result of \cite{vdBCK23}.
For the Global Minimum \emph{Vertex}-Cut problem in undirected graphs, Li et al. \cite{LNP25} recently introduced the first almost-linear time exact algorithm in the special case where all vertex-weights are unit.
In the most general setting of directed graphs with arbitrary (positive integral) weights, the $(1+\epsilon)$-approximation algorithms of Cen et al. \cite{CLN22} achieve time\footnote{In this section, when stating the running times of algorithms, we omit factors that are polylogarithmic in the maximum edge-weight/vertex-weight.} $\Tilde{O}(n^2/\epsilon^2)$ for both Minimum Edge-Cut and Minimum Vertex-Cut, which is near-linear when the input graph is dense and $\epsilon$ is large. Furthermore, if all edge-weights/vertex-weights are unit, then an exact Minimum Edge-Cut can be found in time $\Tilde{O}(n^2)$ \cite{CQ21}, and an exact Minimum Vertex-Cut can be found in time $O\left(n^{2+o(1)}\right)$ \cite{LNP25}, with the latter result having recently been simplified by \cite{FJMY25}.
There have also been several algorithms for the unit-weight versions of these problems whose running times are parameterized by the value $k$ of the optimal solution, and that run in almost-linear time when this value is sufficiently small; namely, the well-known directed Minimum Edge-Cut algorithm of Gabow \cite{Gab91} runs in time $\Tilde{O}(m \cdot k)$, while a time of $O\left(m^{1+o(1)} \cdot k\right)$ was recently achieved by both the directed Minimum Vertex-Cut algorithm of Chuzhoy et al. \cite{CMT25}, and the \emph{deterministic} undirected Minimum Vertex-Cut algorithm of Jiang et al. \cite{JNSY25}.
However, despite all of this progress, when the optimal solution value is unrestricted and the weights are arbitrary positive integers, there are no known constant-approximate algorithms for the directed version of either problem which run in $o(\min\{n^2, m \sqrt{n}\})$ time, for any edge-density regime.
In fact, the Global Minimum \emph{Vertex}-Cut problem is not even known to admit $o(n^2)$-time constant-approximate algorithms for sparse graphs.

In this paper, we present almost-linear time constant-approximate randomized algorithms for both the Global Minimum Edge-Cut and the Global Minimum Vertex-Cut problems, in the most general setting of directed graphs with arbitrary (positive integral) weights.
More generally, given any $\epsilon > 0$, our algorithms can compute a $(1+\epsilon)$-approximate solution in time $O\left(m^{1+o(1)} / \epsilon\right)$, directly improving over the $\Tilde{O}\left(n^2 / \epsilon^2\right)$ running time of \cite{CLN22} in cases where the input graph is not too dense, and in cases where $\epsilon$ is small.
Furthermore, by choosing $\epsilon = 1/k$, our algorithms can produce exact solutions in time that matches (up to subpolynomial factors) the aforementioned unit-weight algorithms of Gabow \cite{Gab91} and Chuzhoy et al. \cite{CMT25}, while not requiring the unit-weight assumption.
We also note that for the Global Minimum Vertex-Cut problem with arbitrary positive integral vertex-weights, our algorithm beats the best known constant-approximate algorithms even in the undirected setting.
See \Cref{table:known-running-times} for a more comprehensive summary of the state-of-the-art running times of $(1+\epsilon)$-approximation (and exact) algorithms for the Global Minimum Edge-Cut and Vertex-Cut problems. (As far as we are aware, all state-of-the-art approximation algorithms for these problems provide a $(1+\epsilon)$ approximation guarantee.)

In terms of techniques, we build upon and extend a framework that was recently introduced by Chuzhoy et al. \cite{CMT25} for their $O\left(m^{1+o(1)} \cdot k\right)$-time unit-weight Minimum Vertex-Cut algorithm. We discuss this framework and our new contributions to it in \Cref{sec:overview-of-rooted-minimum-cut}.
Using this framework, we are able to design algorithms for the \emph{Rooted variants} of the Minimum Edge-Cut and Minimum Vertex-Cut problems: in the Rooted variant of the Minimum Edge-Cut (resp. Vertex-Cut) problem, the input additionally consists of a "root" vertex $r$, and we must search for an edge-cut $(X,Y)$ (resp. vertex-cut $(L,S,R)$) of minimum weight \emph{only among those cuts that satisfy $r \in Y$ (resp. $r \in R$)}.\footnote{We note that the Rooted Minimum Vertex-Cut problem is usually defined with the constraint that the root must lie in $L$, rather than in $R$. However, we use the formulation with $r \in R$ in this paper so as to be consistent with the aforementioned $O\left(m^{1+o(1)} \cdot k\right)$-time algorithm of Chuzhoy et al., which implicitly uses a reduction to this version of the problem. Similarly, the Rooted Minimum Edge-Cut problem is usually defined with the constraint that $r$ must lie in $X$, rather than in $Y$, but we use the formulation with $r \in Y$ in order to be consistent between the edge-cut and vertex-cut settings.}
In the case of the Minimum \emph{Edge}-Cut problem, there is a simple well-known reduction from its Global variant to its Rooted variant, and we use this reduction to achieve our result for the Global variant of this problem.
However, in the case of the Minimum \emph{Vertex}-Cut problem, translating between these two variants is not as straightforward: while several reductions from Global Minimum Vertex-Cut to Rooted Minimum Vertex-Cut exist in the literature, they all either are limited to the unit-weight setting (such as the one implicitly used by the aforementioned algorithm of Chuzhoy et al.), or are algorithm-specific (such as the one used by Chekuri and Quanrud \cite{CQ21}).\footnote{There are also many existing reductions from the Global Minimum Vertex-Cut problem to the $s$-$t$ Minimum Vertex-Cut problem that can easily be modified in order to instead reduce to the Rooted Minimum Vertex-Cut problem. However, all of these either suffer from one of the aforementioned drawbacks, or can increase the total running time by an $m^{\Omega(1)}$ multiplicative factor.}
Therefore, in order to achieve our result for the \emph{Global} Minimum Vertex-Cut problem, we introduce a novel black-box reduction from the Global Minimum Vertex-Cut problem to the Rooted Minimum Vertex-Cut problem, that works in the arbitrary-weight setting on directed graphs.
We hope that this reduction will be useful in the design of future algorithms for the Global Minimum Vertex-Cut problem, in both the approximate and non-approximate settings.

\paragraph{Independent work.}
Concurrently and independently of our work, \cite{Quan25} obtained randomized $(1+\epsilon)$-approximate algorithms for the global and rooted versions of Minimum Edge-Cut and Minimum Vertex-Cut in directed weighted graphs, and \cite{Jiang25} obtained a randomized $O(\log^5 n)$-approximate algorithm for the Rooted Minimum Edge-Cut problem in weighted directed graphs. The Minimum Edge-Cut algorithms of \cite{Quan25} run in the time of $O(\log^4(n)/\epsilon)$ single-commodity Max-Flow calls, and their Minimum Vertex-Cut algorithms run in the time of $O(\log^5(n)/\epsilon)$ single-commodity Max-Flow calls; by \cite{vdBCK23}, these running times are at most $O\left(m^{1+o(1)}/\epsilon\right)$, matching our algorithms for these problems.
The algorithm of \cite{Jiang25} runs in time $O(m^{1+o(1)})$, matching our running time for the case of $\epsilon=\Theta(1)$, albeit with a somewhat worse approximation factor.

\begin{table}[t]
	\centering
	\small
	\begin{threeparttable}
		\caption{Best previously-known randomized running times for computing $(1+\epsilon)$-approximate (and exact) Global Minimum Cut in different settings.}
		\label{table:known-running-times}
		\begin{tabularx}{\linewidth}{l X Y Y}
			\toprule
			\multicolumn{2}{c}{} &
			\multicolumn{2}{c}{\textbf{Problem}} \\
			\cmidrule(lr){3-4}
			\textbf{ } & \textbf{Weights} &
			\textbf{Minimum Edge-Cut} &
			\textbf{Minimum Vertex-Cut} \\
			\midrule
			\multirow{2}{*}{Directed}
			& Arbitrary Weights (the total weight is denoted by $W$)
			& \makecell[tl]{%
%				\res{$O(\min\{n m^{2/3+o(1)}, m^{1+o(1)} \sqrt{n}\})$}{CLN22}\\
				\res{$O(n m^{2/3+o(1)})$}{CLN22}\\
				\res{$O(m^{1+o(1)} \sqrt{n})$}{CLN22}\\
				\res{$\Tilde{O}(n^2/\epsilon^2)$}{CLN22}\\
				\textcolor{red}{$O(m^{1+o(1)} / \epsilon)$\, this paper}
			}
			& \makecell[tl]{%
				\res{$O(\min\{m n^{0.976}, n^{2.677}\})$}{CMT25}\\
				\res{$\Tilde{O}(n^2/\epsilon^2)$}{CLN22}\\
				\res{$\Tilde{O}(nW / \epsilon)$}{CQ21}\\
				\textcolor{red}{$O(m^{1+o(1)} / \epsilon)$\, this paper}
			} \\
			\addlinespace
			& Unit Weights
			& \makecell[tl]{%
				\res{$\Tilde{O}(n^2)$}{CQ21}\\
				\res{$\Tilde{O}(m \cdot k)$}{Gab91}\\
				\textcolor{red}{$O(m^{1+o(1)} / \epsilon)$\, this paper}
			}
			& \makecell[tl]{%
				\res{$O(n^{2+o(1)})$}{LNP25,FJMY25}\\
				\res{$O(m^{1+o(1)} \cdot k)$}{CMT25}\\
				\res{$\Tilde{O}(\min\{m k^2, m k / \epsilon\})$}{NSY19b}\\
				\res{$\Tilde{O}(n^2/\epsilon)$}{CQ21}\\
				\textcolor{red}{$O(m^{1+o(1)} / \epsilon)$\, this paper}
			} \\
			\Xhline{2\arrayrulewidth}
			\addlinespace
			\multirow{2}{*}{Undirected}
			& Arbitrary Weights (the total weight is denoted by $W$)
			& \makecell[tl]{%
				\res{$\tilde{O}(m)$}{Kar00}
			}
			& \makecell[tl]{%
				\res{$O(\min\{m^{3/2+o(1)}, n^{2.677}\})$}{CT25,CMT25}\\
				\res{$\Tilde{O}(n^2/\epsilon^2)$}{CLN22}\\
				\res{$\Tilde{O}(nW / \epsilon)$}{CQ21}\\
				\textcolor{red}{$O(m^{1+o(1)} / \epsilon)$\, this paper}
			} \\
			\addlinespace
			& Unit Weights
			& \makecell[tl]{%
				\res{$\tilde{O}(m)$}{Kar00}
			}
			& \makecell[tl]{%
				\res{$O(m^{1+o(1)})$}{LNP25}\\
				\res{$\Tilde{O}(\min\{m + n k^3, m \cdot k / \epsilon\})$}{NSY19b}\\
				\res{$\Tilde{O}(n^2/\epsilon)$}{CQ21}
			} \\
			\bottomrule
		\end{tabularx}
		\begin{tablenotes}[flushleft]
			\footnotesize
			\item This table shows the state-of-the-art randomized running times for achieving $(1+\epsilon)$ approximation for the Global Minimum Edge-Cut and Global Minimum Vertex-Cut problems on an $n$-vertex $m$-edge graph, where $k$ denotes the value of the optimal solution, and $W$ denotes the sum of weights. Our new algorithms are listed in those settings where they provide an improvement upon the previous state-of-the-art for some choice of parameters $n,m,k,\epsilon$.
			\item Running times without dependence on $\epsilon$ are for exact algorithms. The $\Tilde{O}$ notation is used to hide $\polylog(m)$ factors. In all algorithms for the arbitrary-weight setting, the running time hides an additional $\polylog W$ multiplicative factor.
			\item We note that in the unit-weight setting, the value $k$ of the optimal solution can be no larger than the average degree $2m/n$. This means, for example, that in the directed unit-weight setting for Minimum Edge-Cut, it is possible to achieve time $\Tilde{O}(\min\{n^2, m \cdot k\}) \leq \Tilde{O}(m \sqrt{n})$ by combining the results of \cite{CQ21} and \cite{Gab91}.
			\item A survey of deterministic algorithms for these problems can be found, for example, in Jiang et al. \cite{JNSY25}.
		\end{tablenotes}
	\end{threeparttable}
\end{table}

\subsection{Our Results}\label{sec:our-results}

The following two theorems summarize our main results for the Global Minimum Edge-Cut and Global Minimum Vertex-Cut problems.

\begin{theorem}\label{thm : main-thm-edge-version}
	There is a randomized algorithm that, given a directed $m$-edge graph $G$ with positive integral weights $w(e) \leq W$ on the edges $e \in E(G)$, as well as a precision parameter $\epsilon \in (0,1)$, outputs an edge-cut $(X,Y)$ in $G$ such that, with probability at least $\frac{1}{2}$, $w(E_G(X,Y)) \leq (1+\epsilon)\OPT_{E}$ holds, where $\OPT_{E}$ denotes the weight of the global minimum edge-cut in $G$.
	The running time of the algorithm is $O\left(\frac{m^{1+o(1)} \cdot \log W}{\epsilon}\right)$.
\end{theorem}

\begin{theorem}\label{thm : main-thm-vertex-version}
	There is a randomized algorithm whose input is a directed $m$-edge graph $G$ with positive integral weights $w(v) \leq W$ on the vertices $v \in V(G)$, such that $G$ contains some vertex-cut, as well as a precision parameter $\epsilon \in (0,1)$. The algorithm outputs a vertex-cut $(L,S,R)$ in $G$ such that, with probability at least $\frac{1}{2}$, $w(S) \leq (1+\epsilon)\OPT_{V}$ holds, where $\OPT_{V}$ denotes the weight of the global minimum vertex-cut in $G$.
	The running time of the algorithm is $O\left(\frac{m^{1+o(1)} \cdot \log W}{\epsilon}\right)$.
\end{theorem}

\noindent We comment that the additional condition in \Cref{thm : main-thm-vertex-version} that the input graph $G$ "has some vertex-cut" is equivalent to requiring that $G$ is not a complete graph; furthermore, this condition can be verified in $O(m)$ time.
We also comment that the running times of the algorithms from \Cref{thm : main-thm-edge-version} and \Cref{thm : main-thm-vertex-version} are dominated by calls to the almost-linear time exact $s$-$t$ Maximum Flow procedure of \cite{vdBCK23}, while all other parts of these algorithms can be implemented in $\Tilde{O}\left(\frac{m \cdot \log W}{\epsilon}\right)$ time.
This means that if a near-linear time maximum $s$-$t$ flow algorithm will be discovered in the future, then plugging it into our algorithms would allow Theorems \ref{thm : main-thm-edge-version} and \ref{thm : main-thm-vertex-version} to be implemented in time $\Tilde{O}\left(\frac{m \cdot \log W}{\epsilon}\right)$.

As mentioned earlier, we achieve \Cref{thm : main-thm-vertex-version} by using a new reduction from the Global Minimum Vertex-Cut problem to the Rooted Minimum Vertex-Cut problem, that works in the vertex-weighted directed setting.
This reduction is our other main result.
%Our other main result is a black-box reduction from the Global Minimum Vertex-Cut problem to the Rooted Minimum Vertex-Cut problem, that works in the vertex-weighted directed setting.
More formally, this reduction shows that, in order to solve the Global Minimum Vertex-Cut problem on an $n$-vertex $m$-edge graph, it suffices to solve the Rooted Minimum Vertex-Cut problem on $O\left(\frac{m}{n}\right)$ different $n$-vertex graphs, that together have $O(m)$ edges overall.
Furthermore, if an approximate algorithm is used to solve the Rooted Minimum Vertex-Cut instances, then the reduction obtains an approximate solution to the Global Minimum Vertex-Cut problem, with the same approximation factor.
We formulate this reduction as an oracle algorithm, with access to an oracle for the Rooted Minimum Vertex-Cut problem; so we begin by defining this oracle.
As can be seen by the following definition and theorem, all the instances of the Rooted Minimum Vertex-Cut problem that our reduction must solve have the same vertex-weights as the input instance of the Global Minimum Vertex-Cut problem.

\begin{definition}[Rooted Minimum Vertex-Cut Oracle]\label{def:oracle}
	Consider any set $V$ of vertices, any weight function $w$ assigning positive integral weights $w(v) \leq W$ to the vertices $v \in V$, and any precision parameter $\epsilon \geq 0$.
	A \emph{$(1+\epsilon)$-Approximate Rooted Minimum Vertex-Cut oracle on $V$ with weight function $w$} is an oracle that implements the following type of query:
	given any directed graph $G'$ with $V(G')=V$, and any vertex $y^* \in V$, such that there exists some vertex-cut $(L',S',R')$ in $G'$ with $y^* \in R$, the oracle returns a vertex-cut $(L,S,R)$ in $G'$ such that $y^* \in R$ and such that, with probability at least $\frac{1}{2}$, the inequality $w(S) \leq (1+\epsilon)\OPT_{G',y^*}$ holds, where $\OPT_{G',y^*}$ denotes the minimum possible weight of a vertex-cut $(L',S',R')$ in $G'$ with $y^* \in R'$.
\end{definition}

\begin{theorem}[Reduction from Global to Rooted Minimum Vertex-Cut]\label{thm:main-reduction}
	There is a randomized algorithm whose input consists of a directed $m$-edge $n$-vertex graph $G$ with positive integral weights $w(v) \leq W$ on the vertices $v \in V(G)$, such that $G$ has some vertex-cut, as well as a precision parameter $\epsilon \geq 0$.
	The algorithm is further given access to a $(1+\epsilon)$-Approximate Rooted Minimum Vertex-Cut oracle on $V(G)$ with the aforementioned vertex-weights.
	The output of the algorithm is a vertex-cut $(L,S,R)$ in $G$ such that, with probability at least $\frac{1}{2}$, $w(S) \leq (1+\epsilon)\OPT_{G}$ holds, where $\OPT_{G}$ denotes the weight of the global minimum vertex-cut in $G$.
	The algorithm makes at most $O\left(m/n\right)$ queries to the oracle, and the overall number of edges in the graphs from these queries is $O(m)$.
	The algorithm spends at most $\Tilde{O}(m \log W)$ time outside of oracle queries.
\end{theorem}

We emphasize that, by setting $\epsilon=0$, the above algorithm can be used to reduce the \emph{exact} Global Minimum Vertex-Cut problem to the \emph{exact} Rooted Minimum Vertex-Cut problem. The above theorem also allows values of $\epsilon$ that are larger than $1$, as well as values that depend on the size of the graph $G$.

Lastly, in the next two theorems, we summarize the guarantees of our algorithms for the rooted variants of the Minimum Edge-Cut and Minimum Vertex-Cut problems.
These are essentially identical to Theorems \ref{thm : main-thm-edge-version} and \ref{thm : main-thm-vertex-version}, except that they deal with the rooted variants of the Minimum Edge-Cut and Vertex-Cut problems, rather than with the global variants of these problems.
%Observe that the following \Cref{thm : main-thm-rooted-vertex-version} can be used directly to implement the Rooted Minimum Vertex-Cut oracle for \Cref{thm:main-reduction}; by doing so, we obtain an algorithm for \Cref{}
%
%\Cref{thm : main-thm-vertex-version} is achieved directly by using \Cref{thm:main-reduction}, with the oracle implemented by the following \Cref{thm : main-thm-rooted-vertex-version}

\begin{theorem}\label{thm : main-thm-rooted-edge-version}
	There is a randomized algorithm whose input consists of a directed $m$-edge graph $G$ with positive integral weights $w(e) \leq W$ on the edges $e \in E(G)$, as well as a precision parameter $\epsilon \in (0,1)$, and a root vertex $y^* \in V(G)$. The output of the algorithm is an edge-cut $(X,Y)$ in $G$ such that $y^* \in Y$ and, with probability at least $\frac{1}{2}$, $w(E_G(X,Y)) \leq (1+\epsilon)\OPT_{E,y^*}$ holds, where $\OPT_{E,y^*}$ denotes the minimum possible weight of an edge-cut $(X',Y')$ in $G$ with $y^* \in Y'$.
	The running time of the algorithm is $O\left(\frac{m^{1+o(1)} \cdot \log W}{\epsilon}\right)$.
\end{theorem}

\begin{theorem}\label{thm : main-thm-rooted-vertex-version}
	There is a randomized algorithm whose input consists of a directed $m$-edge graph $G$ with positive integral weights $w(v) \leq W$ on the vertices $v \in V(G)$, as well as a precision parameter $\epsilon \in (0,1)$, and a root vertex $y^* \in V(G)$ such that there exists some vertex-cut $(L',S',R')$ in $G$ with $y^* \in R'$. The output of the algorithm is a vertex-cut $(L,S,R)$ in $G$ such that $y^* \in R$ and, with probability at least $\frac{1}{2}$, $w(S) \leq (1+\epsilon)\OPT_{V,y^*}$ holds, where $\OPT_{V,y^*}$ denotes the minimum possible weight of a vertex-cut $(L',S',R')$ in $G$ with $y^* \in R'$.
	The running time of the algorithm is $O\left(\frac{m^{1+o(1)} \cdot \log W}{\epsilon}\right)$.
\end{theorem}

We also comment that, even though all of our theorems are stated and proved for the setting in which edge-weights and vertex-weights are not allowed to be zero, it is possible to convert our algorithms to the setting in which zero-weight edges and vertices are allowed, via a standard reduction. For the sake of completeness, we describe this reduction in \Cref{sec:zero-weights}.

%\rnote{The following paragraph probably goes at the end of the results section.}
%
%
%\paragraph{Dependence on $W$.}
%We note that we intentionally did not optimize the dependence on $W$ in the running times of \Cref{thm : main-thm-edge-version} and \Cref{thm : main-thm-vertex-version}, in order to keep both the statements and proofs of these theorems as simple as we can.
%However, in \Cref{sec: low-dependence-on-W}, we show that the dependence on $W$ in these running times can be made as small as $O(\log \log W)$, provided that the input representation of each weight $w$ includes both its binary representation in $O(\log W)$ bits, and an $O(\log\log W)$-bit representation of the value $\ceil{\log w}$.
%Since this special representation of the weights can be computed in time $O(m \cdot \log W)$ for the Global Minimum Edge-Cut problem, and in time $O(n \cdot \log W)$ for the Global Minimum Vertex-Cut problem, it follows that the running time in \Cref{thm : main-thm-edge-version} can be made as small as $O\left(\frac{m^{1+o(1)} \cdot \log \log W}{\epsilon} + m \log W\right)$, and that the running time in \Cref{thm : main-thm-vertex-version} can be made as small as $O\left(\frac{m^{1+o(1)} \cdot \log \log W}{\epsilon} + n \log W\right)$.

\paragraph{Completing the proofs of Theorems \ref{thm : main-thm-edge-version} and \ref{thm : main-thm-vertex-version}}
We complete this subsection by demonstrating how to prove our Global Minimum Cut theorems (Theorems \ref{thm : main-thm-edge-version} and \ref{thm : main-thm-vertex-version}) using our Rooted Minimum Cut theorems (Theorems \ref{thm : main-thm-rooted-edge-version} and \ref{thm : main-thm-rooted-vertex-version}) and our reduction theorem (\Cref{thm:main-reduction}).

To prove \Cref{thm : main-thm-vertex-version}, we simply use the reduction from \Cref{thm:main-reduction}, and, whenever it makes a query to the $(1+\epsilon)$-Approximate Rooted Minimum Vertex-Cut oracle, we answer this query using a call to \Cref{thm : main-thm-rooted-vertex-version}.
It is not hard to verify that \Cref{thm : main-thm-rooted-vertex-version} satisfies the guarantees required by the oracle, and that this reduction results in a Global Minimum Vertex-Cut algorithm exactly as specified by \Cref{thm : main-thm-vertex-version}.

The proof of our Global Minimum Edge-Cut theorem (\Cref{thm : main-thm-edge-version}) follows from our Rooted Minimum Edge-Cut theorem via standard techniques. Specifically, given an input to \Cref{thm : main-thm-edge-version} consisting of a graph $G$ with weights $w(e)$ on the edges $e \in E(G)$, and a precision parameter $\epsilon \in (0,1)$, we first compute the "reversed graph" $\overline{G}$ that is obtained from $G$ by reversing the direction of all edges; we then fix any vertex $r \in V(G)$ as the root, and make two calls to the Rooted Minimum Edge-Cut algorithm of \Cref{thm : main-thm-rooted-edge-version} -- one call is on the graph $G$ with root $r$, while the other call is on the reversed graph $\overline{G}$ with the same root $r$; both calls use the precision parameter $\epsilon$ from the input. Let $(X,Y)$ denote the edge-cut in $G$ obtained by the first call, and let $(\overline{X}',\overline{Y}')$ denote the edge-cut in $\overline{G}$ obtained by the second call. Define $X' = \overline{Y}'$ and $Y' = \overline{X}'$ -- it is not hard to verify that $(X',Y')$ is an edge-cut in $G$ whose value $w(E_G(X',Y'))$ is the same as the value of the edge-cut $(\overline{X}',\overline{Y}')$ in the reverse graph $\overline{G}$, namely, the same as $w(E_{\overline{G}}(\overline{X}',\overline{Y}'))$.
Lastly, we return the edge-cut whose value in $G$ is smallest among the two edge-cuts $(X,Y)$ and $(X',Y')$.
Therefore, to prove the correctness of the algorithm, it must be shown that with probability at least $\frac{1}{2}$, at least one of the two edge-cuts $(X,Y)$ and $(X',Y')$ has value at most $(1+\epsilon)\OPT_{E}$, where $\OPT_{E}$ denotes the weight of the global minimum edge-cut in $G$. For the analysis, consider any global minimum edge-cut $(X^*,Y^*)$ in $G$. The analysis can then be divided into two cases based on whether $r \in Y^*$ or $r \in X^*$: in the case where $r \in Y^*$, it follows from the guarantees of our Rooted Minimum Edge-Cut algorithm that the weight of the edge-cut $(X,Y)$ satisfies $w(E_G(X,Y)) \leq (1+\epsilon)\OPT_{E}$ with probability at least $\frac{1}{2}$, as needed; in the complementary case where $r \in X^*$, it can be verified using a similar argument that the edge-cut $(X',Y')$ satisfies $w(E_G(X',Y')) = w(E_{\overline{G}}(\overline{X}',\overline{Y}')) \leq (1+\epsilon)\OPT_{E}$ with probability at least $\frac{1}{2}$ -- namely, this holds because the edge-cut $(Y^*,X^*)$ in the reversed graph $\overline{G}$ is a solution of value $\OPT_{E}$ to the Rooted Minimum Edge-Cut instance on $\overline{G}$ with root $r$. This concludes the proof of \Cref{thm : main-thm-edge-version} using \Cref{thm : main-thm-rooted-edge-version}.

\subsection{Technical Overview}

\subsubsection{Rooted Minimum Cut}\label{sec:overview-of-rooted-minimum-cut}

We now give an overview of the techniques used in our Rooted Minimum Edge-Cut and Rooted Minimum Vertex-Cut algorithms (Theorems \ref{thm : main-thm-rooted-edge-version} and \ref{thm : main-thm-rooted-vertex-version}).
Our algorithms are inspired by techniques from the $O\left(m^{1+o(1)} \cdot k\right)$-time exact algorithm of \cite{CMT25} for the unit-weight Global Minimum Vertex-Cut problem, so we begin with a discussion of some ideas used by that algorithm. Then, we discuss our algorithms and how they relate to these ideas.
We note that, while \cite{CMT25} solve the Global Minimum Vertex-Cut problem rather than its rooted variant, their algorithm implicitly begins by reducing the problem to the rooted variant, using a reduction that is tailored to the unit-weight setting.
% -- in fact, the remainder of their algorithm can easily be made to work without the unit-weight assumption.
We now describe the setup for the techniques used by \cite{CMT25} for solving the resulting Rooted Minimum Vertex-Cut instance.

Consider an input instance of the Rooted Minimum Vertex-Cut problem, consisting of a graph $G$ with weights $w(v)$ on the vertices $v \in V(G)$, and a root vertex $y^*$. Let $\OPT$ denote the value of the optimal solution to this instance.
Observe that, in order to (exactly) solve this problem, it suffices to find some vertex $v$ for which the weighted $v$-$y^*$ vertex-connectivity, $\kappa_G(v,y^*)$, is equal to $\OPT$ -- once such a vertex is found, an optimal solution to the problem can be extracted using a single call to an $s$-$t$ Minimum Cut subroutine.
%
%
%The idea behind the algorithm of \cite{CMT25} is to first define the notion of a \emph{sparsifier} graph $G_x$ for a vertex $x$, in such a manner that allows the existence of small sparsifiers, while guaranteeing that a single $s$-$t$ max-flow invocation on a the sparsifier $G_x$ can be used to determine whether $x$ is on the side $L$ of a good cut.
%
%Suppose we sampled some set $T$ of \emph{terminals} such that there exists some $x^* \in L \cap T$.
%
A common paradigm for finding such a vertex $v$ is to first sample some set $T$ of \emph{terminals} so that, with high probability, there exists some terminal $x^* \in T$ satisfying $\kappa_G(x^*,y^*)=\OPT$; then, for every terminal $x \in T$, construct a \emph{sparsifier graph} $G_x$, so that a single call to an $s$-$t$ Max-Flow subroutine on the sparsifier $G_x$ suffices to distinguish between the case where $x=x^*$ and the case where $\kappa_G(x,y^*)>\OPT$.
Thus, as long as the total size $\sum_{x \in T}|E(G_x)|$ of the sparsifiers is sufficiently small, it is possible to efficiently find a terminal $x \in T$ with $\kappa_G(x,y^*) = \OPT$, by calling an $s$-$t$ Max-Flow subroutine once on each sparsifier.
The main challenge is thus to efficiently construct such sparsifiers while guaranteeing that $\sum_{x \in T}|E(G_x)|$ is small.

The work of \cite{CMT25} introduced a new technique for constructing such sparsifiers, which we describe next.
At a high level, the technique of \cite{CMT25} shows that in order to construct the desired collection of sparsifiers, it suffices to have an efficient subroutine, which we refer to as \exactsparsify, that solves the following problem:
%given a set $B'$ of terminals, a subset $B \subseteq B'$, and a graph $G_{B'}$ that is simultaneously a sparsifier for every terminal $x \in B'$, construct a small graph $G_B$ that is simultaneuously a sparsifier for every terminal $x \in B$, so that the size of $G_B$ is proportional to $|B|$.
given a set $B'$ of terminals, a subset $B \subseteq B'$, and a graph $G_{B'}$ that simultaneously serves as a sparsifier for the whole set $B'$ of terminals (in a formalized sense, that we choose to omit from this high-level overview), construct a small graph $G_B$ that serves as a sparsifier for the set $B$ of terminals, so that the size of $G_B$ is proportional to $|B|$.
Furthermore, the running time of this subroutine must be roughly equal to the size of the graph $G_{B'}$, which may be smaller than $G$.

More concretely, the algorithm of \cite{CMT25} employs a divide-and-conquer approach consisting of $z \approx \log T$ phases: in Phase $i$, the algorithm considers a partition $\bset_i$ of the set $T$ of terminals, and constructs a sparsifier $G_B$ for each part $B \in \bset_i$ of this partition.
(We refer to the parts $B \in \bset_i$ of this partition as \emph{batches}.)
Furthermore, for each $i>1$, the partition $\bset_i$ is a refinement of the partition $\bset_{i-1}$ that was considered in the previous phase, in the sense that every batch $B \in \bset_i$ has a \emph{parent batch} $B' \in \bset_{i-1}$ which satisfies $B \subseteq B'$.
Thus, the sparsifier $G_B$ for each batch $B \in \bset_i$ can be computed by using the \exactsparsify subroutine on the sparsifier $G_{B'}$ of its parent batch.
Then, for each $i' \in [z]$, the total size $\sum_{B \in \bset_{i'}}|E(G_B)|$ of the sparsifiers constructed in Phase $i$ is proportional to $\sum_{B \in \bset_{i'}}|B|=|T|$; and, for every $i > 1$, by guaranteeing that each batch $B' \in \bset_{i-1}$ serves as the parent batch of at most $2$ batches of $\bset_i$, it follows that the total running time of Phase $i$ is roughly $\sum_{B' \in \bset_{i-1}}|E(G_{B'})|$.
%Lastly, the final sparsifier $G_x$ for each terminal $x$ is the one constructed for its batch in Phase $z$, where the partition $\bset_z$ consists only of constant-sized batches. Recall that the goal is producing sparsifiers such that $\sum_{x \in T} |E(G_x)|$ is small, and observe that, for every batch $B \in \bset_z$, the size of the sparsifier $G_B$ is counted exactly $|B|$ times in the expression $\sum_{x \in T} |E(G_x)|$ -- thus, the small batch-size in the partition $\bset_z$ allows for upper-bounding $\sum_{x \in T} |E(G_x)|$ by $O\left(\sum_{B \in \bset_z}|E(G_B)|\right)$.
Lastly, the final sparsifier $G_x$ for each terminal $x$ is the one constructed for its batch in Phase $z$, where the partition $\bset_z$ consists only of constant-sized batches -- it is then possible to use the fact that $\bset_z$ has small batches in order to upperbound the expression that we wanted to minimize, namely, $\sum_{x \in T} |E(G_x)|$. Specifically, this expression is upper-bounded as $\sum_{x \in T} |E(G_x)| \leq O\left(\sum_{B \in \bset_z}|E(G_B)|\right)$, because the size $|E(G_B)|$ of each sparsifier $\set{G_B}_{B \in \bset_z}$ is counted only $O(1)$ times in the left-hand side of this inequality.
The key remaining part is the implementation of the \exactsparsify subroutine.

%We now give some more details regarding the sparsifiers used by \cite{CMT25}, which help them solve the \exactsparsify problem. Assume from now on that there indeed exists some $x^* \in T$ with $\kappa_G(x^*,y^*)=\OPT$, and fix one such vertex $x^*$.
%It turns out that in order to construct a sparsifier $G_{B'}$, it suffices to find a low out-volume set $A_{B'} \subseteq V(G)$ of vertices such that, if $x^* \in B'$, then there exists an $x^*$-$y^*$ vertex-cut $(L,S,R)$ in $G$ with $L \subseteq A_{B'}$, whose weight is $w(S)=\OPT_{V,y^*}$.
%For reasons that are described next, the algorithm of \cite{CMT25} also makes sure to choose the set $A_{B'}$ such that, if $x^* \in B'$, then for a carefully chosen parameter $\nu$, the aforementioned vertex-cut $(L,S,R)$ satisfies $\vol^+_G(L) \leq \nu$.
%Then, if $x^* \in B'$, the obtained sparsifier $G_{B'}$ satisfies that $x^*,y^* \in V(G_{B'}) \subseteq V(G)$ and that there exists an $x^*$-$y^*$ vertex-cut $(L',S',R')$ in $G_{B'}$ of weight $w(S')=\OPT$, which additionally satisfies that the out-volume of $L'$ \textbf{in $\mathbf{G}$} is $\vol^+_G(L') \leq \nu$ -- a fact that is useful for solving  the \exactsparsify problem.

The algorithm of \cite{CMT25} (implicitly) provides a \exactsparsify subroutine that runs in time\\$O\left(|E(G_{B'})|^{1+o(1)}\right)$ (in the unit-weight setting), and produces a sparsifier $G_B$ of size $|E(G_B)| = \Tilde{O}\left(\frac{|B|}{|T|} \cdot m \cdot \OPT\right)$. (To achieve this, the sparsifiers must satisfy some additional technical properties.)
Thus, the total size of the sparsifiers they construct in each phase is $\Tilde{O}(m \cdot \OPT)$, and the total running time of their algorithm is $O\left(m^{1+o(1)} \cdot \OPT\right)$.

In this paper, our main contribution is introducing a $(1+\epsilon)$-approximate version of the \exactsparsify subroutine, which produces a sparsifier $G_B$ of size $|E(G_B)| = \Tilde{O}\left(\frac{|B|}{|T|} \cdot m \cdot \frac{1}{\epsilon}\right)$, using different techniques than those of \cite{CMT25}. (See \Cref{sec:comparison-of-techniques-for-batch-sparsify} for a comparison of these techniques.) The running time of this subroutine is $O\left(|E(G_{B'})|^{1+o(1)} \cdot \log W\right)$, where $W$ denotes the maximum vertex-weight.
Furthermore, we extend the divide-and-conquer framework of \cite{CMT25} to the approximate setting, so that it can be used together with this approximate subroutine in order to approximately solve the Rooted Minimum Vertex-Cut problem. (We also slightly simplify this framework by omitting one technical notion that was used in \cite{CMT25}, which stated that certain batches become "invalid" throughout the execution of the algorithm.)
Thus, the total size of the sparsifiers that our algorithm constructs in each phase is $\Tilde{O}\left(\frac{m}{\epsilon}\right)$, and the total running time of our algorithm is $O\left(\frac{m^{1+o(1)} \cdot \log W}{\epsilon}\right)$, as promised in \Cref{thm : main-thm-rooted-vertex-version}.\footnote{Actually, to achieve the probability bound stated in \Cref{thm : main-thm-rooted-vertex-version}, we must amplify the success probability using $\polylog(mW)$ many repetitions, thus ending with a running time of $O\left(\frac{m^{1+o(1)} \cdot \polylog W}{\epsilon}\right)$. However, we then transform this algorithm into one that runs in time $O\left(\frac{m^{1+o(1)} \cdot \log W}{\epsilon}\right)$, by using a standard reduction that we discuss in \Cref{sec: low-dependence-on-W}.}
We furthermore observe that all of the algorithms mentioned above have straightforward analogues in the edge-cut setting, thus allowing us to obtain our result for the Rooted Minimum Edge-Cut problem (\Cref{thm : main-thm-rooted-edge-version}).

\subsubsection{Reducing Global Minimum Vertex-Cut to Rooted Minimum Vertex-Cut.}

Consider an instance of the Global Minimum Vertex-Cut problem, consisting of a directed $m$-edge graph $G$ with positive integral weights $w(v)$ on the vertices $v \in V(G)$.
Let $\OPT$ denote the weight of the global minimum vertex-cut on $G$.

A common paradigm for solving the Global Minimum Vertex-Cut problem, is to reduce it to several instances of the $s$-$t$ Minimum Vertex-Cut problem (\cite{LNP25,JNSY25}).
This kind of reduction consists of two steps.
In the first step, the algorithm constructs a collection $\{(s_i,t_i)\}_{i=1}^{\ell}$ of pairs of vertices from $G$, which naturally induces a collection of $\ell$ instances of the $s$-$t$ Minimum Vertex-Cut problem on the graph $G$, such that every feasible solution of any of these instances is also a feasible solution to the Global Minimum Vertex-Cut problem on $G$.
Furthermore, the algorithm constructs the collection $\{(s_i,t_i)\}_{i=1}^{\ell}$ such that with high probability, there exists some global minimum vertex-cut $(L,S,R)$ of $G$ and some $i \in [\ell]$ such that $s_i \in L$ and $r_i \in R$ holds; whenever this happens, it follows that the optimal solution value of the $i$'th instance of $s$-$t$ Minimum Vertex-Cut is exactly equal to $\OPT$. In that case, an optimal solution to the original problem can be obtained by solving these $\ell$ instances of $s$-$t$ Minimum Vertex-Cut, and selecting the minimum-weight vertex-cut among these solutions.
However, explicitly solving these instances is often prohibitively expensive, as their total size is $\ell \cdot m$.
Therefore, the second step of the paradigm consists of "sparsifying" each of these instances -- that is, for each $i \in [\ell]$, the algorithm finds a graph $G_i$ such that the $s_i$-$t_i$ minimum vertex-cut in $G_i$ has the same weight as the $s_i$-$t_i$ minimum vertex-cut in $G$, and any $s_i$-$t_i$ vertex-cut in $G_i$ can be easily transformed into an $s_i$-$t_i$ vertex-cut in $G$ of the same value. Then, these new sparsified instances of the $s$-$t$ Minimum Vertex-Cut problem can be solved instead of the instances on $G$.
By guaranteeing that the total size of the graphs $\{G_i\}_{i \in [\ell]}$ is sufficiently small, algorithms that use this method are able to solve all $s$-$t$ Minimum Vertex-Cut instances faster than the time it would take to explicitly solve them on the original graph $G$.

At a very high level, our reduction in \Cref{thm:main-reduction} follows a paradigm similar to the above, except we reduce to the Rooted Minimum Vertex-Cut problem instead of to the $s$-$t$ Minimum Vertex-Cut problem.
Specifically, in the first step of the reduction, we construct an appropriate set $T \subseteq V(G)$ of $\ell$ "root" vertices, which naturally induces a collection of $\ell$ instances of the Rooted Minimum Vertex-Cut problem on the graph $G$, such that any solution of these instances is also a feasible solution to the Global Minimum Vertex-Cut problem on $G$.
Furthermore, if there exists some global minimum vertex-cut $(L,S,R)$ of $G$ and some root $r \in T \cap R$, then one of these implicit instances of the Rooted Minimum Vertex-Cut problem has the same optimal solution value as that of the Global Minimum Vertex-Cut instance on $G$.
In the second step, we sparsify each of these instances by constructing, for each root $r \in T$, a sparsified graph $G_r$ with $V(G_r)=V(G)$, such that the set of feasible solutions of the Rooted Minimum Vertex-Cut problem on $G_r$ with root $r$ is exactly the same as that on $G$ with root $r$.
The key challenge is in making sure that the total size of all sparsified graphs $G_r$ is small, while at the same time selecting sufficiently many roots into the set $T$ so as to guarantee that with high probability, there exists some global minimum vertex-cut $(L,S,R)$ of $G$ with $T \cap R \neq \emptyset$.
Next, we explain our approach to overcoming this challenge.

For the sake of simplicity of this overview, we will make the additional assumption that some global minimum vertex-cut $(L,S,R)$ of $G$ has $w(R) \geq w(L)$ -- it is not difficult to get rid of this assumption, and we explain how to do this in \Cref{sec:removing-assumption-in-reduction:reduction}.
Our strategy for overcoming the aforementioned challenge is inspired by techniques from the $O\left(m^{1+o(1)}n\right)$-time deterministic Global Minimum Vertex-Cut algorithm of Jiang et al. \cite{JNSY25}.
The algorithm of \cite{JNSY25} follows the earlier-mentioned paradigm of reducing the Global Minimum Vertex-Cut problem to the $s$-$t$ Minimum Vertex-Cut problem;
furthermore, one of the key ideas in their algorithm is to construct the collection $\{(s_i,t_i)\}_{i=1}^{\ell}$ of source-sink pairs such that each vertex $u \in V(G)$ serves as the source $s_i$ in at most $\Tilde{O}\left(n \cdot\min\left\{\frac{w(u)}{w(V(G)) - \OPT},1\right\}\right)$ many pairs. Using this guarantee, as well as some clever analysis that is beyond the scope of this high-level overview, they are able to create sparsified graphs of total size $O(mn)$, which suffices for their $O\left(m^{1+o(1)}n\right)$-time deterministic algorithm.
At a high level, our key idea is to sample each vertex $u$ into our set $T$ of roots independently with probability roughly $\min\left\{\frac{w(u)}{w(V(G)) - \OPT},1\right\}$, which is $n$-fold smaller than the number of times that $u$ serves as the source vertex in the reduction of \cite{JNSY25}.
Then, by using a similar sparsification method to the one used by \cite{JNSY25}, we can guarantee that the expected total size of our sparsified graphs is $O(m)$ -- an $n$-fold improvement over the total size achieved by \cite{JNSY25}.
It is also not difficult to verify that this method of sampling ensures that for some global minimum vertex-cut $(L,S,R)$, we have a constant probability of sampling a root from the set $R$: indeed, observe that $w(L \cup R) = w(V(G)) - \OPT$ holds for every global minimum vertex-cut, and so, by our earlier assumption that some global minimum vertex-cut $(L,S,R)$ satisfies $w(R) \geq w(L)$, it follows that this cut satisfies $w(R) \geq \left(w(V(G)) - \OPT\right)/2$; it can then be verified that with probability $\Omega(1)$, at least one vertex of $R$ is sampled into $T$.
%By our earlier assumption that some global minimum vertex-cut $(L,S,R)$ of $G$ has $w(R) \geq w(L)$, it follows that we have a constant probability of sampling a root from $R$, as needed.
To summarize, our $n$-fold improvement over the reduction of \cite{JNSY25} comes both from leveraging the power of random sampling, which could not be used by the deterministic algorithm of \cite{JNSY25}, and, perhaps even more crucially, from the fact that we are reducing to the \emph{Rooted} Minimum Vertex-Cut problem rather than the $s$-$t$ Minimum Vertex-Cut problem -- this means that our sampling procedure only needs to guarantee that for some global minimum vertex-cut $(L,S,R)$, it samples a vertex $r \in R$, rather than having to sample a pair of vertices $(s,t)$ that simultaneously satisfies $s \in L$ and $t \in R$.

\subsection{Organization}
We cover preliminaries in \Cref{sec:prelims}.
In \Cref{sec:rooted-minimum-edge-cut}, we present our algorithm for the Rooted Minimum Edge-Cut problem -- we choose to present this algorithm before the vertex-cut version, as its analysis is simpler.
Furthermore, as these two algorithm are very similar, we defer the presentation of the Rooted Minimum \emph{Vertex}-Cut algorithm to \Cref{sec:rooted-minimum-vertex-cut}.
Finally, in \Cref{sec:main-reduction}, we present our reduction from Global Minimum Vertex-Cut to Rooted Minimum Vertex-Cut (\Cref{thm:main-reduction}).

In \Cref{sec: low-dependence-on-W} we describe a standard method, that is used in our algorithms for both the Rooted Minimum Edge-Cut and the Rooted Minimum Vertex-Cut problems, which allows transforming an algorithm with running time $O\left(\frac{m^{1+o(1)} \cdot \polylog W}{\epsilon}\right)$ for these problems into one that has running time $O\left(\frac{m^{1+o(1)} \cdot \log W}{\epsilon}\right)$.
\Cref{sec:missing-proofs-in-prelims} covers missing proofs from \Cref{sec:prelims}; \Cref{sec:missing-proofs-in-main-edge-version} covers a missing proof from \Cref{sec:rooted-minimum-edge-cut}; and \Cref{sec:removing-assumption-in-reduction:reduction} explains how to remove an additional assumption that we make in \Cref{sec:main-reduction}.
Lastly, \Cref{sec:zero-weights} explains how our main results can be extended to the setting in which weights are allowed to be $0$.

\section{Preliminaries}\label{sec:prelims}
All logarithms in this paper are base $2$.
All graphs considered in this paper are directed, and without self-loops. Furthermore, graphs are assumed to be without parallel edges unless explicitly specified otherwise.
We use the $\Tilde{O}$ and $\Tilde{\Omega}$ notations to hide $\polylog(m)$ factors.
For any natural number $n$, we use $[n]$ to denote the set $\{1,\ldots,n\}$.

\paragraph{Basic Graph-Theoretic Notation.}
Consider a directed graph $G$. We use $V(G)$ and $E(G)$, respectively, to denote its set of vertices and its set edges.
For every two sets $A,B \subseteq V(G)$ of vertices, we use $E_G(A,B) \subseteq E(G)$ to denote the set of edges that begin at $A$ and end at $B$.
For every vertex $v \in V(G)$, its set of \emph{outgoing edges} is $\delta^+_G(v) = E_G(\set{v},V(G))$, and its set of \emph{incoming edges} is $\delta^-_G(v) = E_G(V(G),\set{v})$.
The \emph{out-degree} of a vertex $v \in V(G)$ is $\deg^+_G(v) = |\delta^+_G(v)|$, and the \emph{in-degree} of a vertex $v \in V(G)$ is $\deg^-_G(v) = |\delta^-_G(v)|$.
The \emph{total degree} of a vertex $v \in V(G)$ is $\deg_G(v) = \deg^+_G(v) + \deg^-_G(v)$.
Furthermore, for every vertex $v \in V(G)$, its set of \emph{out-neighbors} is $N^+_G(v)=\set{u \in V(G) \mid (v,u) \in E(G)}$, and its set of \emph{in-neighbors} is $N^-_G(v)=\set{u \in V(G) \mid (u,v) \in E(G)}$.

For every set $A \subseteq V(G)$ of vertices, its set of \emph{outgoing boundary edges} is $\partial^+_G(A) = E_G(A,V(G) \setminus A)$, and its set of \emph{incoming boundary edges} is $\partial^-_G(A) = E_G(V(G) \setminus A,A)$.
Additionally, for a set $A \subseteq V(G)$ of vertices, its \emph{out-volume} is $\vol^+_G(A) = \sum_{v \in A} \deg^+_G(v) = |E_G(A,V(G))|$, and its \emph{in-volume} is $\vol^-_G(A) = \sum_{v \in A} \deg^-_G(v) = |E_G(V(G),A)|$.
Observe that $\vol^+_G(A)$ is not the same thing as $|\partial^+_G(A)|$.
The \emph{volume} of $A$ is $\vol_G(A) = \sum_{v \in A} \deg_G(v) = \vol^+_G(A) + \vol^-_G(A)$.
Lastly, for every set $A \subseteq V(G)$, its set of \emph{out-neighbors} is $N^+_G(A) = \left(\bigcup_{v \in A}N^+_G(v)\right) \setminus A$, and its set of \emph{in-neighbors} is $N^-_G(A) = \left(\bigcup_{v \in A}N^-_G(v)\right) \setminus A$.
When using the above notation, we often omit the subscript $G$ when it is obvious from the context.

\paragraph{Weighted Graphs.}
In this paper, we deal with two different types of weight-functions on graphs.
The first is a weight function $w$ that assigns weights $w(e)$ to the \emph{edges} $e \in E(G)$ of a graph $G$, and the second is a weight function $w$ that assigns weights $w(v)$ to the \emph{vertices} $v \in V(G)$.
When given a weight function $w$ of the former type, then, for every set $E' \subseteq E(G)$ of edges, we use the notation $w(E') = \sum_{e \in E'} w(e)$.
Similarly, when given a weight function $w$ of the latter type, then, for every set $V' \subseteq V(G)$ of vertices, we use the notation $w(V') = \sum_{v \in V'} w(v)$.

\paragraph{Edge-Cuts and Edge-Connectivity.}
Consider a graph $G$ with weights $w(e)$ on the edges $e \in E(G)$.
An edge-cut in $G$ is an (ordered) pair $(X,Y)$ such that $(X,Y)$ is a partition of the set $V(G)$, and $X,Y \neq \emptyset$. The weight of this cut, also known as its \emph{value}, is $w(E_G(X,Y))$.
For any pair $u,v \in V(G)$ of vertices, we say that an edge-cut $(X,Y)$ is \emph{a $u$-$v$ edge-cut} if $u \in X$ and $v \in Y$; if the weight of $(X,Y)$ is minimal among all $u$-$v$ edge-cuts, then we say that it is a \emph{minimum $u$-$v$ edge-cut}.
The (weighted) \emph{$u$-$v$ edge-connectivity}, denoted $\lambda_G(u,v)$, is the maximum possible cardinality of a (multi)set of $u$-$v$ paths in $G$ such that every edge $e \in E(G)$ belongs to at most $w(e)$ many of these paths.
By the Max Flow-Min Cut theorem, $\lambda_G(u,v)$ is also equal to the weight of the minimum $u$-$v$ edge-cut in $G$.

\paragraph{Vertex-Cuts and Vertex-Connectivity.}
Consider a graph $G$ with weights $w(v)$ on the vertices $v \in V(G)$.
A vertex-cut in $G$ is an (ordered) triple $(L,S,R)$ such that $(L,S,R)$ is a partition of the set $V(G)$, and $L,R \neq \emptyset$ (but $S$ may be empty). The weight of this cut, also known as its \emph{value}, is $w(S)$.
For any pair $u,v \in V(G)$ of vertices, we say that a vertex-cut $(L,S,R)$ is \emph{a $u$-$v$ vertex-cut} if $u \in L$ and $v \in R$; if the weight of $(L,S,R)$ is minimal among all $u$-$v$ vertex-cuts, then we say that it is a \emph{minimum $u$-$v$ vertex-cut}.
The (weighted) \emph{$u$-$v$ vertex-connectivity}, denoted $\kappa_G(u,v)$, is the maximum possible cardinality of a (multi)set of $u$-$v$ paths in $G$ such that every vertex $v' \in V(G) \setminus \{u,v\}$ belongs to at most $w(v')$ many of these paths.
By Menger's theorem, $\kappa_G(u,v)$ is also equal to the weight of the minimum $u$-$v$ vertex-cut in $G$.

\paragraph{Global Minimum Edge-Cut and Rooted Minimum Edge-Cut.}
In the Global Minimum Edge-Cut problem, the input is a directed graph $G$ with integral weights $w(e) > 0$ on the edges $e \in E(G)$. The goal of this problem is to produce an edge-cut $(X,Y)$ in $G$ while minimizing the weight $w(E_G(X,Y))$ of this edge-cut.
For such an instance of the Global Minimum Edge-Cut problem, we say that any edge-cut $(X,Y)$ of $G$ is a \emph{feasible solution} to this instance.

In the Rooted Minimum Edge-Cut problem, the input is a directed graph $G$ with integral weights $w(e) > 0$ on the edges $e \in E(G)$, as well as a \emph{root vertex} $r \in V(G)$.
A \emph{feasible solution} to this problem is an edge-cut $(X,Y)$ in $G$ with $r \in Y$. (We note that this problem is usually defined with the constraint $r \in X$ instead of $r \in Y$, but we use the symmetric version with $r \in Y$ in this paper.)
The goal of the problem is to compute a feasible solution $(X,Y)$ while minimizing the weight of the edge-cut $(X,Y)$.

\paragraph{Global Minimum Vertex-Cut and Rooted Minimum Vertex-Cut.}
In the Global Minimum Vertex-Cut problem, the input is a directed graph $G$ with integral weights $w(v) > 0$ on the vertices $v \in V(G)$, such that $G$ contains some vertex-cut. The goal of this problem is to produce a vertex-cut $(L,S,R)$ in $G$ while minimizing the weight $w(S)$ of this vertex-cut.
For such an instance of the Global Minimum Vertex-Cut problem, we say that any vertex-cut $(L,S,R)$ of $G$ is a \emph{feasible solution} to this instance.

In the Rooted Minimum Vertex-Cut problem, the input is a directed graph $G$ with integral weights $w(v) > 0$ on the vertices $v \in V(G)$, as well as a \emph{root vertex} $r \in V(G)$.
A \emph{feasible solution} to this problem is a vertex-cut $(L,S,R)$ in $G$ with $r \in R$. (We note that this problem is usually defined with the constraint $r \in L$ instead of $r \in R$, but we use the symmetric version with $r \in R$ in this paper.)
The goal of the problem is to compute a feasible solution $(L,S,R)$ while minimizing the weight of this vertex-cut.

\paragraph{Algorithmic Representation of Graphs.}
In this paper, whenever we say that an algorithm \emph{receives a graph $G$ as part of its input}, or that the algorithm is \emph{given access to} a graph $G$, we mean that the algorithm is given access to $G$ in all of the following forms: adjacency-list access, that, given a vertex $v \in V(G)$, allows accessing a linked-list containing all its outgoing edges, as well as another linked-list containing all its incoming edges; adjacency matrix access, that, given an ordered pair $(u,v)$ of vertices, declares in time $\Tilde{O}(1)$ whether a $u$-$v$ edge exists in $G$, as well as the weight of the edge if applicable (or a pointer to a linked-list containing all such edges, in the case of multigraphs); and, degree access, that allows querying the out-degree and in-degree of any vertex $v \in V(G)$ in $O(1)$ time.
We note that, using standard techniques, data-structures implementing these types of access can be constructed and maintained at the cost of only a $\Tilde{O}(1)$ multiplicative blowup in the running time of our algorithms.

\subsection{Submodularity of Cuts}
A folklore result that is often useful when dealing with cuts in graphs, is that the weight of edge-cuts and vertex-cuts is a \emph{submodular function} when viewed as a function of one side of the cut.
Our algorithms will use a corollary of this submodularity property, which is stated in the next two lemmas for the cases of edge-cuts and vertex-cuts, respectively.
The proofs of these lemmas follow from standard arguments, and are deferred to \Cref{sec:missing-proofs-in-prelims}.

\begin{lemma}\label{lem:submodularity-of-minimum-cuts:edge-version}
	Consider a graph $G$ with weights $w(e)$ on the edges $e \in E(G)$, as well as a pair $s,t$ of vertices in $G$, and a minimum $s$-$t$ edge-cut $(S,T)$ in $G$.
	Furthermore, let $X \subseteq V(G)$ be any set of vertices.
	Then $w(\partial^+_G(X \cap S)) \leq w(\partial^+_G(X))$.
\end{lemma}

\begin{lemma}\label{lem:submodularity-of-minimum-cuts:vertex-version}
	Consider a graph $G$ with weights $w(v)$ on the vertices $v \in V(G)$, as well as a pair $s,t$ of vertices in $G$, and a minimum $s$-$t$ vertex-cut $(L,S,R)$ in $G$.
	Furthermore, let $X \subseteq V(G)$ be any set of vertices.
	Then $w(N^+_G(X \cap L)) \leq w(N^+_G(X))$.
\end{lemma}

\subsection{Fast Algorithms for $s$-$t$ Minimum Cut}
Our algorithms use the recent result of \cite{vdBCK23} which solves the Min-Cost $s$-$t$ Maximum Flow problem in deterministic almost-linear time, and in particular shows that the $s$-$t$ Minimum Edge-Cut and $s$-$t$ Minimum Vertex-Cut problems can be solved in deterministic almost linear time.

\begin{lemma}[\cite{vdBCK23}]\label{thm:fast-max-s-t-edge-cut}
	There is a deterministic algorithm that, given an instance of the $s$-$t$ Minimum Edge-Cut problem on an $m$-edge directed graph $G$ with weights $w(e) \leq W$ on the edges $e \in E(G)$, produces an exact solution to this instance in time $O(m^{1+o(1)} \cdot \log W)$.
\end{lemma}

\begin{lemma}[\cite{vdBCK23}]\label{thm:fast-max-s-t-vertex-cut}
	There is a deterministic algorithm that, given an instance of the $s$-$t$ Minimum Vertex-Cut problem on an $m$-edge directed graph $G$ with weights $w(v) \leq W$ on the vertices $v \in V(G)$, produces an exact solution to this instance in time $O(m^{1+o(1)} \cdot \log W)$.
\end{lemma}

\subsection{Selecting Terminals}
%Consider an input graph $G$.
%In several of our algorithms, it will be useful for us to obtain an estimate $\nu$ of the out-volume $\vol^+_G(A)$ of some fixed unknown set $A \subseteq V(G)$ of vertices. (For example, $A$ may be the small side of an edge-cut/vertex-cut that we are trying to find.)
%In this setting, it will often also often be useful to for us to obtain a relatively small set $T \subseteq V(G)$ of \emph{terminals}, such that at least one terminal belongs to $A$.
%The following lemma provides a subroutine that, with probability $\Omega(\frac{1}{\log m})$, obtains such an estimate $\nu$ and set $T$ of terminals.
%Similar lemmas have been provided in prior work on Global Minimum Vertex-Cut (see Claim 4.2 from \cite{CMT25}), though they provide an estimate of $|A|$ rather than of $\vol^+_G(A)$.
%
%\begin{lemma}\label{lem: picking terminals}
%	There is a randomized algorithm, whose input consists of an $m$-edge graph $G$.
%	The output of the algorithm consists an integer $1 \leq \nu \leq 2m$, as well as a set $T \subseteq V(G)$ of $O(\frac{m}{\nu})$ vertices.
%	We say that the algorithm is successful with respect to a fixed set $A \subseteq V(G)$ of vertices, if $\frac{\vol^+_G(A)}{2} < \nu \leq \vol^+_G(A)$, and $T \cap A \neq \emptyset$.
%	For every fixed non-empty set $A \subseteq V(G)$, the algorithm is successful with respect to $A$ with probability $\Omega\left(\frac{1}{\log m}\right)$.
%	The running time of the algorithm is $\Tilde{O}(m)$.
%\end{lemma}
%
Consider an $m$-edge graph $G$.
In some of our algorithms, it will be useful to select a small set $T \subseteq V(G)$ of \emph{terminals}, such that at least one terminal belongs to a certain target set $A \subseteq V(G)$ of vertices, where $A$ is not known to the algorithm.
In this case, we will employ the following approach, which, with probability $\Omega\left(\frac{1}{\log m}\right)$, selects such a set $T$ of size $O\left(\frac{m}{\vol^+(A)}\right)$:
first, guess an estimate $\mu$ of $\vol^+_G(A)$ by picking $\mu$ uniformly at random among the integral powers of $2$ in the range $[1,m]$; then, sample $\frac{2m}{\nu}$ vertices into $T$, where each vertex is sampled with probability proportional to its out-degree.
The following lemma describes the guarantees of the sampling step in more detail. For the sake of convenience, we also state this lemma so that, if the algorithm knows a set $F$ of vertices that do not belong to $A$, then it can guarantee that these vertices do not get sampled.

\begin{lemma}\label{obs:picking terminals}
	There is a randomized algorithm whose input is an $m$-edge graph $G$, as well as a parameter $1 \leq \nu \leq m$, and a set $F \subset V(G)$ of \emph{forbidden vertices}.
	The output of the algorithm is a non-empty set $T \subseteq V(G)$ of at most $\frac{2m}{\nu}$ vertices called \emph{terminals}, such that $T \cap F = \emptyset$.
	The algorithm guarantees that, for every set $A \subseteq V(G) \setminus F$ of vertices with $\vol^+_G(A) \geq \nu$, the inequality $\Pr[A \cap T \neq \emptyset] \geq \frac{1}{100}$ holds.
	The running time of the algorithm is $\Tilde{O}(m)$.
\end{lemma}

\begin{proof}
	The algorithm from the lemma is implemented by sampling $\floor{\frac{2m}{\nu}}$ edges from $E(G)$ independently uniformly at random with repetition, and then choosing $T$ to be the set of those vertices $u \in V(G) \setminus F$ that have at least one outgoing sampled edge.
	To prove the correctness of the algorithm, fix any set $A \subseteq V(G) \setminus F$ with $\vol^+_G(A) \geq \nu$, and we will prove that $\Pr[A \cap T \neq \emptyset] \geq \frac{1}{100}$.
	Indeed, since $A \subseteq V(G) \setminus F$, it suffices to show that with probability at least $\frac{1}{100}$, there is at least one sampled edge in the set $E_G(A,V(G))$ of edges outgoing from vertices of $A$.
	Now, since each sampled edge independently has a probability of $\frac{\vol^+_G(A)}{m} \geq \frac{\nu}{m}$ to belong to $E_G(A,V(G))$, and since we sampled $\floor{\frac{2m}{\nu}} \geq \frac{m}{\nu}$ edges, it follows that the probability of at least one edge being sampled from $E_G(A,V(G))$ is at least
	\[
	 1 - \left(1 - \frac{\nu}{m}\right)^{m/\nu}
	 \geq 1 - \frac{1}{e}
	 \geq \frac{1}{100}.
	\]
	This concludes the proof of \Cref{obs:picking terminals}.
\end{proof}

\section{Algorithm for Rooted Minimum Edge-Cut: Proof of \Cref{thm : main-thm-rooted-edge-version}}\label{sec:rooted-minimum-edge-cut}

In this section, we describe and analyze the algorithm implementing \Cref{thm : main-thm-rooted-edge-version} -- our main result for the Rooted Minimum Edge-Cut problem.
As explained in \Cref{sec:overview-of-rooted-minimum-cut}, our algorithm employs a hierarchical divide-and-conquer technique due to \cite{CMT25}.
We now formally describe the algorithm.

Recall that the input to the algorithm consists of a directed $m$-edge graph $G$ with positive integer weights $w(e) \leq W$ on the edges $e \in E(G)$, as well as a precision parameter $\epsilon \in (0,1)$, and a root vertex $y^* \in V(G)$.
Let $\OPT_{E,y^*}$ denote the minimum value of an edge-cut among all edge-cuts $(X',Y')$ of $G$ with $y^* \in Y'$.
Furthermore, for the sake of the analysis, fix one such edge-cut achieving this minimum value, and denote it by $(X^*,Y^*)$ -- we refer to $(X^*,Y^*)$ as \emph{the distinguished cut}.
Our goal is to output an edge-cut $(X,Y)$ of $G$ with $w(E_G(X,Y)) \leq (1+\epsilon) \OPT_{E,y^*}$.

It will also be convenient for us to assume that $\OPT_{E,y^*} > 0$, and we indeed make this assumption in the remainder of this section.
This assumption can be removed via the following standard technique: before running the algorithm, we check whether $\OPT_{E,y^*} > 0$ by performing a BFS from $y^*$ in the reverse graph $\overline{G}$ of $G$; if the BFS reaches all vertices of $\overline{G}$, then $\OPT_{E,y^*} > 0$ must hold, in which case we run the algorithm; otherwise, we can use the set of vertices reached by the BFS in order to construct and output an edge-cut $(X,Y)$ of $G$ with $w(E_G(X,Y)) = 0$ and $y^* \in Y$ -- specifically, $Y$ is the set of vertices reached by the BFS, and $X = V(G) \setminus Y$.

\paragraph{Estimates of $\OPT_{E,y^*}$ and $\vol^+_G(X^*)$, and a Set of Terminals.}
Our algorithm begins by guessing an estimate $\widehat{\OPT}$ of $\OPT_{E,y^*}$, and an estimate $\nu$ of $\vol^+_G(X^*)$ -- specifically, it picks $\widehat{\OPT}$ uniformly at random among the integral powers of $2$ in the range $[1,m \cdot W]$, and it picks $\nu$ uniformly at random among the integral powers of $2$ in the range $[1,m]$.
Let $\event^{\text{estimates}}$ denote the good event in which $\frac{\OPT_{E,y^*}}{2} < \widehat{\OPT} \leq \OPT_{E,y^*}$ and $\frac{\vol^+_G(X^*)}{2} < \nu \leq \vol^+_G(X^*)$ both hold.
Observe that $\Pr[\event^{\text{estimates}}] \geq \Omega\left(\frac{1}{\log^2(m \cdot W)}\right)$.
Next, our algorithm picks a set $T \subseteq V(G)$ of $O\left(\frac{m}{\nu}\right)$ \emph{terminals}, using the procedure from \Cref{obs:picking terminals} with parameter $\nu$ and the set $F=\{y^*\}$ of forbidden vertices.
Let $\event$ denote the good event that event $\event^{\text{estimates}}$ occurred and $T \cap X^* \neq \emptyset$ holds.
By \Cref{obs:picking terminals}, $\Pr[\event \mid \event^{\text{estimates}}] \geq \frac{1}{100}$, and therefore
\begin{equation}\label{eq:probability-of-good-event:edge-version}
	\Pr[\event] = \Pr[\event \mid \event^{\text{estimates}}] \cdot \Pr[\event^{\text{estimates}}] \geq \Omega\left(\frac{1}{\log^2(m \cdot W)}\right).
\end{equation}

\paragraph{Hierarchical Partition of the Terminals.}
Our algorithm constructs a hierarchical partition of the set $T$ of terminals, which is identical to the hierarchical partition used by \cite{CMT25}. For the sake of completeness, we now describe this construction.
Let $z = \ceil{\log |T|}$. For each integer $0 \leq i \leq z$, the algorithm creates a partition of the set $T$ of terminals into at most $2^i$ \emph{level-$i$ batches}, that each contain at most $2^{z-i}$ terminals. The partition of $T$ into level-$i$ batches is referred to as the \emph{level-$i$ partition}, and the set of all level-$i$ batches is denoted by $\bset_i$.
Specifically, the level-$0$ partition has only a single batch, and this batch contains all the terminals of $T$.
For each $1 \leq i \leq z$, the algorithm creates the level-$i$ partition by splitting each batch $B \in \bset_{i-1}$ of the level-$(i-1)$ partition into two parts $B',B'' \subseteq B$ whose sizes are $\ceil{\frac{|B|}{2}}$ and $\floor{\frac{|B|}{2}}$, respectively, and adding both $B'$ and $B''$ to $\bset_i$ as level-$i$ batches; we say that the batch $B$ is the level-$(i-1)$ \emph{parent batch} of the level-$i$ batches $B'$ and $B''$.
(The exception is when the batch $B \in \bset_{i-1}$ is of size $|B|=1$; in this case, we only create a single level-$i$ batch whose level-$(i-1)$ parent batch is $B$ -- this level-$i$ batch is equal to $B$.)
The next observation is useful in the following parts of this section.

\begin{observation}\label{obs:size-of-batches}
	For every $0 \leq i \leq z$, every level-$i$ batch $B \in \bset_i$ has size $|B| = O\left(\frac{m}{\nu \cdot 2^i}\right)$.
\end{observation}
\begin{proof}
	By the description of the $i$'th level partition, every level-$i$ batch has size at most $2^{z-i}$, where $z = \ceil{\log |T|}$.
	Thus, the size of every level-$i$ batch is $O\left(\frac{|T|}{2^i}\right)$.
	The observation follows since the set $T$ has size $|T| = O\left(\frac{m}{\nu}\right)$.
\end{proof}

\paragraph{Main Loop of the Algorithm.}
Our algorithm consists of $z$ phases.
For each $1 \leq i \leq z$, the goal of Phase $i$ is to construct, for each level-$i$ batch $B \in \bset_i$, a graph $G_B$ and positive integral weights $w_B(e) \leq W$ on the edges $e \in E(G_B)$, that together satisfy the following properties:

\begin{properties}{P}
	\item\label{property:G_B-P1:edge-version} the graph $G_B$ contains at most $\Tilde{O}\left(\frac{m}{2^i \cdot \epsilon}\right)$ edges and vertices, and no more than the number of edges and vertices in $G$;
	\item\label{property:G_B-P2:edge-version} $\{y^*\} \subseteq V(G_B) \subseteq V(G)$;
	\item\label{property:G_B-P3:edge-version} for every set $X' \subseteq V(G_B) \setminus \{y^*\}$, the inequality $w_B(\partial^+_{G_B}(X')) \geq w(\partial^+_{G}(X'))$ holds -- in particular, for every terminal $x \in B \cap V(G_B)$, the $x$-$y^*$ edge-connectivity $\lambda_{G_B}(x,y^*)$ in $G_B$ is at least as large as the edge-connectivity $\lambda_G(x,y^*)$ in $G$;
	\item\label{property:G_B-P4:edge-version} if there exists a terminal $x^* \in B$ such that $x^* \in X^*$ and the good event $\event$ occurred, then $x^* \in V(G_B)$ holds and the $x^*$-$y^*$ edge-connectivity $\lambda_{G_B}(x^*,y^*)$ in $G_B$ satisfies $\lambda_{G_B}(x^*,y^*) \leq \left(1 + \frac{i\epsilon}{z}\right)\OPT_{E,y^*}$ -- furthermore, there exists an $x^*$-$y^*$ edge-cut $(X',Y')$ in $G_B$ of value\mynewline$w_B(E_{G_B}(X',Y')) \leq \left(1 + \frac{i\epsilon}{z}\right)\OPT_{E,y^*}$ that additionally satisfies the condition $\vol^+_G(X') \leq 2\nu$.
\end{properties}
Each phase will be partitioned into iterations, where each iteration computes the graph $G_B$ for a single batch $B$.
The main technical component of our algorithm is the following claim, which states that each iteration of Phase $i$ can be performed in $O\left(\frac{m^{1+o(1)} \cdot \log W}{2^i \cdot \epsilon}\right)$ time.
The proof of this claim- and the description of procedure implementing it- are presented in \Cref{sec:description-of-iteration-of-main-algorithm:edge-version}.

\begin{claim}\label{cl:implementing-a-single-iteration-of-main-algorithm:edge-version}
	There is a deterministic procedure that we refer to as \approxsparsify, that, given an integer $1 \leq i \leq z$, a batch $B \in \bset_i$, the parameters $\epsilon,\nu,\widehat{\OPT}$, and access to the graph $G$, computes the graph $G_B$ satisfying Properties \ref{property:G_B-P1:edge-version}-\ref{property:G_B-P4:edge-version} in $O\left(\frac{m^{1+o(1)} \cdot \log W}{2^i \cdot \epsilon}\right)$ time. If $i>1$, then the procedure also requires access to the graph $G_{B'}$ corresponding to the level-$(i-1)$ parent batch $B'$ of $B$.
\end{claim}

\noindent The following corollary of \Cref{cl:implementing-a-single-iteration-of-main-algorithm:edge-version} summarizes the running time of each phase.

\begin{corollary}\label{cor:running-time-of-one-phase:edge-version}
	Each phase of the algorithm can be implemented in time $O\left(\frac{m^{1+o(1)} \cdot \log W}{\epsilon}\right)$.
\end{corollary}
\begin{proof}
	Observe that for every $1 \leq i \leq z$, Phase $i$ consists of at most $2^i$ iterations -- one for each level-$i$ batch $B \in \bset_i$.
	Furthermore, each iteration can be implemented using a single call to the \approxsparsify subroutine.
	Thus, this corollary follows from \Cref{cl:implementing-a-single-iteration-of-main-algorithm:edge-version}.
\end{proof}

\paragraph{Termination.}
We now describe how our algorithm proceeds after the final phase.
For each terminal $x \in T$, let $B^{(x)} \in \bset_z$ denote the level-$z$ batch that contains $x$.
Furthermore, let $T' \subseteq T$ denote the set of all terminals $x$ such that $x \in V(G_{B^{(x)}})$.
After the final phase, for each terminal $x \in T'$, our algorithm uses \Cref{thm:fast-max-s-t-edge-cut} on the graph $G_{B^{(x)}}$ in order to determine the edge connectivity $\lambda_{G_{B^{(x)}}}(x,y^*)$.
Then, our algorithm picks the terminal $\tilde{x}$ that achieves the minimum edge-connectivity $\lambda_{G_{B^{(x)}}}(x,y^*)$ among all the terminals of $T'$; computes a minimum $\tilde{x}$-$y^*$ edge-cut $(X,Y)$ in $G$ using \Cref{thm:fast-max-s-t-edge-cut}, and outputs this edge-cut. (If the set $T'$ is empty, then the algorithm instead outputs an arbitrary edge-cut $(X,Y)$ with $y^* \in Y$.)
Observe the following.

\begin{observation}\label{obs:termination-value-of-cut-part1:edge-version}
	If $T' \neq \emptyset$, then the value of the edge-cut $(X,Y)$ returned by the algorithm is at most $\min_{x \in T'}\lambda_{G_{B^{(x)}}}(x,y^*)$.
\end{observation}
\begin{proof}
	As $(X,Y)$ is selected to be a minimum $\tilde{x}$-$y^*$ edge-cut in $G$, it suffices to prove that $\lambda_G(\tilde{x},y^*) \leq \min_{x \in T'}\lambda_{G_{B^{(x)}}}(x,y^*)$.
	Indeed, Property \ref{property:G_B-P3:edge-version} guarantees that $\lambda_G(\tilde{x},y^*) \leq \lambda_{G_{B^{(\tilde{x})}}}(\tilde{x},y^*)$, and the choice of $\tilde{x}$ guarantees that $\lambda_{G_{B^{(\tilde{x})}}}(\tilde{x},y^*) = \min_{x \in T'}\lambda_{G_{B^{(x)}}}(x,y^*)$.
\end{proof}

\noindent The next corollary follows from the above observation.
\begin{corollary}\label{cor:good-event-implies-success:edge-version}
	If the good event $\event$ occurs, then the edge-cut $(X,Y)$ returned by the algorithm has value at most $\left(1+\epsilon\right)\OPT_{E,y^*}$.
\end{corollary}
\begin{proof}
	Suppose that the good event $\event$ occurs.
	In this case, by the definition of this event, there exists a vertex $x^* \in T \cap X^*$.
	Now, consider the graph $G_{B^{(x^*)}}$ corresponding to the level-$z$ batch $B^{(x^*)}$.
	By Property \ref{property:G_B-P4:edge-version}, this graph satisfies that $x^* \in V\left(G_{B^{(x^*)}}\right)$ and that $\lambda_{G_{B^{(x^*)}}}(x^*,y^*) \leq \left(1+\epsilon\right)\OPT_{E,y^*}$.
	In other words, $x^*$ is a vertex of $T'$ that satisfies $\lambda_{G_{B^{(x^*)}}}(x^*,y^*) \leq \left(1+\epsilon\right)\OPT_{E,y^*}$.
	Therefore, $T' \neq \emptyset$ and $\min_{x \in T'}\lambda_{G_{B^{(x)}}}(x,y^*) \leq \left(1+\epsilon\right)\OPT_{E,y^*}$.
	The corollary now follows from \Cref{obs:termination-value-of-cut-part1:edge-version}.
\end{proof}

\paragraph{Running Time Analysis.}
We now prove that the total running time of the algorithm described in this section is $O\left(\frac{m^{1+o(1)}}{\epsilon} \cdot \log W\right)$.
We begin by upper-bounding the running time of the main loop:
by \Cref{cor:running-time-of-one-phase:edge-version}, each phase of the main loop of the algorithm runs in $O\left(\frac{m^{1+o(1)} \cdot \log W}{\epsilon}\right)$ time, so the total time of all $z=O(\log n)$ phases is also upper-bounded by $O\left(\frac{m^{1+o(1)} \cdot \log W}{\epsilon}\right)$.
It remains to upper-bound the time spent after the last phase.
Indeed, after the last phase, the algorithm makes at most one call to \Cref{thm:fast-max-s-t-edge-cut} on each of the graphs $G_B$ constructed during Phase $z$;
observe that, by Property \ref{property:G_B-P1:edge-version}, the running time of each such call is $O(\frac{m^{1+o(1)} \cdot \log W}{2^z \cdot \epsilon})$. So, the total time spent on calls to \Cref{thm:fast-max-s-t-edge-cut} after the end of the final phase is upper-bounded by $O(|T| \cdot \frac{m^{1+o(1)} \cdot \log W}{2^z \cdot \epsilon}) = O(\frac{m^{1+o(1)} \cdot \log W}{\epsilon})$, as needed. (In the last equality, we used the definition of $z$ as $z = \ceil{\log |T|}$.)

\paragraph{Correctness Analysis.}
By the description of the algorithm, it always outputs an edge-cut $(X,Y)$ in $G$ such that $y^* \in Y$.
Furthermore, by \Cref{cor:good-event-implies-success:edge-version}, whenever event $\event$ occurs, this edge-cut has value at most $(1+\epsilon)\OPT_{E,y^*}$; by \Cref{eq:probability-of-good-event:edge-version}, this happens with probability $\Omega\left(\frac{1}{\log^2(m \cdot W)}\right)$.
Thus, the algorithm described in this section runs in time $O\left(\frac{m^{1+o(1)} \cdot \log W}{\epsilon}\right)$, and the output edge-cut $(X,Y)$ satisfies $w(E_G(X,Y)) \leq (1+\epsilon)\OPT_{E,y^*}$ with probability at least $\frac{1}{\polylog(mW)}$.
By repeating the algorithm $\polylog(mW)$ many times and choosing the lowest-value cut among the obtained edge-cuts, we can boost the success probability of the algorithm to $\frac{1}{2}$, at the cost of increasing the running time to $O\left(\frac{m^{1+o(1)} \cdot \polylog W}{\epsilon}\right)$.
We can then reduce the running time back to $O\left(\frac{m^{1+o(1)} \cdot \log W}{\epsilon}\right)$ using the reduction from \Cref{sec: low-dependence-on-W}.

\subsection{Execution of a Single Iteration: Proof of \Cref{cl:implementing-a-single-iteration-of-main-algorithm:edge-version}}\label{sec:description-of-iteration-of-main-algorithm:edge-version}\label{sec:comparison-of-techniques-for-batch-sparsify}

In this section, we describe the \approxsparsify procedure from \Cref{cl:implementing-a-single-iteration-of-main-algorithm:edge-version}.
We firstly deal with the easy special case where $i=1$: in this case, the procedure simply lets $G_B$ be the entire graph $G$; it can be easily verified that in this case, selecting $G_B=G$ guarantees that $G_B$ satisfies all four Properties \ref{property:G_B-P1:edge-version}-\ref{property:G_B-P4:edge-version}.

In \Cref{sec:approx-sparsify-for-i-larger-than-1:edge-version}, we present the formal description and analysis of the procedure \approxsparsify for the case of $i>1$.
In the remainder of the current section, we present an informal overview of the techniques used in this procedure, and put them in the context of prior work.

\paragraph{Overview of Techniques.}
Consider an input to \approxsparsify.
Using standard techniques (that we describe in \Cref{cl:reduction-to-computing-A_B:edge-version}), the problem of computing the graph $G_B$ can be reduced to the problem of computing a set $A_B \subseteq V(G)$ which satisfies two properties, (defined later as Properties \ref{property:A_B-P'1:edge-version} and \ref{property:A_B-P'2:edge-version}): the first property is that $\vol^+_G(A_B) = \Tilde{O}\left(\frac{m}{2^i \cdot \epsilon}\right)$; the second property is that, if there exists a terminal $x^* \in B \cap X^*$ and the good event $\event$ occurred, then there exists an $x^*$-$y^*$ edge-cut $(X,Y)$ in $G$ with $X \subseteq A_B$ and $\vol^+_G(X) \leq 2\nu$, whose value is $w(E_{G}(X,Y)) \leq \left(1 + \frac{i\epsilon}{z}\right)\OPT_{E,y^*}$. We use this reduction. (Indeed, a vertex-analogue of this reduction is also used by the algorithm of \cite{CMT25}.)
Thus, the main challenge is computing such a set $A_B \subseteq V(G)$ of vertices.

Consider the graph $G_{B'}$ corresponding to the parent batch $B'$ of $B$, and suppose for the sake of intuition that event $\event$ occurred, and that there exists some terminal $x^* \in B \cap X^* \subseteq B' \cap X^*$.
Then, by Property \ref{property:G_B-P4:edge-version} of the graph $G_{B'}$, there exists an $x^*$-$y^*$ edge-cut $(X'_{B'},Y'_{B'})$ in $G_{B'}$ of value $w_{B'}(E_{G_{B'}}(X'_{B'},Y'_{B'})) \leq \left(1 + \frac{(i-1)\epsilon}{z}\right)\OPT_{E,y^*}$ that additionally satisfies the condition $\vol^+_G(X'_{B'}) \leq 2\nu$.
We now describe a simple method for attempting to compute the set $A_B$, which is instructive for understanding our techniques, but which does not work as-is: begin by transforming the graph $G_{B'}$ into a graph $H_{B'}$ by introducing a new "super-source" vertex $s$, and adding edges from $s$ to every vertex of $B$; then, compute a minimum $s$-$y^*$ edge-cut $(A,A')$ in $H_{B'}$, and output the set $A_B=A \setminus \set{s}$.
Using the submodularity of cuts, it can be shown that as long as the edges leaving $s$ have sufficiently large weight, the set $A_B$ must satisfy Property \ref{property:A_B-P'2:edge-version} -- specifically, it can be shown that the edge-cut $(X,Y)=(X'_{B'} \cap A_B, Y'_{B'} \cup A')$ is an $x^*$-$y^*$ edge-cut in $G$ with value no larger than the value of the edge-cut $(X'_{B'},Y'_{B'})$ in $G_{B'}$, which in turn is at most $\left(1 + \frac{(i-1)\epsilon}{z}\right)\OPT_{E,y^*}$.
The issue with this method is that it does not provide any upper-bound on the out-volume $\vol^+_G(A_B)$; thus, some changes must be made to the method in order to guarantee Property \ref{property:A_B-P'1:edge-version}.

At a high-level, we deal with this issue by "disincentivizing" the $s$-$t$ Minimum Edge-Cut subroutine from picking a cut $(A,A')$ where $A \setminus \set{s}$ has high out-volume -- specifically, we add low-capacity edges from each vertex of $H_{B'}$ (except $s$) to $y^*$, where the capacity of each new edge $(v,y)$ is proportional to the out-degree $\deg^+_G(v)$.
We then show that these extra edges only slightly distort the values of the cuts that we care about, meaning that the value of the edge-cut $(X'_{B'} \cap A_B, Y'_{B'} \cup A')$ in $G$ will be larger than the value of $(X'_{B'},Y'_{B'})$ in $G_{B'}$ by at most a $\frac{\epsilon \cdot \OPT_{E,y^*}}{z}$ additive factor, so it will still be upper-bounded by $\left(1 + \frac{i\epsilon}{z}\right)\OPT_{E,y^*}$, as needed.

\paragraph{Comparison to \cite{CMT25}.}
The algorithm of \cite{CMT25} can also be thought of as using a variation of the method described above, but it uses a different and more involved approach for dealing with the issue of bounding $\vol^+_G(A_B)$.
Furthermore, the approach of \cite{CMT25} only provides an upper-bound of $\Tilde{O}(\frac{m \cdot \widehat{\OPT}}{2^i})$ on $\vol^+_G(A_B)$, and it is unclear whether this approach can be improved to yield better upper-bounds.

\paragraph{Relation to Local Cut.}
We note that our technique for bounding $\vol^+_G(A_B)$ is very similar to a technique used by the "$(1+\epsilon)$-Approximate Local Cut" subroutine of \cite{FNY20}. Indeed, our algorithm for finding $A_B$ can be thought of as first solving an Approximate Local Cut problem (see \cite{FNY20}) in the graph $H_{B'}$ to obtain an approximately minimum $s$-$y^*$ cut $(A,A')$ where $\vol^+_G(A)$ is small, and then setting $A_B = A \setminus \set{s}$. However, we could not use the approximate Local Cut subroutine of \cite{FNY20} as-is, since it is too slow in our setting.

\subsubsection{Implementation of \approxsparsify for $i>1$.}\label{sec:approx-sparsify-for-i-larger-than-1:edge-version}

Fix some $1 \leq i \leq z$. We now formally describe how the algorithm computes the graph $G_B$ for a single batch $B \in \bset_i$ during Phase $i$ of the algorithm.
The computation of the graph $G_B$ is performed in two steps.
In the first step, we compute a set $A_B \subseteq V(G) \setminus \{y^*\}$ of vertices that satisfies the following two properties.
\begin{properties}{Q}
	\item\label{property:A_B-P'1:edge-version} $\vol^+_G(A_B) = \Tilde{O}\left(\frac{m}{2^i \cdot \epsilon}\right)$; and
	\item\label{property:A_B-P'2:edge-version} if there exists a terminal $x^* \in B \cap X^*$ and the good event $\event$ occurred, then there exists a subset $X' \subseteq A_B$ with $x^* \in X'$ and $\vol^+_G(X') \leq 2\nu$ such that the edge-cut $(X',V(G) \setminus X')$ in $G$ has value $w(E_{G}(X',V(G) \setminus X')) \leq \left(1 + \frac{i\epsilon}{z}\right)\OPT_{E,y^*}$.
\end{properties}
In the second step, we compute the graph $G_B$ so that it is equal to the graph obtained from $G$ by merging all vertices outside of $A_B$ into the vertex $y^*$, and then removing any resulting self-loops, as well as any edges outgoing from $y^*$ in the resulting graph. (We note that the resulting graph $G_B$ may have parallel edges.)
This second step is summed up in the following claim, whose proof is deferred to \Cref{sec:missing-proofs-in-main-edge-version}.
\begin{claim}\label{cl:reduction-to-computing-A_B:edge-version}
	Consider a set $A_B \subseteq V(G) \setminus \{y^*\}$ that satisfies properties \ref{property:A_B-P'1:edge-version}-\ref{property:A_B-P'2:edge-version}, and consider the graph $G_B$ that is obtained from $G$ by merging all vertices outside of $A_B$ into $y^*$ and then removing any resulting self-loops, as well as any edges outgoing from $y^*$ in the resulting graph.
	Then, the graph $G_B$ must satisfy Properties \ref{property:G_B-P1:edge-version}-\ref{property:G_B-P4:edge-version}.
	Furthermore, given such a set $A_B$, we can construct the graph $G_B$ in time $O\left(\vol^+_G(A_B)\right) = \Tilde{O}\left(\frac{m}{2^i \cdot \epsilon}\right)$.
\end{claim}
In order to complete the description of the iteration, it now suffices to describe the procedure for computing the set $A_B$ of vertices.
We do this next.

\paragraph{Obtaining the set $A_B$.}
We now describe how our algorithm computes the set $A_B \subseteq V(G) \setminus \{y^*\}$ of vertices satisfying Properties \ref{property:A_B-P'1:edge-version}-\ref{property:A_B-P'2:edge-version} with respect to a level-$i$ batch $B \in \bset_i$ during an iteration of Phase $i$.
Let $B' \in \bset_{i-1}$ denote the level-$(i-1)$ parent batch of $B$.
In order to compute $A_B$, we will use the graph $G_{B'}$ that we computed for batch $B'$ in the previous phase.
To compute $A_B$, we first construct a graph $H_{B'}$ from $G_{B'}$ as follows: for every vertex $v \in V(G_{B'})$ whose out-degree $\deg^+_G(v)$ in $G$ is non-zero, we add to $G_{B'}$ an edge $(v,y^*)$ of capacity $\frac{\epsilon\cdot\widehat{\OPT} \cdot \deg^+_G(v)}{z \cdot 2\nu}$; furthermore, we add a new \emph{super-source} vertex $s$, and we connect $s$ to each vertex of $B \cap V(G_{B'})$ via an edge of capacity $4\widehat{\OPT}$.
In the remainder of this section, we will overload notation and use $w_{B'}$ to refer to the edge-weight function in both $G_{B'}$ and $H_{B'}$, as these two functions agree on all edges that are shared by these two graphs.
The next observation about the graph $H_{B'}$ follows immediately from the construction of this graph.
\begin{observation}\label{obs:relating-cut-value-after-modification:edge-version}
	For every set $\hat{X} \subseteq V(H_{B'}) \setminus \{s,y^*\}$ of vertices, the difference between the total weight of edges outgoing from $\hat{X}$ in $H_{B'}$ and the total weight of edges outgoing from $\hat{X}$ in $G_{B'}$ is $w_{B'}(\partial^+_{H_{B'}}(\hat{X})) - w_{B'}(\partial^+_{G_{B'}}(\hat{X})) = \frac{\epsilon\cdot\widehat{\OPT} \cdot \vol^+_G(\hat{X})}{z \cdot 2\nu}$.
\end{observation}

\noindent To conclude the computation of $A_B$, we use \Cref{thm:fast-max-s-t-edge-cut} in order to compute a minimum $s$-$y^*$ edge-cut $(A,A')$ in the graph $H_{B'}$, and we let $A_B = A \setminus \{s\}$.
This concludes the description of the algorithm, and it now only remains to analyze it.
Before proving that the set $A_B$ satisfies Properties \ref{property:A_B-P'1:edge-version}-\ref{property:A_B-P'2:edge-version}, we first present the next two observations about the edge-cut $(A,A')$ in the graph $H_{B'}$.
\begin{observation}\label{obs:bound-on-size-of-cut-in-H-graph:edge-version}
	The value $w_{B'}(E_{H_{B'}}(A,A'))$ of the edge-cut $(A,A')$ in the graph $H_{B'}$ satisfies $w_{B'}(E_{H_{B'}}(A,A'))=O\left(\frac{m \cdot \widehat{\OPT}}{\nu \cdot 2^i}\right)$.
\end{observation}
\begin{proof}
	Recall that $(A,A')$ is chosen to be a minimum $s$-$y^*$ edge-cut in $H_{B'}$; thus, it suffices to show that some $s$-$y^*$ edge-cut in $H_{B'}$ has value $O\left(\frac{m \cdot \widehat{\OPT}}{\nu \cdot 2^i}\right)$.
	Indeed, by the construction of the graph $H_{B'}$, there are at most $|B|$ edges outgoing from the vertex $s$, and the capacity of each such edge is exactly $4\widehat{\OPT}$ -- therefore, the value of the $s$-$y^*$ edge-cut $(\set{s},V(H_{B'}) \setminus\set{s})$ in $H_{B'}$ is at most $|B| \cdot 4\widehat{\OPT} \leq O\left(\frac{m \cdot \widehat{\OPT}}{\nu \cdot 2^i}\right)$, where the last inequality follows from \Cref{obs:size-of-batches}.
	This concludes the proof of \Cref{obs:bound-on-size-of-cut-in-H-graph:edge-version}.
\end{proof}

\begin{observation}\label{obs:minimum-cut-in-modified-graph-contains-distinguished-terminal}
	The set $A$ contains every terminal $x \in B \cap V(G_{B'})$ whose edge-connectivity to $y^*$ in $H_{B'}$ satisfies $\lambda_{H_{B'}}(x,y^*) < 4\widehat{\OPT}$.
\end{observation}
\begin{proof}
	Suppose towards contradiction that some terminal $x \in B \cap V(G_{B'})$ does not belong to $A$, and yet has $\lambda_{H_{B'}}(x,y^*) < 4\widehat{\OPT}$.
	Then, there must exist an $x$-$y^*$ edge-cut $(\hat{X},\hat{Y})$ in $H_{B'}$ with value $w_{B'}(\partial_{H_{B'}}(\hat{X})) < 4\widehat{\OPT}$.
	To complete the proof of the observation, we will now show that the edge-cut $(A \cup \hat{X}, A' \cap \hat{Y})$ in $H_{B'}$ has a smaller value than the edge-cut $(A,A')$, contradicting the choice of $(A,A')$ as a minimum $s$-$y^*$ edge-cut in $H_{B'}$:
	indeed, recall that the graph $H_{B'}$ contains an edge $(s,x)$ of weight $w_{B'}(s,x) = 4\widehat{\OPT} > w_{B'}(\partial_{H_{B'}}(\hat{X}))$; furthermore, observe that, as $x \in \hat{X} \setminus A$, this edge exits the set $A$ but not the set $A \cup \hat{X}$, meaning that $\partial_{H_{B'}}(A \cup \hat{X}) \subseteq \partial_{H_{B'}}(\hat{X}) \cup \partial_{H_{B'}}(A) \setminus \set{(s,x)}$.
	Therefore, the value of the edge-cut $(A \cup \hat{X}, A' \cap \hat{Y})$ in $H_{B'}$ is $w_{B'}(\partial_{H_{B'}}(A \cup \hat{X})) \leq w_{B'}(\partial_{H_{B'}}(A)) + w_{B'}(\partial_{H_{B'}}(\hat{X})) - w_{B'}(s,x) < w_{B'}(\partial_{H_{B'}}(A))$.
	This concludes the proof of \Cref{obs:minimum-cut-in-modified-graph-contains-distinguished-terminal}.
\end{proof}

Next, we show that the set $A_B$ computed by the above procedure indeed satisfies Properties \ref{property:A_B-P'1:edge-version}-\ref{property:A_B-P'2:edge-version}.
We begin by proving that $A_B$ satisfies Property \ref{property:A_B-P'1:edge-version}.

\begin{claim}\label{cl:A_B-satisfies-property-P'1:edge-version}
	The set $A_B$ satisfies Property \ref{property:A_B-P'1:edge-version}. That is, $\vol^+_G(A_B) = \Tilde{O}\left(\frac{m}{2^i \cdot \epsilon}\right)$.
\end{claim}
\begin{proof}
	Recall that for each vertex $v \in A_B$, if $\deg^+_G(v) > 0$, then the graph $H_{B'}$ contains an edge $(v,y^*)$ of capacity $\frac{\epsilon\cdot\widehat{\OPT} \cdot \deg^+_G(v)}{z \cdot 2\nu}$, and this edge must cross the edge-cut $(A,A')$ in $H_{B'}$.
	Therefore, the value $w_{B'}(E_{H_{B'}}(A,A'))$ of the edge-cut $(A,A')$ in $H_{B'}$ is at least as large as $\frac{\epsilon\cdot\widehat{\OPT} \cdot \vol^+_G(A_B)}{z \cdot 2\nu}$, meaning that $\vol^+_G(A_B)$ is upper-bounded by $\vol^+_G(A_B) \leq w_{B'}(E_{H_{B'}}(A,A')) \cdot \frac{z \cdot 2\nu}{\epsilon\cdot\widehat{\OPT}}$.
	By plugging the upper-bound $w_{B'}(E_{H_{B'}}(A,A'))=O\left(\frac{m \cdot \widehat{\OPT}}{\nu \cdot 2^i}\right)$ from \Cref{obs:bound-on-size-of-cut-in-H-graph:edge-version} into this last inequality, we get that $\vol^+_G(A_B) \leq O\left(\frac{z \cdot m}{\epsilon \cdot 2^{i}}\right)$.
	The corollary now follows as $z = \Tilde{O}(1)$.
\end{proof}

Next, we prove that the set $A_B$ satisfies Property \ref{property:A_B-P'2:edge-version}.
\begin{claim}\label{cl:A_B-satisfies-property-P'2:edge-version}
	The set $A_B$ satisfies Property \ref{property:A_B-P'2:edge-version}.
\end{claim}
\begin{proof}
	Suppose that there exists a terminal $x^* \in B$ such that $x^* \in X^*$, and suppose that the good event $\event$ occurred.
	We need to prove that there exists a subset $X' \subseteq A_B$ with $x^* \in X'$ and $\vol^+_G(X') \leq 2\nu$ such that the edge-cut $(X',V(G) \setminus X')$ in $G$ has value $w(E_{G}(X',V(G) \setminus X')) \leq \left(1 + \frac{i\epsilon}{z}\right)\OPT_{E,y^*}$.
	
	As $x^* \in B \cap X^*$ and as $B'$ is the parent batch of $B$, it follows that $x^* \in B' \cap X^*$.
	Thus, by Property \ref{property:G_B-P4:edge-version} of the graph $G_{B'}$, $x \in V(G_{B'})$ holds, and there exists an $x^*$-$y^*$ edge-cut $(X'_{B'},Y'_{B'})$ in $G_{B'}$ of value $w_{B'}(E_{G_{B'}}(X'_{B'},Y'_{B'})) \leq \left(1 + \frac{(i-1)\epsilon}{z}\right)\OPT_{E,y^*}$, that additionally satisfies $\vol^+_G(X'_{B'}) \leq 2\nu$.
	
	Let $X' = A \cap X'_{B'}$, and observe that $X' \subseteq A \cap V(G_{B'}) = A_B$.
	We now prove that $x^* \in X'$ and $\vol^+_G(X') \leq 2\nu$, and that the edge-cut $(X',V(G) \setminus X')$ in $G$ has value $w(E_{G}(X',V(G) \setminus X')) \leq \left(1 + \frac{i\epsilon}{z}\right)\OPT_{E,y^*}$.
	Indeed, the fact that $\vol^+_G(X') \leq 2\nu$ follows from the fact that $X' \subseteq X'_{B'}$, and that $\vol^+_G(X'_{B'}) \leq 2\nu$.
	It remains to prove that $x^* \in X'$, as well as that the edge-cut $(X',V(G) \setminus X')$ in $G$ has value $w(E_{G}(X',V(G) \setminus X')) \leq \left(1 + \frac{i\epsilon}{z}\right)\OPT_{E,y^*}$.
	We begin by proving the latter statement. Specifically, this latter statement follows directly from the next claim.
	
	\begin{claim}\label{cl:upper-bound-on-value-of-new-promised-cut:edge-version}
		The three inequalities $w(E_{G}(X',V(G) \setminus X')) \leq w_{B'}(\partial^+_{H_{B'}}(X')) \leq w_{B'}(\partial^+_{H_{B'}}(X'_{B'})) \leq \left(1 + \frac{i\epsilon}{z}\right)\OPT_{E,y^*}$ all hold.
	\end{claim}
	\begin{proof}
		We begin by proving the first inequality of the claim. That is, we begin by proving that $w(E_{G}(X',V(G) \setminus X')) \leq w_{B'}(\partial^+_{H_{B'}}(X'))$.
		Indeed, by Property \ref{property:G_B-P3:edge-version} of the graph $G_{B'}$, the value of the edge-cut $(X',V(G) \setminus X')$ in $G$ is no larger than the value of the edge-cut $(X',V(G_{B'}) \setminus X')$ in $G_{B'}$.
		Furthermore, by the construction of the graph $H_{B'}$, this latter value is no larger than the value of the edge-cut $(X',V(H_{B'}) \setminus X')$ in $H_{B'}$, which is $w_{B'}(\partial^+_{H_{B'}}(X'))$.
		
		Next, we prove the second inequality from \Cref{cl:upper-bound-on-value-of-new-promised-cut:edge-version}.
		Indeed, since $X'$ is defined as $X' = A \cap X'_{B'}$, this inequality follows from \Cref{lem:submodularity-of-minimum-cuts:edge-version} (with $X'_{B'}$, $A$, and $A'$ in place of $X$, $S$, and $T$).
		
		Lastly, we now prove the third inequality from \Cref{cl:upper-bound-on-value-of-new-promised-cut:edge-version}.
		Recall we are assuming that event $\event$ occurred, which means that $\widehat{\OPT} \leq \OPT_{E,y^*}$;
		therefore, since $\vol^+_G(X'_{B'}) \leq 2\nu$, it follows from \Cref{obs:relating-cut-value-after-modification:edge-version} that $w_{B'}(\partial^+_{H_{B'}}(X'_{B'})) - w_{B'}(\partial^+_{G_{B'}}(X'_{B'})) \leq \frac{\epsilon\cdot\widehat{\OPT}}{z} \leq \frac{\epsilon\cdot\OPT_{E,y^*}}{z}$.
		So, by the choice of $(X'_{B'},Y'_{B'})$ as an edge-cut in $G_{B'}$ whose value is at most $\left(1 + \frac{(i-1)\epsilon}{z}\right)\OPT_{E,y^*}$, it follows that
		\begin{equation*}
			w_{B'}(\partial^+_{H_{B'}}(X'_{B'})) \leq w_{B'}(\partial^+_{G_{B'}}(X'_{B'})) + \frac{\epsilon\cdot\OPT_{E,y^*}}{z} \leq \left(1 + \frac{i\epsilon}{z}\right)\OPT_{E,y^*}.
		\end{equation*}
		This concludes the proof of \Cref{cl:upper-bound-on-value-of-new-promised-cut:edge-version}.
	\end{proof}

	We now continue the proof of \Cref{cl:A_B-satisfies-property-P'2:edge-version}.
	To prove this claim, it now remains to prove that $x^* \in X'$. As $X'= A \cap X'_{B'}$ and $x^* \in X'_{B'}$, it suffices to prove that $x^* \in A$.
	We do this next, by first proving that $\lambda_{H_{B'}}(x^*,y^*) < 4\widehat{\OPT}$, and then using \Cref{obs:minimum-cut-in-modified-graph-contains-distinguished-terminal}. Indeed since $x^* \in X'_{B'}$ and $y^* \notin X'_{B'}$, it follows by the third inequality of \Cref{cl:upper-bound-on-value-of-new-promised-cut:edge-version} that the $x^*$-$y^*$ edge-connectivity in $H_{B'}$ is
	\begin{equation}\label{eq:edge-connectivity-of-distinguished-terminal-in-modified-graph:edge-version}
		\lambda_{H_{B'}}(x^*,y^*) \leq w_{B'}(\partial^+_{H_{B'}}(X'_{B'})) \leq \left(1 + \frac{i\epsilon}{z}\right)\OPT_{E,y^*} \leq 2\OPT_{E,y^*}.
	\end{equation}
	Recall also that we are assuming that event $\event$ occurred, which, by the definition of this event, means that $\OPT_{E,y^*} < 2\widehat{\OPT}$.
	By plugging this inequality into \Cref{eq:edge-connectivity-of-distinguished-terminal-in-modified-graph:edge-version}, it follows that $\lambda_{H_{B'}}(x^*,y^*) < 4\widehat{\OPT}$.
	Since $x^* \in B \cap X'_{B'} \subseteq B \cap V(G_{B'})$, it now follows from \Cref{obs:minimum-cut-in-modified-graph-contains-distinguished-terminal} that $x^* \in A$.
	Thus, $x^* \in A \cap X'_{B'} = X'$, as we needed to prove.
	This concludes the proof of \Cref{cl:A_B-satisfies-property-P'2:edge-version}, and thus concludes the analysis of \approxsparsify.
\end{proof}

\section{From Rooted Minimum Vertex-Cut to Global Minimum Vertex-Cut: Proof of \Cref{thm:main-reduction}}\label{sec:main-reduction}
%In this section, we present our second main technical contribution: a randomized $\Tilde{O}(m)$-time approximation-preserving reduction from the Global Minimum Vertex-Cut problem on an $m$-edge graph $G$ to a collection of several instances of the Rooted Minimum Vertex-Cut problems, whose total size is $\Tilde{O}(m)$.

%\begin{theorem}
%	There exists a $\Tilde{O}(m)$-time randomized algorithm whose input is an instance $\iset$ of the vertex-weighted Global Minimum Vertex-Cut problem, consisting of a directed strongly-connected $m$-edge graph $G$ with non-negative weights $w(v)$ on the vertices $v \in V(G)$.
%	The output of the algorithm is a collection $\cset$ of instances of the vertex-weighted Rooted Minimum Vertex-Cut problem on directed graphs of total size $\Tilde{O}(m)$, such that all vertex-weights in these instances belong to the set $\{w(v) \mid v \in V(G)\}$ of vertex-weights that are used in $\iset$.
%	Additionally, there exists a deterministic procedure that, given any feasible solution to any instance in $\cset$, transforms this solution to a feasible solution of $\iset$ of the same value in time $\Tilde{O}(m)$.
%	Lastly, with probability at least $\frac{1}{2}$, there exists some instance $\iset' \in \cset$ whose optimal solution value is equal to that of $\iset$.
%\end{theorem}

In this section, we prove \Cref{thm:main-reduction} -- our reduction from the Global Minimum Vertex-Cut problem to the Rooted Minimum Vertex-Cut problem.

Consider an input to \Cref{thm:main-reduction}. That is, the input consists of a parameter $\epsilon \geq 0$, as well as a directed $m$-edge $n$-vertex graph $G$ with positive integer weights $w(v) \leq W$ on its vertices $v \in V(G)$, such that $G$ contains some vertex-cut.
Let $\OPT_{G}$ denote the weight of the global minimum vertex-cut in $G$.
Our goal is to design an algorithm implementing \Cref{thm:main-reduction}.
That is, upon receiving the aforementioned input, as well as access to a $(1+\epsilon)$-approximate Rooted Minimum Vertex-Cut oracle as specified in \Cref{thm:main-reduction}, the algorithm must output a vertex-cut $(L,S,R)$ such that, with probability at least $\frac{1}{2}$, the inequality $w(S) \leq (1+\epsilon)\OPT_{G}$ holds.
Furthermore, the algorithm must make at most $O\left(\frac{m}{n}\right)$ queries to the oracle, and the graphs in these queries must together contain at most $O(m)$ edges. Outside of these queries, the running time of the algorithm must be $\Tilde{O}(m \log W)$.

The algorithm that we design in the remainder of this section, which we refer to as \algglobalvertexcut, will actually only guarantee that the output vertex-cut $(L,S,R)$ satisfies $w(S) \leq (1+\epsilon)\OPT_{G}$ with probability $\Omega(1)$, rather than with probability $\frac{1}{2}$.
To boost the probability of the event $w(S) \leq (1+\epsilon)\OPT_{G}$, we run the algorithm $O(1)$ times with independent randomness, and select the lowest-weight cut among the obtained vertex-cuts.
In the design and analysis of the algorithm, we furthermore make the assumption that $G$ has some global minimum vertex-cut $(L^*,S^*,R^*)$ with $w(R^*) \geq w(L^*)$. This assumption can be removed using standard techniques, and we show how to do this in \Cref{sec:removing-assumption-in-reduction:reduction}.
%To implement \Cref{thm:main-reduction}, we then run \algglobalvertexcut $O(1)$ times, and select the lowest-weight cut among the returned vertex-cuts.

We are now ready to present the algorithm \algglobalvertexcut.
In the majority of this section, we consider the case in which every vertex of $G$ has at least one outgoing edge -- in the complementary case where there exists a vertex $v \in V(G)$ with no outgoing edges, the algorithm simply returns the vertex-cut $(\set{v},\emptyset,V(G)\setminus\set{v})$. Observe that this condition can be checked in time $O(m)$, and if there indeed exists a vertex $v$ with no outgoing edges, then the cut $(\set{v},\emptyset,V(G)\setminus\set{v})$ can be computed in time $O(n)$; furthermore, this cut must be a global minimum vertex-cut in $G$, because its weight is $w(\emptyset)=0$.
So, from now on, we assume that every vertex of $G$ has at least one outgoing edge.
%
%The algorithm begins by checking if there exists some vertex $v \in V(G)$ with no outgoing edges in $G$ -- if yes, then the algorithm returns the vertex-cut $(\set{v},\emptyset,V(G)\setminus\set{v})$, which must be a global minimum vertex-cut because its weight is $w(\emptyset)=0$.
%So, from now on, we consider only the case in which every vertex $v \in V(G)$ has at least one outgoing edge in $G$.
In this case, the algorithm \algglobalvertexcut consists of four steps, described next.
In the following parts of this section, we discuss the implementation of each step.
%When presenting this algorithm, we will make the (unwarranted) assumption that every vertex of $G$ has out-degree at least $1$, as well as that $G$ has some global minimum vertex-cut $(L^*,S^*,R^*)$ with $w(R^*) \geq w(L^*)$, and we will fix one such cut $(L^*,S^*,R^*)$ throughout the proof. We show how to get rid of these assumptions in \Cref{sec:remocing-the-assumptions:reduction}.
%Additionally, for the sake of the proof, we fix one global minimum vertex-cut $(L^*,S^*,R^*)$ with $w(R^*) \geq w(L^*)$.
%We indeed make these assumptions, and, for the sake of the proof, we fix one such vertex-cut $(L^*,S^*,R^*)$.
%The algorithm \algglobalvertexcut consists of four steps, described next.
%In the following parts of this section, we discuss the implementation of each step.

\begin{enumerate}
	\item\label{step:selecting-roots:reduction} select a set $T \subseteq V(G)$ of "roots", such that, with probability at least $\frac{1}{4}$, there exists some vertex of $T$ that also belongs to $R^*$. Furthermore, for every root $y \in T$, there must exist some feasible solution to the Rooted Minimum Vertex-Cut problem on $G$ with root $y$;
	\item\label{step:constructing-sparsified-graphs:reduction} construct, for every root $y \in T$, a subgraph $G_y$ of $G$ with $V(G_y)=V(G)$, such that any feasible solution to the Rooted Minimum Vertex-Cut problem on graph $G_y$ with root $y$ is also a feasible solution to the Rooted Minimum Vertex-Cut problem on $G$ with root $y$, and vice versa;
	\item\label{step:applying-rooted-algorithm:reduction} for each root $y \in T$, query the $(1+\epsilon)$-approximate Rooted Minimum Vertex-Cut oracle on the graph $G_{y}$ with root $y$, and let $(L_y,S_y,R_y)$ denote the vertex-cut returned by this query; lastly,
	\item\label{step:output:reduction} output the minimum-weight vertex-cut among the obtained feasible solutions $\{(L_y,S_y,R_y)\}_{y \in T}$. If the set $T$ is empty, then instead find and output any vertex-cut in $G$.
\end{enumerate}

\noindent It is not hard to see that the output of \algglobalvertexcut satisfies the required guarantees. We formally prove this in the next claim.

\begin{claim}\label{cl:output-is-correct:reduction}
	The output $(L,S,R)$ of the above algorithm is a vertex-cut in $G$, and, with probability at least $\frac{1}{8}$, its weight is $w(S) \leq (1+\epsilon)\OPT_{G}$.
\end{claim}
\begin{proof}
	We begin by arguing that $(L,S,R)$ is indeed a vertex-cut in $G$.
	Indeed, by the description of steps \ref{step:constructing-sparsified-graphs:reduction} and \ref{step:applying-rooted-algorithm:reduction}, it can be seen that for every root $y \in T$, the triplet $(L_y,S_y,R_y)$ is a feasible solution for the Rooted Minimum Vertex-Cut problem on graph $G$ with root $y$, which means in particular that it is a vertex-cut in $G$.
	Thus, by the description of Step \ref{step:output:reduction}, the output of the algorithm must be a vertex-cut in $G$.
	Next, we show that $w(S) \leq (1+\epsilon)\OPT_{G}$ holds with probability at least $\frac{1}{8}$.
	By the description of Step \ref{step:output:reduction}, it suffices to show that with probability at least $\frac{1}{8}$, there exists some $y \in T$ for which $w(S_y) \leq (1+\epsilon)\OPT_{G}$.
	Indeed, by the description of Step \ref{step:selecting-roots:reduction}, with probability at least $\frac{1}{4}$, there exists some $y^* \in T$ that also belongs to $R^*$, and, by the description of Step \ref{step:constructing-sparsified-graphs:reduction}, the global minimum vertex-cut $(L^*,S^*,R^*)$ in $G$ is a feasible solution to the Rooted Minimum Vertex-Cut problem on $G_{y^*}$ with root $y^*$. Furthermore, conditioned on the event that such a vertex $y^*$ exists, there is a probability of $\frac{1}{2}$ that the corresponding query to the Rooted Minimum Vertex-Cut oracle outputs a vertex-cut $(L_{y^*},S_{y^*},R_{y^*})$ with $w(S_{y^*}) \leq (1+\epsilon)\cdot w(S^*)$.
	Therefore, with probability at least $\frac{1}{4} \cdot \frac{1}{2}=\frac{1}{8}$, there exists such $y^* \in T$ and the corresponding vertex-cut $(L_{y^*},S_{y^*},R_{y^*})$ has weight $w(S_{y^*}) \leq (1+\epsilon)\cdot w(S^*) = (1+\epsilon)\OPT_{G}$.
%	Thus
%	This concludes the proof of \Cref{cl:output-is-correct:reduction}.
%	
%	
%	Therefore, since each vertex-cut $(L_y,S_y,R_y)$ is generated by a query to the $(1+\epsilon)$-approximate Rooted Minimum Vertex-Cut oracle (see \Cref{def:oracle}) on a Rooted Minimum Vertex-Cut instance with the corresponding graph $G_y$ and root $y$, it suffices to show that with probability at least $\frac{1}{4}$, there exists some $y \in T$ for which this corresponding instance has a solution of value at most $\OPT_{G}$.
%	Indeed, by the description of step \ref{step:selecting-roots:reduction}, with probability at least $\frac{1}{4}$, there exists some $y \in T$ that also belongs to $R^*$, and, by the description of step \ref{step:constructing-sparsified-graphs:reduction}, the global minimum vertex-cut $(L^*,S^*,R^*)$ in $G$ is a feasible solution to the Rooted Minimum Vertex-Cut problem on $G_y$ with root $y$.
	This concludes the proof of \Cref{cl:output-is-correct:reduction}.
\end{proof}

To complete the proof of \Cref{thm:main-reduction}, it now remains to describe how to implement each step of \algglobalvertexcut in time $\Tilde{O}(m \log W)$, while also guaranteeing
% that the algorithm makes at most $O\left(\frac{m}{n}\right)$ queries to the oracle, and the graphs in these queries have at most $O(m)$ edges overall.
 that the total number of edges in the graphs $\set{G_y}_{y \in T}$ is $\sum_{y \in T}|E(G_y)| = O(m)$, and that there are at most $O\left(\frac{m}{n}\right)$ such graphs.
%Observe that the algorithm makes exactly one query to the oracle for each graph $\set{G_y}_{y \in T}$.
%Thus, to guarantee the required conditions about the queries made by the algorithm, it suffices to guarantee that $\sum_{y \in T}|E(G_y)| = O(m)$ and that each graph $G_y$ has $\Omega(n)$ edges -- indeed, these conditions together imply that there are at most $O\left(\frac{m}{n}\right)$ such graphs, and the former condition implies the required bound on the total number of edges in the queries.
In order to do so, we first need to impose some additional constraints on the set $T$ of roots selected in Step \ref{step:selecting-roots:reduction}, and on the graphs $\set{G_y}_{y \in T}$ generated in Step \ref{step:constructing-sparsified-graphs:reduction}.
For this purpose, it is convenient to use the following notation of "far-away sets".

\begin{definition}[Far-Away Sets]\label{def:far-away-sets}
	For every vertex $v \in V(G)$, we define the "forward far-away set of $v$", denoted $F^+(v)$, to be the set $F^+(v) = V(G) \setminus (\set{v} \cup N^+_G(v))$ of all vertices that cannot be reached from $v$ by traversing at most $1$ edge.
	Similarly, define the "backward far-away set of $v$", denoted $F^-(v)$, to be the set $F^-(v) = V(G) \setminus (\set{v} \cup N^-_G(v))$ of all vertices that cannot reach $v$ by traversing at most $1$ edge.
\end{definition}

We are now ready to define the additional constraints that we will have on the output of steps \ref{step:selecting-roots:reduction} and \ref{step:constructing-sparsified-graphs:reduction} of \algglobalvertexcut.
Specifically, in our implementation of the algorithm, we will guarantee that the set $T$ of roots selected in Step \ref{step:selecting-roots:reduction} satisfies $\sum_{y \in T} \vol^+(F^-(y)) = O(m)$; as well as that each graph $G_y$ generated in Step \ref{step:constructing-sparsified-graphs:reduction} satisfies $|E(G_y)| = \vol^+(F^-(y)) + \deg^-_G(y)$.
The next two claims state that steps \ref{step:selecting-roots:reduction} and \ref{step:constructing-sparsified-graphs:reduction} can be implemented in the correct running time, while satisfying these additional constraints.

\begin{claim}\label{cl:implementing-step-selecting-roots:reduction}
	Step \ref{step:selecting-roots:reduction} of \algglobalvertexcut can be implemented in time $\Tilde{O}(m \log W)$, while also guaranteeing that the selected set $T$ satisfies $\sum_{y \in T} \vol^+(F^-(y)) = O(m)$.
\end{claim}

\begin{claim}\label{cl:implementing-step-constructing-sparsified-graphs:reduction}
	Step \ref{step:constructing-sparsified-graphs:reduction} of \algglobalvertexcut can be implemented so that the constructed graph $G_y$ for each $y \in T$ satisfies $|E(G_y)| = \vol^+(F^-(y)) + \deg^-_G(y)$, and the time required for constructing each graph $G_y$ is $\Tilde{O}(|E(G_y)| + n)$.
\end{claim}

We defer the proof of \Cref{cl:implementing-step-selecting-roots:reduction} to \Cref{sec:implementing-step-selecting-roots:reduction}, and the proof of \Cref{cl:implementing-step-constructing-sparsified-graphs:reduction} to \Cref{sec:implementing-step-constructing-sparsified-graphs:reduction}.
We now complete the proof of \Cref{thm:main-reduction} using the above two claims.

We now specify our exact implementation of \algglobalvertexcut: we implement steps \ref{step:selecting-roots:reduction} and \ref{step:constructing-sparsified-graphs:reduction} using \Cref{cl:implementing-step-selecting-roots:reduction} and \Cref{cl:implementing-step-constructing-sparsified-graphs:reduction}; steps \ref{step:applying-rooted-algorithm:reduction} and \ref{step:output:reduction} can then be implemented in a straightforward manner in time $\Tilde{O}(|T| \cdot n)$, and we indeed implement them as such.

To complete the proof of \Cref{thm:main-reduction}, it now remains to prove that this implementation of \algglobalvertexcut guarantees that the graphs $\set{G_y}_{y \in T}$ on which we query the oracle have at most $O(m)$ edges overall; and that there are at most $O\left(\frac{m}{n}\right)$ such graphs; as well as that the whole implementation runs in time $\Tilde{O}(m \log W)$.
The following claim shows that the total number of edges in these graphs is indeed $O(m)$.

\begin{claim}\label{cl:total-size-of-sparsified-graphs:reduction}
	After implementing steps \ref{step:selecting-roots:reduction} and \ref{step:constructing-sparsified-graphs:reduction} as above, the total number of edges $\sum_{y \in T} |E(G_y)|$ in the sparsified graphs $\set{G_y}_{y \in T}$ is $O(m)$.
\end{claim}
\begin{proof}
	Indeed, by the guarantees of \cref{cl:implementing-step-constructing-sparsified-graphs:reduction}, the number of edges in each sparsified graph $G_y$ is $\deg^-_G(y) + \vol^+(F^-(y))$.
	Thus, the total number of edges in all sparsified graphs is
	$$\sum_{y \in T} \left(\deg^-_G(y) + \vol^+_G(F^-(y)\right) \leq \left(\sum_{y \in V(G)} \deg^-_G(y)\right) + \left(\sum_{y \in T} \vol^+_G(F^-(y)\right) = m + \sum_{y \in T} \vol^+_G(F^-(y)).$$
	Since our implementation of Step \ref{step:selecting-roots:reduction} guarantees that $\sum_{y \in T} \vol^+(F^-(y)) = O(m)$, it then follows that $\sum_{y \in T} |E(G_y)| = O(m)$.
%	Now, to show that $\sum_{y \in T} (|V(G)| + |E(G_y)|) = O(m)$, it suffices to show that $|E(G_y)| \geq |V(G)|$ holds for every $y \in T$: indeed recall that we are assuming that every vertex of $G$ has at least one outgoing edge (see \rnote{!}); therefore, for every root $y \in T$ and for every vertex $u \in V(G)$, our construction of $E(G_y)$ includes at least one edge of $\delta^+_G(u)$.
\end{proof}

Next, in the following claim, we prove that there are at most $O\left(\frac{m}{n}\right)$ graphs in the set $\set{G_y}_{y \in T}$.
\begin{claim}\label{cl:total-number-of-sparsified-graphs-is-small:reduction}
	The size of the set $|T|$ is $|T| = O\left(\frac{m}{n}\right)$.
\end{claim}
\begin{proof}
	Since \Cref{cl:total-size-of-sparsified-graphs:reduction} guarantees that $\sum_{y \in T} |E(G_y)| = O(m)$, it suffices if we prove that each graph $G_y$ contains $\Omega(n)$ edges; we indeed prove this next:
	consider some $y \in T$; by \Cref{cl:implementing-step-constructing-sparsified-graphs:reduction}, the number of edges in the graph $G_y$ is exactly $\vol^+(F^-(y)) + \deg^-_G(y)$. It remains to prove that $\vol^+(F^-(y)) + \deg^-_G(y)=\Omega(n)$. Recall that we are assuming each vertex of $G$ has at least one outgoing edge, implying that $\vol^+(F^-(y)) \geq |F^-(y)|$. Furthermore, by the definition of $F^-(y)$, the size of this set is $|F^-(y)| = n - 1- \deg^-_G(y)$, as it contains exactly those vertices of $G$ that are not in the set $N^-_G(y)\cup\set{y}$.
	To summarize, $\vol^+(F^-(y)) + \deg^-_G(y) \geq |F^-(y)| + \deg^-_G(y) \geq n-1 = \Omega(n)$, as we needed.
	This concludes the proof of \Cref{cl:total-number-of-sparsified-graphs-is-small:reduction}.
\end{proof}

Lastly, we must show that the running time of \algglobalvertexcut is $\Tilde{O}(m \log W)$. Our implementation of Step \ref{step:selecting-roots:reduction} has the correct running time based on the statement of \Cref{cl:implementing-step-selecting-roots:reduction}. By the statement of \Cref{cl:implementing-step-constructing-sparsified-graphs:reduction}, and by our implementation of steps \ref{step:applying-rooted-algorithm:reduction} and \ref{step:output:reduction}, the total time required for steps \ref{step:constructing-sparsified-graphs:reduction}-\ref{step:output:reduction} is $\Tilde{O}\left(|T| \cdot n + \sum_{y \in T} |E(G_y)|\right)$ -- based on \Cref{cl:total-size-of-sparsified-graphs:reduction} and \Cref{cl:total-number-of-sparsified-graphs-is-small:reduction}, this running time is $\Tilde{O}(m)$.
This concludes the analysis of the running time.
%Recall that, to complete the proof of the theorem, it remains to explain how we implement the four steps of \algglobalvertexcut, as well as to prove that our implementation runs in time $\Tilde{O}(m \log W)$, and makes at most $O\left(\frac{m}{n}\right)$ oracle queries, such that the graphs in these queries contain at most $O(m)$ edges in total.
%
%To implement steps \ref{step:selecting-roots:reduction} and \ref{step:constructing-sparsified-graphs:reduction} of \algglobalvertexcut, we use \Cref{cl:implementing-step-selecting-roots:reduction} and \Cref{cl:implementing-step-constructing-sparsified-graphs:reduction}.
%\rnote{Continue from here}
%Recall that, to complete the proof, it remains to show that \algglobalvertexcut can be implemented in time $\Tilde{O}(m \cdot \log W)$.
%To do this, we implement Step \ref{step:selecting-roots:reduction} using \Cref{cl:implementing-step-selecting-roots:reduction}, which guarantees that the obtained set $T$ satisfies the constraint $\sum_{y \in T} \vol^+(F^-(y)) = O(m)$, and that the running time of this step is sufficiently small.
%We then implement Step \ref{step:constructing-sparsified-graphs:reduction} as follows.
Now, to complete the proof of \Cref{thm:main-reduction}, it remains only to prove \Cref{cl:implementing-step-constructing-sparsified-graphs:reduction} and \Cref{cl:implementing-step-selecting-roots:reduction}. We do this in the next two subsections. (We choose to prove \Cref{cl:implementing-step-constructing-sparsified-graphs:reduction} first, as its proof is shorter.)

\subsection{Implementing Step \ref{step:constructing-sparsified-graphs:reduction}: Proof of \Cref{cl:implementing-step-constructing-sparsified-graphs:reduction}} \label{sec:implementing-step-constructing-sparsified-graphs:reduction}

To implement Step \ref{step:constructing-sparsified-graphs:reduction} of \algglobalvertexcut as specified by \Cref{cl:implementing-step-constructing-sparsified-graphs:reduction}, we construct each graph $G_y$ so that it is equal to the graph obtained from $G$ by removing all edges outgoing from vertices of $N^-_G(y) \cup \set{y}$, except those edges ending at $y$.
In other words, $V(G_y)=V(G)$ and $E(G_y)=\delta^-_G(y) \cup E_G(F^-(y),V(G))$.
Then, it is not hard to verify that $|E(G_y)| = \vol^+_G(F^-(y)) + \deg^-_G(y)$ holds, as required by \Cref{cl:implementing-step-constructing-sparsified-graphs:reduction}. (Indeed, recall that the sets of edges $\delta^-_G(y)$ and $E_G(F^-(y),V(G))$ are disjoint, based on the definition of $F^-(y)$.)
%: indeed, the set $\delta^-_G(y)$ contains exactly $\deg^-_G(y)$ edges, and the set $E_G(F^-(y),V(G))$ contains exactly $\vol^+_G(F^-(y))$ edges; it is also easy to verify that these two sets are disjoint, based on the definition of $F^-(y)$ (\Cref{def:far-away-sets}).
Furthermore, it is also not hard to construct each graph $G_y$ in time $\Tilde{O}(|E(G_y)| + n)$, as required in \Cref{cl:implementing-step-constructing-sparsified-graphs:reduction}: the set $F^-(y)$ can be computed in time $O(n)$, and then the set of edges in $G_y$ can be enumerated by going over the lists of outgoing edges of all vertices in $F^-(y)$, and the list of incoming edges of $y$.

Next, to complete the proof of \Cref{cl:implementing-step-constructing-sparsified-graphs:reduction}, we show that constructing the graphs $G_y$ in this manner indeed satisfies the requirements of Step \ref{step:constructing-sparsified-graphs:reduction}.
Specifically, we must show that any feasible solution to the Rooted Minimum Vertex-Cut problem on graph $G_y$ with root $y$ is also a feasible solution to the Rooted Minimum Vertex-Cut problem on graph $G$ with root $y$, and vice versa.
We prove the first direction in the next claim.
\begin{claim}
	any feasible solution to the Rooted Minimum Vertex-Cut problem on $G_y$ with root $y$ must also be a feasible solution to the Rooted Minimum Vertex-Cut problem on $G$ with root $y$.
\end{claim}
\begin{proof}
	Consider any feasible solution $(L,S,R)$ to the Rooted Minimum Vertex-Cut problem on $G_y$ with root $y$. We need to show that $(L,S,R)$ is also a feasible solution to the Rooted Minimum Vertex-Cut problem on $G$ with root $y$. Specifically, as $y \in R$ clearly holds, it remains to show that $(L,S,R)$ is valid vertex-cut in the graph $G$; that is, it remains to show that $N^+_G(L) \subseteq S$.
	Since $(L,S,R)$ is a valid vertex-cut in $G_y$, the graph $G_y$ must satisfy $N^+_{G_y}(L) \subseteq S$, so it suffices to prove that $N^+_{G_y}(L) = N^+_G(L)$, which is what we will do.
	Indeed, recall that for every vertex $v \in F^-(y)$, the set of edges outgoing from $v$ in $G_y$ is exactly the same as the set of edges outgoing from $v$ in $G$.
	Therefore, it suffices to show that $L \subseteq F^-(y)$: by the construction of $G_y$, every vertex outside of $F^-(y)$ belongs to $N^-_{G_y}(y)$, so it suffices to show that $L \cap N^-_{G_y}(y) = \emptyset$; indeed, this follows directly from the fact that $(L,S,R)$ is a feasible solution to the Rooted Minimum Vertex-Cut problem on $G_y$ with root $y$.
%	
%	Then, $y \in R$ must hold, and so $N^-_{G_y}(y) \subseteq S \cup R$; since every vertex outside of $F^-(y)$ belongs to $N^-_{G_y}(y)$, this implies that $L \subseteq F^-(y)$.
%	In order to show that $(L,S,R)$ is also a feasible solution to the Rooted Minimum Vertex-Cut problem on $G$ with root $y$, we must show that the graph $G$ satisfies $N^+_G(L) \subseteq S$:
%	indeed, by the construction of the graph $G_y$, the set of outgoing edges from each vertex of $F^-(y)$ in $G_y$ is the same as in $G$; so, as $L \subseteq F^-(y)$, it follows that the set of outgoing edges from $L$ in $G_y$ is the same as in $G$; meaning that $N^+_G(L) = N^+_{G_y}(L) \subseteq S$.
\end{proof}

For the "vice versa" direction, we need to prove that any feasible solution to the Rooted Minimum Vertex-Cut problem on $G$ with root $y$ must also be a feasible solution to the Rooted Minimum Vertex-Cut problem on $G_y$ with root $y$. However, this holds trivially as $E(G_y) \subseteq E(G)$: indeed, consider any vertex-cut $(L,S,R)$ in $G$, meaning that no edge of $G$ goes from $L$ to $R$; then, as $E(G_y) \subseteq E(G)$, it follows that no edge of $G_y$ goes from $L$ to $R$, so $(L,S,R)$ is also a vertex-cut of $G_y$.
This concludes the proof of \Cref{cl:implementing-step-constructing-sparsified-graphs:reduction}.

\subsection{Implementing Step \ref{step:selecting-roots:reduction}: Proof of \Cref{cl:implementing-step-selecting-roots:reduction}}\label{sec:implementing-step-selecting-roots:reduction}
In this section, we implement Step \ref{step:selecting-roots:reduction} of \algglobalvertexcut, with the added constraint imposed in \Cref{cl:implementing-step-selecting-roots:reduction}.
Recall that, the goal of Step \ref{step:selecting-roots:reduction} of \algglobalvertexcut is to construct a set $T \subseteq V(G)$ of vertices such that with probability at least $\frac{1}{4}$, at least one vertex of $T$ belongs to $R^*$, and furthermore, for every $y \in T$, the Rooted Minimum Vertex-Cut problem on $G$ with root $y$ must have a feasible solution.
In order to satisfy the latter condition, we will guarantee that every vertex $y \in T$ has a non-empty backward far-away set $F^-(y)$, as this guarantees that $(F^-(y),N^-(y),\set{y})$ is a feasible solution to the Rooted Minimum Vertex-Cut problem on $G$ with root $y$.
The main challenge lies in making sure that $T$ satisfies the additional constraint imposed by \Cref{cl:implementing-step-selecting-roots:reduction}. That is, in making sure that $\sum_{y \in T}\vol^+(F^-(y)) = O(m)$.

The procedure that we describe next for constructing $T$ will only guarantee that $\sum_{y \in T}\vol^+(F^-(y)) = O(m)$ holds in expectation. More formally, it will guarantee that $\expect{\sum_{y \in T}\vol^+(F^-(y))} \leq 2m$ holds, and that $\Pr\left[T \cap R^* \neq \emptyset\right] \geq \frac{1}{2}$.
To transform the obtained set $T$ into one that is guaranteed to satisfy $\sum_{y \in T}\vol^+(F^-(y)) = O(m)$, we can simply remove all elements from $T$ in the case where the quantity $\sum_{y \in T}\vol^+(F^-(y))$ exceeds $8m$.
By Markov's inequality, this only occurs with probability at most $\frac{1}{4}$, so the probability of the event $T \cap R^* \neq \emptyset$ remains at least $\frac{1}{4}$.
Furthermore, observe that the value of $\sum_{y \in T}\vol^+(F^-(y))$ can be computed in $O(m)$ time, as the value of $\vol^+(F^-(y))$ can be computed for each vertex $y$ in time $O(\deg(y))$ using the equality $\vol^+(F^-(y)) = m - \sum_{v \in \set{y} \cup N^-_G(y)}\deg^+(v)$.

Next, we describe how to construct the set $T$ such that $\expect{\sum_{y \in T}\vol^+(F^-(y))} \leq 2m$.
To do this, we guarantee that each vertex $v \in V(G)$ belongs in expectation to at most two of the sets $\{F^-(y)\}_{y \in T}$.
Specifically, observe that for every two vertices $v,y \in V(G)$, the vertex $v$ belongs to $F^-(y)$ if and only if $y$ belongs to the set $F^+(v) = V(G) \setminus \left(\set{v} \cup N^+_G(v)\right)$.
Thus, to guarantee that $\expect{\sum_{y \in T}\vol^+(F^-(y))} \leq 2m$, it suffices to choose $T$ such that for every $v \in V(G)$, the expected size of the intersection $T \cap F^+(v)$ is at most $2$.
Our algorithm for constructing $T$ works as follows:

\begin{enumerate}[label=1.{\alph*.}]
	\item calculate the value $w^* = \max_{v \in V(G)} w(F^+(v))$;
	\item\label{step:selecting-vertices-into-T:reduction} construct $T$ by selecting each vertex $y \in V(G)$ into $T$ independently with probability $\min\left\{1,\frac{2w(y)}{w^*}\right\}$ (note that $w^* > 0$ must hold\footnote{Indeed, every vertex $v \in L^*$ must have $R^* \subseteq F^+(v)$ and thus $F^+(v) \neq \emptyset$, which means that $w(F^+(v)) > 0$ since we are assuming that vertex-weights are positive.}); and lastly,
	\item\label{step:removing-vertices-from-T:reduction} remove from $T$ any vertex $y$ with $F^-(y) = \emptyset$.
\end{enumerate}
Based on the description of this algorithm, it is easy to see that for every vertex $v \in V(G)$, the inequality $\expect{\left|T \cap F^+(v)\right|} \leq \frac{2w(F^+(v))}{w^*} \leq 2$ holds.
By the above discussion, this means that $\expect{\sum_{y \in T}\vol^+(F^-(y))} \leq 2m$.

Next, we describe how to implement the above algorithm efficiently, and then prove that $\Pr\left[T \cap R^* \neq \emptyset\right] \geq \frac{1}{2}$.
Indeed, to calculate $w^*$ in time $O(m)$, it suffices to calculate $w(F^+(v))$ for each $v \in V(G)$ in time $O(\deg(v))$, which we can do using the equality $w(F^+(v)) = w(V(G)) - \sum_{u \in \set{v} \cup N^+(v)}w(u)$.
Step \ref{step:selecting-vertices-into-T:reduction} can be implemented using standard techniques for generating Bernoulli random variables.
Finally, to implement Step \ref{step:removing-vertices-from-T:reduction}, we can check for each vertex $y$ whether $F^-(y) = \emptyset$ in time $O(1)$ by using the equality $|F^-(y)| = n - 1 - \deg^-_G(y)$, which takes $O(n) \leq O(m)$ time overall. (Observe that $n \leq m$ holds by the assumption that every vertex has out-degree at least $1$.)

In order to complete the description of step \ref{step:selecting-roots:reduction}, it remains only to prove that the set $T$ obtained by the above procedure satisfies $\Pr\left[T \cap R^* \neq \emptyset\right] \geq \frac{1}{2}$.
Since vertices are selected into $T$ with probability proportional to their weight, it is useful to first prove a lower-bound on $w(R^*)$, which we do in the following claim.

\begin{claim}\label{cl:relation-of-w^*-to-mincut:reduction}
	$w(R^*) \geq \frac{w^*}{2}$.
\end{claim}
\begin{proof}
	Recall that $w^*$ is defined as $w^* = \max_{v \in V(G)} w(F^+(v))$, and consider the vertex $v$ that attains this maximum.
	Since $w(F^+(v)) = w^* > 0$, it follows that $F^+(v) \neq \emptyset$, and thus the triplet $(\set{v}, N^+(v), F^+(v))$ is a valid vertex-cut in $G$.
	Since $(L^*,S^*,R^*)$ was chosen to be a global minimum vertex-cut in $G$, it follows that $w(S^*) \leq w(N^+(v))$, and thus that $w(L^* \cup R^*) \geq w(V(G)) - w(N^+(v)) = w(\set{v} \cup F^+(v)) \geq w(F^+(v)) = w^*$.
	Since we chose $(L^*,S^*,R^*)$ to be a global minimum vertex-cut satisfying $w(L^*) \leq w(R^*)$, it follows that $w(R^*) \geq \frac{w(L^* \cup R^*)}{2} \geq \frac{w^*}{2}$, as we needed to prove.
	This concludes the proof of \Cref{cl:relation-of-w^*-to-mincut:reduction}.
\end{proof}

\noindent We are now ready to prove that $\Pr\left[T \cap R^* \neq \emptyset\right] \geq \frac{1}{2}$.

\begin{claim}\label{cl:some-terminal-belongs-to-R^*:reduction}
	When choosing $T$ according to steps \ref{step:selecting-vertices-into-T:reduction}-\ref{step:removing-vertices-from-T:reduction}, the inequality $\Pr\left[T \cap R^* \neq \emptyset\right] \geq \frac{1}{2}$ holds.
\end{claim}
\begin{proof}
	Firstly, observe that if any vertex $y \in R^*$ is selected into $T$ during Step \ref{step:selecting-vertices-into-T:reduction} of the procedure, then it will not be removed during Step \ref{step:removing-vertices-from-T:reduction}, since any such vertex satisfies $L^* \subseteq F^-(y)$.
	Thus, it suffices to show that with probability at least $\frac{1}{2}$, some vertex of $R^*$ is selected into $T$ in Step \ref{step:selecting-vertices-into-T:reduction}.
	In other words, if we let $\event$ denote the bad event in which no vertex of $R^*$ is selected into $T$ in Step \ref{step:selecting-vertices-into-T:reduction}, then it suffices to show that $\Pr[\event] \leq \frac{1}{2}$.
	We do this next.
	Indeed, for every vertex $y \in R^*$, let $p_y = \min\left\{1,\frac{2w(y)}{w^*}\right\}$ denote the probability of $y$ being selected into $T$.
	Since each vertex $y$ is selected into $T$ independently with probability $p_y$, we get that
	\[
	 \Pr[\event] = \prod_{y \in R^*} (1-p_y) \leq \prod_{y \in R^*} e^{-p^y} = \operatorname{exp}\left(-\sum_{y \in R^*}p_y\right).
	\]
	Since $\exp(-1) \leq \frac{1}{2}$, it now suffices to prove that $\sum_{y \in R^*}p_y \geq 1$, which we do next.
	Consider two cases: the first case is where every $y \in R^*$ satisfies $p_y = \frac{2w(y)}{w^*}$, and in this case $\sum_{y \in R^*}p_y = \frac{2w(R^*)}{w^*} \geq 1$ holds by \Cref{cl:relation-of-w^*-to-mincut:reduction};
	the second case is where some $y' \in R^*$ has $p_{y'} \neq \frac{2w(y)}{w^*}$ -- in this case, since $p_{y'}$ is defined as $p_{y'}=\min\left\{1,\frac{2w(y)}{w^*}\right\}$, it follows that $p_{y'} =1$ and thus that $\sum_{y \in R^*}p_y \geq 1$, as needed.
	This concludes the proof of
%	
%	
%	Consider two cases: the first case is where every $y \in R^*$ satisfies $p_y = \frac{2w(y)}{w^*}$, and in this case $\sum_{y \in R^*}p_y = \frac{2w(R^*)}{w^*} \geq 1$ holds by \Cref{cl:relation-of-w^*-to-mincut:reduction};
%	the second case is where some $y' \in R^*$ has $p_{y'} \neq \frac{2w(y)}{w^*}$, which means that $p_{y'} =1$, and thus $\sum_{y \in R^*}p_y \geq 1$ holds because all $p_y$ are non-negative.\rnote{Can rephrase this devision into cases so that it is not one sentence.}
%	This concludes the proof of
	 \Cref{cl:some-terminal-belongs-to-R^*:reduction}, and thus concludes the proof of \Cref{cl:implementing-step-selecting-roots:reduction}.
\end{proof}

%\subsection{Getting Rid of the Assumption $w(R^*) \geq w(L^*)$}\label{sec:remocing-the-assumptions:reduction}
%Recall that our goal in \Cref{sec:main-reduction} is to design an algorithm as promised by \Cref{thm:main-reduction}.
%In the preceding parts of this section, we designed an algorithm \algglobalvertexcut that does as promised in \Cref{thm:main-reduction}, under the additional assumption that the input graph $G$ has a global minimum vertex-cut $(L^*,S^*,R^*)$ with $w(R^*) \geq w(L^*)$.
%In the following section, we show how to get rid of this assumption.
%
%
%
%\paragraph{Proof of \Cref{thm:main-reduction}}
%To implement the algorithm promised in \Cref{thm:main-reduction}, we run both of the algorithms \algglobalvertexcut' and \algglobalvertexcut''.
%Consider any global minimum vertex-cut $(L^*,S^*,R^*)$ of the input graph $G$.
%Since at least one of the assumptions $w(R^*) \geq w(L^*)$ and $w(R^*) \leq w(L^*)$ must hold, it follows that at least one of the algorithms \algglobalvertexcut' and \algglobalvertexcut'' is guaranteed to run in time $\Tilde{O}(m \log W)$ while making at most $O\left(\frac{m}{n}\right)$ oracle queries,
%
%\rnote{Write this}
%
%\subsubsection{The algorithm \algglobalvertexcut'}\label{sec:algorithms-under-both-sides-of-assumption:reduction}

\newpage

\appendix
\section{Algorithm for Rooted Minimum Vertex-Cut: Proof of \Cref{thm : main-thm-rooted-vertex-version}}\label{sec:rooted-minimum-vertex-cut}

In this section, we describe and analyze the algorithm implementing \Cref{thm : main-thm-rooted-vertex-version} -- our main result for the Rooted Minimum Vertex-Cut problem.
This algorithm is a fairly straightforward adaptation of the one from \Cref{sec:rooted-minimum-edge-cut} to the vertex-cut setting -- indeed, the only non-trivial task is in adapting the \approxsparsify subroutine, which we do in \Cref{sec:description-of-iteration-of-main-algorithm:vertex-version}.
Recall that the input to the algorithm consists of a directed $m$-edge graph $G$ with positive integer weights $w(v) \leq W$ on the vertices $v \in V(G)$, as well as a precision parameter $\epsilon \in (0,1)$, and a root vertex $y^* \in V(G)$, such that there exists some vertex-cut $(L',S',R')$ in $G$ with $y^* \in R'$.
Let $\OPT_{V,y^*}$ denote the minimum value of a vertex-cut among all vertex-cuts $(L',S',R')$ of $G$ with $y^* \in R'$.
Furthermore, for the sake of the analysis, fix one such vertex-cut achieving this minimum value as the \emph{distinguished cut}, and denote it by $(L^*,S^*,R^*)$.
Our goal is to output a vertex-cut $(L,S,R)$ of $G$ with $w(S) \leq (1+\epsilon) \OPT_{V,y^*}$.

\paragraph{Assumption that $\OPT_{V,y^*} > 0$.}
As in \Cref{sec:rooted-minimum-edge-cut}, it will be convenient for us to assume that the value of the optimal solution is positive, and we indeed make this assumption in the remainder of this section.
This assumption can be removed via the following standard technique: before running the algorithm, we check whether $\OPT_{V,y^*} > 0$ by performing a BFS from $y^*$ in the reverse graph $\overline{G}$ of $G$; if the BFS reaches all vertices of $\overline{G}$, then $\OPT_{V,y^*} > 0$ must hold, in which case we run the algorithm; otherwise, we can use the set of vertices reached by the BFS in order to construct and output a vertex-cut $(L,S,R)$ of $G$ with $w(S) = 0$ and $y^* \in R$ -- specifically, we let $S = \emptyset$, let $R$ be the set of vertices reached by the BFS, and let $L = V(G) \setminus R$.

\paragraph{Estimates of $\OPT_{V,y^*}$ and $\vol^+_G(L^*)$, and a Set of Terminals.}
%Our algorithm begins by guessing an estimate $\widehat{\OPT}$ of $\OPT_{V,y^*}$ -- specifically, it picks $\widehat{\OPT}$ uniformly at random among the integral powers of $2$ in the range $[1,m \cdot W]$.
%Let $\event^{\widehat{\OPT}}$ denote the good event in which $\frac{\OPT_{V,y^*}}{2} < \widehat{\OPT} \leq \OPT_{V,y^*}$.
%Observe that $\Pr[\event^{\widehat{\OPT}}] \geq \Omega\left(\frac{1}{\log(m \cdot W)}\right)$.
%Next, our algorithm uses \Cref{lem: picking terminals}
%
%
%
Our algorithm begins by guessing an estimate $\widehat{\OPT}$ of $\OPT_{V,y^*}$, and an estimate $\nu$ of $\vol^+_G(L^*)$, similarly to what was done in \Cref{sec:rooted-minimum-edge-cut} -- specifically, it picks $\widehat{\OPT}$ uniformly at random among the integral powers of $2$ in the range $[1,n \cdot W]$, and it picks $\nu$ uniformly at random among the integral powers of $2$ in the range $[1,m]$.
Similarly to \Cref{sec:rooted-minimum-edge-cut}, let $\event^{\text{estimates}}$ denote the good event in which $\frac{\OPT_{V,y^*}}{2} < \widehat{\OPT} \leq \OPT_{V,y^*}$ and $\frac{\vol^+_G(L^*)}{2} < \nu \leq \vol^+_G(L^*)$ both hold, and observe that $\Pr[\event^{\text{estimates}}] \geq \Omega\left(\frac{1}{\log^2(m \cdot W)}\right)$.
Next, our algorithm picks a set $T \subseteq V(G)$ of $O\left(\frac{m}{\nu}\right)$ \emph{terminals}, using the procedure from \Cref{obs:picking terminals} with parameter $\nu$ and the set $F=\{y^*\} \cup N^-_G(y^*)$ of forbidden vertices. (The difference compared to \Cref{sec:rooted-minimum-edge-cut} is that the set $F$ now also includes the vertices of $N^-_G(y^*)$.)
Let $\event$ denote the good event that event $\event^{\text{estimates}}$ occurred and $T \cap L^* \neq \emptyset$ holds.
By \Cref{obs:picking terminals}, $\Pr[\event \mid \event^{\text{estimates}}] \geq \frac{1}{100}$, and therefore
\begin{equation}\label{eq:probability-of-good-event:vertex-version}
	\Pr[\event] = \Pr[\event \mid \event^{\text{estimates}}] \cdot \Pr[\event^{\text{estimates}}] \geq \Omega\left(\frac{1}{\log^2(m \cdot W)}\right).
\end{equation}

\paragraph{Hierarchical Partition of the Terminals.}
Our algorithm constructs a hierarchical partition of the set $T$ of terminals with $z = \ceil{\log |T|}$ levels, exactly as in \Cref{sec:rooted-minimum-edge-cut}.

\paragraph{Main Loop of the Algorithm.}
As in \Cref{sec:rooted-minimum-edge-cut}, our algorithm consists of $z$ phases, where for each $1 \leq i \leq z$, the goal of Phase $i$ is to construct a sparsifier graph $G_B$ for each level-$i$ batch $B \in \bset_i$. The properties that this sparsifier graph must satisfy are:

%\end{itemize}
\begin{properties}{P'}
	\item\label{property:G_B-P1:vertex-version} the graph $G_B$ contains at most $\Tilde{O}\left(\frac{m}{2^i \cdot \epsilon}\right)$ edges and vertices, and no more than the number of edges and vertices in $G$;
	\item\label{property:G_B-P2:vertex-version} $\{y^*\} \subseteq V(G_B) \subseteq V(G)$;
	\item\label{property:G_B-P3:vertex-version} for every set $L' \subseteq V(G_B) \setminus (\{y^*\} \cup N^-_{G_B}(y^*))$, the inequality $w_B(N^+_{G_B}(L')) \geq w(N^+_{G}(L'))$ holds and $y^* \notin N^+_G(L')$ -- in particular, for every terminal $x \in B \cap V(G_B)$, the $x$-$y^*$ vertex-connectivity $\kappa_{G_B}(x,y^*)$ in $G_B$ is at least as large as the vertex-connectivity $\kappa_G(x,y^*)$ in $G$;
	\item\label{property:G_B-P4:vertex-version} if there exists a terminal $x^* \in B \cap L^*$ and the good event $\event$ occurred, then $x^* \in V(G_B) \setminus N^-_{G_B}(y^*)$ holds and the $x^*$-$y^*$ vertex-connectivity $\kappa_{G_B}(x^*,y^*)$ in $G_B$ satisfies $\kappa_{G_B}(x^*,y^*) \leq \left(1 + \frac{i\epsilon}{z}\right)\OPT_{V,y^*}$ -- furthermore, there exists an $x^*$-$y^*$ vertex-cut $(L',S',R')$ in $G_B$ of value $w_B(S') \leq \left(1 + \frac{i\epsilon}{z}\right)\OPT_{V,y^*}$ that additionally satisfies the condition $\vol^+_G(L') \leq 2\nu$.
%	\item\label{property:G_B-P5:vertex-version} for every vertex $v \in V(G_B)$, there exists an edge $(v,y^*)$ in $G_B$ if and only if there exists an edge $(v,y^*)$ in $G$. In other words, $N^-_{G_B}(y^*) = N^-_G(y^*) \cap V(G_B)$.
\end{properties}
The following claim is analogous to \Cref{cl:implementing-a-single-iteration-of-main-algorithm:edge-version}, and is the main technical component of our algorithm.
The proof of this claim- and the description of procedure implementing it- are presented in \Cref{sec:description-of-iteration-of-main-algorithm:vertex-version}.

\begin{claim}\label{cl:implementing-a-single-iteration-of-main-algorithm:vertex-version}
	There is a deterministic procedure that we refer to as \vertexapproxsparsify, that, given an integer $1 \leq i \leq z$, a batch $B \in \bset_i$, the parameters $\epsilon,\nu,\widehat{\OPT}$, and access to the graph $G$, computes the graph $G_B$ satisfying Properties \ref{property:G_B-P1:vertex-version}-\ref{property:G_B-P4:vertex-version} in $O\left(\frac{m^{1+o(1)} \cdot \log W}{2^i \cdot \epsilon}\right)$ time. If $i>1$, then the procedure also requires access to the graph $G_{B'}$ corresponding to the level-$(i-1)$ parent batch $B'$ of $B$.
\end{claim}

\noindent The following corollary of \Cref{cl:implementing-a-single-iteration-of-main-algorithm:vertex-version} summarizes the running time of each phase.

\begin{corollary}\label{cor:running-time-of-one-phase:vertex-version}
	Each phase of the algorithm can be implemented in time $O\left(\frac{m^{1+o(1)} \cdot \log W}{\epsilon}\right)$.
\end{corollary}
\begin{proof}
	The proof is exactly the same as that of \Cref{cor:running-time-of-one-phase:edge-version}, except using the procedure \vertexapproxsparsify instead of \approxsparsify.
\end{proof}

%\begin{claim}\label{cl:implementing-a-single-iteration-of-main-algorithm:vertex-version}
%	There is a deterministic procedure for implementing each phase of the algorithm in time $O\left(\frac{m^{1+o(1)} \cdot \log W}{\epsilon}\right)$. (For every phase except the first one, the procedure makes use of the graphs computed in the previous phase.)
%\end{claim}
%
%\noindent Next, we describe how the algorithm computes its output at the end of the final phase.

\paragraph{Termination.}
We now describe how our algorithm proceeds after the final phase. This is exactly analogous to how the algorithm from \Cref{sec:rooted-minimum-edge-cut} operates after its final phase.
For each terminal $x \in T$, let $B^{(x)} \in \bset_z$ denote the level-$z$ batch that contains $x$.
Furthermore, let $T' \subseteq T$ denote the set of all terminals $x$ such that $x \in V(G_{B^{(x)}})$.
After the final phase, for each terminal $x \in T'$, our algorithm uses \Cref{thm:fast-max-s-t-vertex-cut} on the graph $G_{B^{(x)}}$ in order to determine the vertex connectivity $\kappa_{G_{B^{(x)}}}(x,y^*)$.
Then, our algorithm picks the terminal $\tilde{x}$ that achieves the minimum vertex-connectivity $\kappa_{G_{B^{(x)}}}(x,y^*)$ among all the terminals of $T'$; computes a minimum $\tilde{x}$-$y^*$ vertex-cut $(L,S,R)$ in $G$ using \Cref{thm:fast-max-s-t-vertex-cut}, and outputs this vertex-cut. (If the set $T'$ is empty, then the algorithm instead outputs an arbitrary vertex-cut $(L,S,R)$ with $y^* \in R$. Recall that, by the definition of the input of \Cref{thm : main-thm-rooted-vertex-version}, such a vertex-cut is guaranteed to exist, and observe that it can be found in time $O(n)$.)
Observe the following.

\begin{observation}\label{obs:termination-value-of-cut-part1:vertex-version}
	If $T' \neq \emptyset$, then the value of the vertex-cut $(L,S,R)$ returned by the algorithm is at most $\min_{x \in T'}\kappa_{G_{B^{(x)}}}(x,y^*)$.
\end{observation}
\begin{proof}
	The proof of this observation is exactly the same as that of \Cref{obs:termination-value-of-cut-part1:edge-version}, except using vertex connectivities instead of edge connectivities, and using Property \ref{property:G_B-P3:vertex-version} instead of Property \ref{property:G_B-P3:edge-version}.
%	As $(L,S,R)$ is selected to be a minimum $\tilde{x}$-$y^*$ vertex-cut in $G$, it suffices to prove that $\kappa_G(\tilde{x},y^*) \leq \min_{x \in T'}\kappa_{G_{B^{(x)}}}(x,y^*)$.
%	Indeed, Property \ref{property:G_B-P3:vertex-version} guarantees that $\kappa_G(\tilde{x},y^*) \leq \kappa_{G_{B^{(\tilde{x})}}}(\tilde{x},y^*)$, and the choice of $\tilde{x}$ guarantees that $\kappa_{G_{B^{(\tilde{x})}}}(\tilde{x},y^*) = \min_{x \in T'}\kappa_{G_{B^{(x)}}}(x,y^*)$.
\end{proof}

The next corollary follows from the above observation.
\begin{corollary}\label{cor:good-event-implies-success:vertex-version}
	If the good event $\event$ occurs, then the vertex-cut $(L,S,R)$ returned by the algorithm has value at most $\left(1+\epsilon\right)\OPT_{V,y^*}$.
\end{corollary}
\begin{proof}
	The proof of this corollary is analogous to that of \Cref{cor:good-event-implies-success:edge-version}.
%	Suppose that the good event $\event$ occurs.
%	In this case, by the definition of this event, there exists a vertex $x^* \in T \cap L^*$.
%	Now, consider the graph $G_{B^{(x^*)}}$ corresponding to the level-$z$ batch $B^{(x^*)}$.
%	By Property \ref{property:G_B-P4:vertex-version}, this graph satisfies that $x^* \in V\left(G_{B^{(x^*)}}\right)$ and that $\kappa_{G_{B^{(x^*)}}}(x^*,y^*) \leq \left(1+\epsilon\right)\OPT_{V,y^*}$.
%	In other words, $x^*$ is a vertex of $T'$ that satisfies $\kappa_{G_{B^{(x^*)}}}(x^*,y^*) \leq \left(1+\epsilon\right)\OPT_{V,y^*}$.
%	Therefore, $\min_{x \in T'}\kappa_{G_{B^{(x)}}}(x,y^*) \leq \left(1+\epsilon\right)\OPT_{V,y^*}$.
%	The corollary now follows from \Cref{obs:termination-value-of-cut-part1:vertex-version}.
\end{proof}

\paragraph{Running Time and Correctness Analysis.}
By the description of the algorithm, it always outputs a vertex-cut $(L,S,R)$ in $G$ such that $y^* \in R$.
Furthermore, by \Cref{cor:good-event-implies-success:vertex-version}, whenever event $\event$ occurs, this vertex-cut has value at most $(1+\epsilon)\OPT_{V,y^*}$; by \Cref{eq:probability-of-good-event:vertex-version}, this happens with probability $\Omega\left(\frac{1}{\log^2(m \cdot W)}\right)$.
Furthermore, an analogous running time analysis to the one from \Cref{sec:rooted-minimum-edge-cut} shows that the running time of the above algorithm is $O\left(\frac{m^{1+o(1)} \cdot \log W}{\epsilon}\right)$.
Thus, the algorithm described in this section runs in time $O\left(\frac{m^{1+o(1)} \cdot \log W}{\epsilon}\right)$, and the output vertex-cut $(L,S,R)$ satisfies $w(S) \leq (1+\epsilon)\OPT_{V,y^*}$ with probability at least $\frac{1}{\polylog(mW)}$.
By repeating the algorithm $\polylog(mW)$ times and choosing the lowest-value cut among the obtained vertex-cuts, we can boost the success probability of the algorithm to $\frac{1}{2}$, at the cost of increasing the running time to $O\left(\frac{m^{1+o(1)} \cdot \polylog W}{\epsilon}\right)$.
As we did in \Cref{sec:rooted-minimum-edge-cut}, we can then reduce the running time back to $O\left(\frac{m^{1+o(1)} \cdot \log W}{\epsilon}\right)$ using the reduction from \Cref{sec: low-dependence-on-W}.

\subsection{Execution of a Single Iteration: proof of \Cref{cl:implementing-a-single-iteration-of-main-algorithm:vertex-version}}\label{sec:description-of-iteration-of-main-algorithm:vertex-version}

In this section, we describe the \vertexapproxsparsify procedure from \Cref{cl:implementing-a-single-iteration-of-main-algorithm:vertex-version}.
In the case where $i=1$, this procedure simply returns the graph $G_B=G$, as we did in the procedure \approxsparsify from \Cref{sec:description-of-iteration-of-main-algorithm:edge-version}.
It can be easily verified that when $i=1$, this choice of $G_B=G$ indeed satisfies Properties \ref{property:G_B-P1:vertex-version}-\ref{property:G_B-P4:vertex-version}.
For the remainder of this section, we consider the case where $i>1$.
We begin with a short high-level discussion, explaining the few non-trivial differences between the implementation of \vertexapproxsparsify and that of \approxsparsify.
We then formally present and analyze the subroutine \vertexapproxsparsify in \Cref{sec:approx-sparsify-for-i-larger-than-1:vertex-version}.

\paragraph{High-Level Overview: Difference Between Edge-Version and Vertex-Version.}
The only non-trivial task when adapting the algorithm from \Cref{sec:description-of-iteration-of-main-algorithm:edge-version} to the vertex-version, is the construction of the graph $H_{B'}$.
Recall that $B'$ denotes the level-$(i-1)$ parent batch of the input level-$i$ batch $B$, and recall that, in the edge-version of the \approxsparsify procedure, we constructed a graph $H_{B'}$ from $G_{B'}$ by introducing a new super-source vertex $s$, and connecting it to every vertex of $B \cap A_{B'}$ via an edge of weight $4\widehat{\OPT}$. (Here, $A_{B'}$ denotes the intermediate set that was used to construct $G_{B'}$ in the previous phase. In the edge-version of this procedure, $B \cap A_{B'}$ was equal to $B \cap V(G_{B'})$, while in the vertex-version it will be equal to $B \cap (V(G_{B'}) \setminus N^-_{G_{B'}}(y^*))$.)
In the vertex-version of this procedure, we will also construct a graph $H_{B'}$ that serves the same purpose -- however, we can no longer give the edges outgoing from $s$ a weight of $4\widehat{\OPT}$, as the graph $H_{B'}$ needs to be vertex-weighted rather than edge-weighted.
One natural attempt at solving this would be the following: for each $x \in B \cap A_{B'}$, we introduce the edge $(s,x)$, and then subdivide this edge with a new vertex, whose weight is $4\widehat{\OPT}$.
However, even this solution does not work, for a subtle reason that is described next.
In the analysis from \Cref{sec:description-of-iteration-of-main-algorithm:edge-version}, it was crucial that for every vertex $x \in B \cap A_{B'}$, if the $x$-$y^*$ edge-connectivity in $H_{B'}$ was smaller than $4\widehat{\OPT}$, then $x$ will belong to the same side as $s$ in every minimum $s$-$y^*$ cut in $H_{B'}$; conceptually, this happens because it is possible to ship $\min\{\lambda_{H_{B'}}(x,y^*),4\widehat{\OPT}\}$ units of flow from $s$ to $y^*$ using only flow-paths that exit $s$ through the edge $(s,x)$. (In reality, this condition by itself only suffices to prove that $x$ is on the same side as $s$ in \emph{some} minimum $s$-$y^*$ cut, rather than in every such cut. However, for the sake of simplicity, we will ignore this detail in this high-level overview.)
% , for any $x$-$y^*$ edge-capacitated flow $f$ in $H_{B'}$ that satisfies $\val(f) \leq \widehat{\OPT}$, it is possible to "extend" the flow $f$ to an $s$-$y^*$ flow of the same value, by setting $f(s,x)$ to $\val(f)$.
Now, the issue in the vertex-capacitated setting is that, even if we subdivide the edge $(s,x)$ with a vertex $x^{\tcopy}$ of weight greater or equal to $4\widehat{\OPT}$, and even if the $x$-$y^*$ vertex-connectivity is large, we may still not be able to ship much flow from $s$ to $y^*$ through $x$, because the vertex $x$ itself may have very low weight.
We solve this issue via the same tactic as was employed by \cite{CMT25}: that is, rather than connecting the vertex $x^{\tcopy}$ to the vertex $x$, we instead connect $x^{\tcopy}$ directly to the out-neighbors of $x$; it is then not hard to verify that $\min\{\kappa_{H_{B'}}(x,y^*),4\widehat{\OPT}\}$ units of flow can be shipped from $s$ to $y^*$ via the edge $(s,x^{\tcopy})$. Intuitively, this helps us guarantee that $x^{\tcopy}$ will belong to the same side as $s$ in the minimum $s$-$y^*$ vertex-cut $(L_A,S_A,R_A)$ that the procedure \vertexapproxsparsify computes (which is analogous to the edge-cut $(A,A')$ from \approxsparsify). Then, if $x \notin L_A$, we can simply move $x$ from $S_A \cup R_A$ to $L_A$, and it is not hard to see that the triple $(L_A,S_A,R_A)$ still remains a valid $s$-$y^*$ vertex-cut in $H_{B'}$.

\subsubsection{Implementation of \vertexapproxsparsify for $i>1$.}\label{sec:approx-sparsify-for-i-larger-than-1:vertex-version}
Fix some $1 \leq i \leq z$. We now formally describe how the algorithm computes the graph $G_B$ for a single batch $B \in \bset_i$ during Phase $i$ of the algorithm.
The computation of the graph $G_B$ is performed in two steps.
In the first step, we compute a set $A_B \subseteq V(G) \setminus (\{y^*\} \cup N^-_G(y^*))$ of vertices that satisfies the following two properties.
\begin{properties}{Q'}
	\item\label{property:A_B-P'1:vertex-version} $\vol^+_G(A_B) = \Tilde{O}\left(\frac{m}{2^i \cdot \epsilon}\right)$; and
	\item\label{property:A_B-P'2:vertex-version} if there exists a terminal $x^* \in B \cap L^*$ and the good event $\event$ occurred, then there exists a subset $L' \subseteq A_B$ with $x^* \in L'$ and $\vol^+_G(L') \leq 2\nu$ such that the vertex-cut $(L',N^+_G(L'),V(G) \setminus (L' \cup N^+_G(L')))$ in $G$ has value $w(N^+_G(L')) \leq \left(1 + \frac{i\epsilon}{z}\right)\OPT_{V,y^*}$.
\end{properties}
In the second step, we use this set $A_B$ in order to compute the graph $G_B$ satisfying Properties \ref{property:G_B-P1:vertex-version}-\ref{property:G_B-P4:vertex-version} in time $\Tilde{O}(\vol^+_G(A_B))$.
The implementation of the second step is summed up by the following claim, whose proof is deferred to \Cref{sec:reduction-to-computing-A_B:vertex-version}.

\begin{claim}\label{cl:reduction-to-computing-A_B:vertex-version}
	There exists a procedure that, given a set $A_B \subseteq V(G) \setminus (\{y^*\} \cup N^-_G(y^*))$ satisfying Properties \ref{property:A_B-P'1:vertex-version}-\ref{property:A_B-P'2:vertex-version}, computes a graph $G_B$ satisfying Properties \ref{property:G_B-P1:vertex-version}-\ref{property:G_B-P4:vertex-version}.
	The running time of the procedure is $O\left(\vol^+_G(A_B)\right) = \Tilde{O}\left(\frac{m}{2^i \cdot \epsilon}\right)$.
\end{claim}

\paragraph{Obtaining the set $A_B$.}
In the remainder of the current section, we describe the procedure for computing the set $A_B \subseteq V(G) \setminus (\{y^*\} \cup N^-_G(y^*))$ of vertices.
Let $B' \in \bset_{i-1}$ denote the level-$(i-1)$ parent batch of $B$.
In order to compute $A_B$, we will use the graph $G_{B'}$ that we computed for batch $B'$ in the previous phase.
Furthermore, for the remainder of this section, we let $A_{B'}$ denote the set $A_{B'} = V(G_{B'}) \setminus (\set{y^*} \cup N^-_{G_{B'}}(y^*))$.
To compute $A_B$, we first construct a graph $H_{B'}$ from $G_{B'}$ as follows: we introduce a new "super-source" vertex $s$, and, for every terminal $x \in B \cap A_{B'}$, we introduce a new vertex $x^{\tcopy}$ of weight $4\widehat{\OPT}$, as well as an edge $(s,x^{\tcopy})$, and edges from $x^{\tcopy}$ to all out-neighbors of $x$ in $G_{B'}$; Furthermore, for every vertex $v \in A_{B'}$ (including those in $B \cap A_{B'}$), we introduce a new vertex $v^{\vcopy}$ of weight $\frac{\epsilon\cdot\widehat{\OPT} \cdot \deg^+_G(v)}{z \cdot 2\nu}$, as well as edges $(v,v^{\vcopy})$ and $(v^{\vcopy},y^*)$;\footnote{Observe that the weight of $v^{\vcopy}$ is positive, as the assumption that $\OPT_{V,y^*} > 0$ implies that every vertex $v \in V(G) \setminus \set{y^*}$ has $\deg^+_G(v) > 0$.} if $v \in B \cap A_{B'}$, then we additionally introduce an edge $(v^{\tcopy},v^{\vcopy})$.
The next observation about the graph $H_{B'}$ follows immediately from the construction of this graph, and is analogous to \Cref{obs:relating-cut-value-after-modification:edge-version} from \Cref{sec:rooted-minimum-edge-cut}.
\begin{observation}\label{obs:relating-cut-value-after-modification:vertex-version}
	For every set $\hat{L} \subseteq A_{B'}$ of vertices, we have $w_{B'}(N^+_{H_{B'}}(\hat{L})) - w_{B'}(N^+_{G_{B'}}(\hat{L})) = \frac{\epsilon\cdot\widehat{\OPT} \cdot \vol^+_G(\hat{L})}{z \cdot 2\nu}$.
\end{observation}
We will also use the following definition and observation, which do not have analogous versions in \Cref{sec:rooted-minimum-edge-cut}.
\begin{definition}
	Let $A_{B'}^{\complete}$ denote the set $A_{B'}^{\complete}=A_{B'} \cup \set{x^{\tcopy} \mid x \in B \cap A_{B'}} = V(H_{B'}) \setminus (\set{y^*,s} \cup N^-_{H_{B'}}(y^*))$.
	Furthermore, for every set $\hat{L} \subseteq A_{B'}^{\complete}$, we define the \emph{completed version of $\hat{L}$} to be the set $\hat{L}^{\complete} = \hat{L} \cup \set{x^{\tcopy} \mid x \in B \cap A_{B'} \cap \hat{L}} \cup \set{x \in B \cap A_{B'} \mid x^{\tcopy} \in \hat{L}}$ that is obtained from $\hat{L}$ by adding to it each vertex $x^{\tcopy}$ for which the corresponding vertex $x \in B \cap A_{B'}$ belongs to $\hat{L}$, as well as each vertex $x \in B \cap A_{B'}$ for which the corresponding vertex $x^{\tcopy}$ belongs to $\hat{L}$.
\end{definition}
\begin{observation}\label{obs:out-neighbors-of-complete-set}
	For every set $\hat{L} \subseteq A_{B'}^{\complete}$ of vertices, we have $N^+_{H_{B'}}(\hat{L}^{\complete}) \subseteq N^+_{H_{B'}}(\hat{L})$.
\end{observation}
\begin{proof}
	This observation follows from the fact that for every $x \in B \cap A_{B'}$, the vertices $x$ and $x^{\tcopy}$ have the exact same out-neighbors in $H_{B'}$, together with the fact that $\hat{L} \subseteq \hat{L}^{\complete}$.
\end{proof}

\paragraph{Continuing the computation of $A_B$.}
We now continue the description of the algorithm computing $A_B$.
The algorithm uses \Cref{thm:fast-max-s-t-vertex-cut} in order to compute a minimum $s$-$y^*$ vertex-cut $(L_{A},S_{A},R_{A})$ in the graph $H_{B'}$, and chooses $A_B$ to be the set $A_B=\widehat{L_{A}}^{\complete} \setminus V^{\tcopy}$, where $\widehat{L_A} = L_A \setminus \set{s}$ and where $V^{\tcopy}=\set{x^{\tcopy} \mid x \in B \cap A_{B'}}$. (Observe that $\widehat{L_A} \subseteq A_{B'}^{\complete}$ must hold because $(L_{A},S_{A},R_{A})$ is an $s$-$y^*$ vertex-cut in $H_{B'}$ and because $A_{B'}^{\complete}= V(H_{B'}) \setminus (\set{y^*,s} \cup N^-_{H_{B'}}(y^*))$.)
This concludes the description of how we compute $A_B$. Furthermore, it is not hard to verify that $A_B \subseteq V(G) \setminus (\set{y^*} \cup N^-_{G}(y^*))$ holds by using the fact that $A_B \subseteq A_{B'} = V(G) \setminus (\set{y^*} \cup N^-_{G_{B'}}(y^*))$ as well as using Property \ref{property:G_B-P3:vertex-version} of the graph $G_{B'}$.
It now remains to prove that this set indeed satisfies Properties \ref{property:A_B-P'1:vertex-version}-\ref{property:A_B-P'2:vertex-version}.
The next two observation will be useful for this task; these are analogous to Observations \ref{obs:bound-on-size-of-cut-in-H-graph:edge-version} and \ref{obs:minimum-cut-in-modified-graph-contains-distinguished-terminal} from \Cref{sec:rooted-minimum-edge-cut}.

\begin{observation}\label{obs:bound-on-size-of-cut-in-H-graph:vertex-version}
	The value $w_{B'}(S_A)$ of the vertex-cut $(L_A,S_A,R_A)$ in the graph $H_{B'}$ is $w_{B'}(S_A)=O\left(\frac{m \cdot \widehat{\OPT}}{\nu \cdot 2^i}\right)$.
\end{observation}
\begin{proof}
	The proof of this observation is analogous to that of \Cref{obs:bound-on-size-of-cut-in-H-graph:edge-version}: specifically, since $(L_A,S_A,R_A)$ is chosen to be a minimum $s$-$y^*$ vertex-cut in $H_{B'}$, it suffices to prove that some $s$-$y^*$ vertex-cut in $H_{B'}$ has value $O\left(\frac{m \cdot \widehat{\OPT}}{\nu \cdot 2^i}\right)$; furthermore, similarly to what we did in the proof of \Cref{obs:bound-on-size-of-cut-in-H-graph:edge-version}, it can be verified that the triple $(\set{s},V^{\tcopy},V(H_{B'})\setminus(\set{s} \cup V^{\tcopy}))$ is such a vertex-cut, where $V^{\tcopy}=\set{x^{\tcopy} \mid x \in B \cap A_{B'}}$.
%	By the construction of the graph $H_{B'}$, the out-neighbors of the vertex $s$ are exactly the vertices of $V^{\tcopy}=\set{x^{\tcopy} \mid x \in B \cap A_{B'}}$. Furthermore, $V^{\tcopy}$ contains at most $|B|$ vertices, and each vertex $x^{\tcopy} \in V^{\tcopy}$ has weight exactly $w_{B'}(x^{\tcopy}) = 4\widehat{\OPT}$.
%	Therefore, there exists an $s$-$y^*$ vertex-cut in $H_{B'}$ of weight at most $|B| \cdot 4\widehat{\OPT}$, which is defined by $(\set{s},V^{\tcopy},V(H_{B'}) \setminus (\set{s} \cup V^{\tcopy}))$.
%	Therefore, as $(L_{A},S_{A},R_{A})$ is chosen to be a minimum $s$-$y^*$ vertex-cut in $H_{B'}$, it follows that the value $w_{B'}(S_A)$ of this cut is at most $|B| \cdot 4\widehat{\OPT}$.
%	The claim follows as \Cref{obs:size-of-batches} states that $|B| = O\left(\frac{m}{\nu \cdot 2^i}\right)$.
\end{proof}

\begin{observation}\label{obs:minimum-cut-in-modified-graph-contains-distinguished-terminal:vertex-version}
	The set $L_A$ contains every vertex $x^{\tcopy} \in V^{\tcopy}$ whose vertex-connectivity to $y^*$ in $H_{B'}$ satisfies $\kappa_{H_{B'}}(x^{\tcopy},y^*) < 4\widehat{\OPT}$.
\end{observation}
\begin{proof}
	Suppose towards contradiction that some vertex $x^{\tcopy} \in V^{\tcopy}$ does not belong to $L_A$, and yet has $\kappa_{H_{B'}}(x^{\tcopy},y^*) < 4\widehat{\OPT}$. Furthermore, fix one such vertex $x^{\tcopy}$ for the remainder of the proof.
	We will demonstrate that in this case, there must exist an $s$-$y^*$ vertex-cut $(\tilde{L},\tilde{S},\tilde{R})$ in $H_{B'}$ of weight smaller than that of $(L_A,S_A,R_A)$, in contradiction to the choice of $(L_A,S_A,R_A)$ as a minimum $s$-$y^*$ vertex-cut in $H_{B'}$.
	Specifically, by the assumption that $\kappa_{H_{B'}}(x^{\tcopy},y^*) < 4\widehat{\OPT}$, there must exist an $x^{\tcopy}$-$y^*$ vertex-cut $(\hat{L},\hat{S},\hat{R})$ in $H_{B'}$ with $w_{B'}(\hat{S}) < 4\widehat{\OPT}$; we define the set $\tilde{S}$ of the vertex-cut $(\tilde{L},\tilde{S},\tilde{R})$ to be $\tilde{S} = (S_A \setminus \set{x^{\tcopy}}) \cup \hat{S}$.
	Before revealing the definition of the sets $\tilde{L}$ and $\tilde{R}$, we will prove that the weight of $\tilde{S}$ is smaller than that of $S_A$: indeed, since $\tilde{S} = (S_A \setminus \set{x^{\tcopy}}) \cup \hat{S}$ and since $w_{B'}(\hat{S}) < 4\widehat{\OPT} = w_{B'}(x^{\tcopy})$, it suffices to show that $x^{\tcopy} \in S_A$; furthermore, recall that $x^{\tcopy}$ was chosen as a vertex that does not belong to $L_A$, and observe that $x^{\tcopy} \in L_A \cup S_A$ holds because $(L_A,S_A,R_A)$ is an $s$-$y^*$ vertex-cut in the graph $H_{B'}$, and this graph contains an edge $(s,x^{\tcopy})$. This concludes the proof that the weight of $\tilde{S}$ is smaller than that of $S_A$.
	
	Next, we must show that there exist sets $\tilde{L},\tilde{R} \subseteq V(H_{B'})$ such that $(\tilde{L},\tilde{S},\tilde{R})$ is an $s$-$y^*$ vertex-cut in $H_{B'}$.
	By the definition of $\tilde{S}$, it is not hard to verify that $s,y^* \notin \tilde{S}$.
	Therefore, it suffices to show that every $s$-$y^*$ path in $H_{B'}$ must traverse at least one vertex of $\tilde{S}$: we can then choose $\tilde{L}$ to be the set of vertices that are reachable from $s$ via paths that do not traverse vertices of $\tilde{S}$, and choose $\tilde{R}$ to be the set of remaining vertices in $H_{B'}$.
	We show this now:
	indeed, consider any $s$-$y^*$ path $P$ in $H_{B'}$; since $(L_A,S_A,R_A)$ is an $s$-$y^*$ vertex-cut in $H_{B'}$, it follows that $P$ must traverse a vertex $v$ of $S_A$; if $v \neq x^{\tcopy}$ then $v \in S_A \setminus x^{\tcopy} \subseteq \tilde{S}$ and we are done; otherwise, $v = x^{\tcopy}$, and as $(\hat{L},\hat{S},\hat{R})$ is an $x^{\tcopy}$-$y^*$ vertex-cut, it follows that the path $P$ must visit some vertex of $\hat{S} \subseteq \tilde{S}$ on the way from $v$ to $y^*$.
	This concludes the proof of \Cref{obs:minimum-cut-in-modified-graph-contains-distinguished-terminal:vertex-version}.
\end{proof}

The following corollary of \Cref{obs:minimum-cut-in-modified-graph-contains-distinguished-terminal:vertex-version} will also be useful.

\begin{corollary}\label{cor:minimum-cut-in-modified-graph-contains-distinguished-terminal:vertex-version}
	The set $\widehat{L_A}^{\complete}$ contains every terminal $x \in B \cap A_{B'}$ whose vertex-connectivity to $y^*$ in $H_{B'}$ satisfies $\kappa_{H_{B'}}(x,y^*) < 4\widehat{\OPT}$.
\end{corollary}
\begin{proof}
	Consider any terminal $x \in B \cap A_{B'}$ with $\kappa_{H_{B'}}(x,y^*) < 4\widehat{\OPT}$.
	Since the out-neighbors of $x$ in $H_{B'}$ are exactly the same as the out-neighbors of $x^{\tcopy}$, it follows that $\kappa_{H_{B'}}(x^{\tcopy},y^*) < 4\widehat{\OPT}$ holds as well.
	Therefore, by \Cref{obs:minimum-cut-in-modified-graph-contains-distinguished-terminal:vertex-version}, the set $L_A$ contains $x^{\tcopy}$.
	Since $\widehat{L_A}$ is defined as $\widehat{L_A} = L_A \setminus \set{s}$, it follows that $x^{\tcopy} \in \widehat{L_A}$, and thus $x \in \widehat{L_A}^{\complete}$, as needed.
\end{proof}

Recall that it remains to prove that the set $A_B=\widehat{L_{A}}^{\complete} \setminus V^{\tcopy}$ satisfies Properties \ref{property:A_B-P'1:vertex-version}-\ref{property:A_B-P'2:vertex-version}.
We begin by proving that it satisfies Property \ref{property:A_B-P'1:vertex-version}.

\begin{claim}
	The set $A_B$ satisfies Property \ref{property:A_B-P'1:vertex-version}. That is, $\vol^+_G(A_B) = \Tilde{O}\left(\frac{m}{2^i \cdot \epsilon}\right)$.
\end{claim}
\begin{proof}
	Recall that for every vertex $v \in A_B \subseteq A_{B'}$, the graph $H_{B'}$ contains a vertex $v^{\vcopy}$ of weight $w_{B'}(v^{\vcopy})=\frac{\epsilon\cdot\widehat{\OPT} \cdot \deg^+_G(v)}{z \cdot 2\nu}$, as well as edges $(v,v^{\vcopy})$ and $(v^{\vcopy},y^*)$; and, if $v \in B \cap A_{B'}$, then the graph $H_{B'}$ additionally contains the edge $(v^{\tcopy},v^{\vcopy})$.
	We begin the proof of the claim by showing that for each vertex $v \in A_B$, the set $N^+_{H_{B'}}(L_A)$ must include the vertex $v^{\vcopy}$.
	Indeed, recall that for each vertex $v \in A_B = \widehat{L_{A}}^{\complete} \setminus V^{\tcopy}$, the set $L_A$ must include either $v$, or possibly $v^{\tcopy}$ in the case where $v \in B \cap A_{B'}$; either way, the set $L_A$ contains some in-neighbor of $v^{\vcopy}$.
	Observe further that $L_A$ cannot contain the vertex $v^{\vcopy}$ itself, since $(L_A,S_A,R_A)$ is an $s$-$y^*$ vertex-cut in $H_{B'}$, and this graph includes an edge $(v^{\vcopy},y^*)$. Thus, for each such $v \in A_B$, the set $N^+_{H_{B'}}(L_A)$ must include $v^{\vcopy}$.
	The remainder of this proof is analogous to that of \Cref{cl:A_B-satisfies-property-P'1:edge-version};
	Specifically, as $N^+_{H_{B'}}(L_A) \subseteq S_A$, it follows that $w_{B'}(S_A) \geq w_{B'}(\set{v^{\vcopy} \mid v \in A_B}) = \frac{\epsilon\cdot\widehat{\OPT} \cdot \vol^+_G(A_B)}{z \cdot 2\nu}$.
	This means that $\vol^+_G(A_B)$ is upper-bounded by $\vol^+_G(A_B) \leq w_{B'}(S_A) \cdot \frac{z \cdot 2\nu}{\epsilon\cdot\widehat{\OPT}}$.
	By plugging the upper-bound $w_{B'}(S_A)=O\left(\frac{m \cdot \widehat{\OPT}}{\nu \cdot 2^i}\right)$ from \Cref{obs:bound-on-size-of-cut-in-H-graph:vertex-version} into this last inequality, we get that $\vol^+_G(A_B) \leq O\left(\frac{z \cdot m}{\epsilon \cdot 2^{i}}\right)$.
	The corollary now follows as $z = \Tilde{O}(1)$.
\end{proof}

Lastly, it remains to prove that the set $A_B$ satisfies \ref{property:A_B-P'2:vertex-version}, which we do in the next claim. The proof of this claim is similar to that of \Cref{cl:A_B-satisfies-property-P'2:edge-version}.

\begin{claim}\label{cl:A_B-satisfies-property-P'2:vertex-version}
	The set $A_B$ satisfies Property \ref{property:A_B-P'2:vertex-version}.
\end{claim}
\begin{proof}
	Suppose that there exists a terminal $x^* \in B \cap L^*$, and suppose that the good event $\event$ occurred.
	We need to prove that there exists a subset $L' \subseteq A_B$ with $x^* \in L'$ and $\vol^+_G(L') \leq 2\nu$ such that the vertex-cut $(L',N^+_G(L'),V(G) \setminus (L' \cup N^+_G(L')))$ in $G$ has value $w(N^+_G(L')) \leq \left(1 + \frac{i\epsilon}{z}\right)\OPT_{V,y^*}$.
	
	As $x^* \in B \cap L^*$ and as $B'$ is the parent batch of $B$, it follows that $x^* \in B' \cap L^*$.
	Thus, by Property \ref{property:G_B-P4:vertex-version} of the graph $G_{B'}$, $x^* \in V(G_{B'}) \setminus N^-_{G_{B'}}(y^*) = A_{B'} \cup \set{y^*}$ holds, and there exists an $x^*$-$y^*$ vertex-cut $(L'_{B'},S'_{B'},R'_{B'})$ in $G_{B'}$ of value $w_{B'}(S'_{B'}) \leq \left(1 + \frac{(i-1)\epsilon}{z}\right)\OPT_{V,y^*}$, that additionally satisfies $\vol^+_G(L'_{B'}) \leq 2\nu$.
	
	Let $L' = A_B \cap L'_{B'}$.
	We now prove that $x^* \in L'$ and $\vol^+_G(L') \leq 2\nu$, and that the vertex-cut $(L',N^+_G(L'),V(G) \setminus (L' \cup N^+_G(L')))$ in $G$ has value $w(N^+_G(L')) \leq \left(1 + \frac{i\epsilon}{z}\right)\OPT_{V,y^*}$.
	Indeed, the fact that $\vol^+_G(L') \leq 2\nu$ follows from the fact that $L' \subseteq L'_{B'}$, and that $\vol^+_G(L'_{B'}) \leq 2\nu$.
	It remains to prove that $x^* \in L'$, as well as that the vertex-cut $(L',N^+_G(L'),V(G) \setminus (L' \cup N^+_G(L')))$ in $G$ has value $w(N^+_G(L')) \leq \left(1 + \frac{i\epsilon}{z}\right)\OPT_{V,y^*}$.
	We begin by proving the latter statement.
	Specifically, this latter statement follows directly from the next claim.
	
	\begin{claim}\label{cl:upper-bound-on-value-of-new-promised-cut:vertex-version}
		The three inequalities $w(N^+_G(L')) \leq w_{B'}(N^+_{H_{B'}}(L')) \leq w_{B'}(N^+_{H_{B'}}(L'_{B'})) \leq \left(1 + \frac{i\epsilon}{z}\right)\OPT_{V,y^*}$ all hold.
	\end{claim}
	\begin{proof}
		We begin by proving the first inequality of the claim. That is, we begin by proving that $w(N^+_G(L')) \leq w_{B'}(N^+_{H_{B'}}(L'))$.
		Indeed, by Property \ref{property:G_B-P3:vertex-version} of the graph $G_{B'}$, the weight $w(N^+_G(L'))$ is no larger than the weight $w_{B'}(N^+_{G_{B'}}(L'))$ of the out-neighbors of $L'$ in $G_{B'}$; furthermore, by the construction of the graph $H_{B'}$, this latter value is no larger than the weight $w_{B'}(N^+_{H_{B'}}(L'))$ of the out-neighbors of $L'$ in $H_{B'}$. This concludes the proof of the first inequality.
		
		Next, we prove the second inequality from \Cref{cl:upper-bound-on-value-of-new-promised-cut:vertex-version}.
		Consider the set $L_A' = \widehat{L_A}^{\complete} \cup \set{s}$ of vertices in $H_{B'}$, as well as the sets $S_A' = N^+_{H_{B'}}(L_A')$ and $R_A' = V(H_{B'}) \setminus (L_A' \cup S_A')$.
		It is not hard to verify that the triple $(L_A',S_A',R_A')$ is an $s$-$y^*$ vertex-cut in $H_{B'}$; we prove next that it is a \emph{minimum} $s$-$y^*$ vertex-cut in $H_{B'}$, as well as that the equality $L' = L_A' \cap L'_{B'}$ holds; the second inequality from \Cref{cl:upper-bound-on-value-of-new-promised-cut:vertex-version} then follows from \Cref{lem:submodularity-of-minimum-cuts:vertex-version} (with $L'_{B'}$, $L_A'$, $S_A'$, and $R_A'$ in place of $X$, $L$, $S$, and $R$).
		Thus, we must now prove that $(L_A',S_A',R_A')$ is a minimum $s$-$y^*$ vertex-cut in $H_{B'}$, which we do by proving that $S_A' \subseteq S_A$ (Recall that $(L_A,S_A,R_A)$ is itself a minimum $s$-$y^*$ vertex-cut in $H_{B'}$.): it is not hard to verify that $S'_A \subseteq (N^+_{H_{B'}}(s) \cup N^+_{H_{B'}}(\widehat{L_A}^{\complete}))\setminus L_A'$; furthermore, by \Cref{obs:out-neighbors-of-complete-set}, $N^+_{H_{B'}}(\widehat{L_A}^{\complete}) \subseteq N^+_{H_{B'}}(\widehat{L_A})$; so, as $L_A \subseteq L_A'$, it follows that $S_A' \subseteq (N^+_{H_{B'}}(s) \cup N^+_{H_{B'}}(\widehat{L_A}))\setminus L_A = N^+_{H_{B'}}(L_A) \subseteq S_A$, as needed to show that $(L_A',S_A',R_A')$ is a minimum $s$-$y^*$ vertex-cut in $H_{B'}$.
		Next, we must prove that $L' = L_A' \cap L'_{B'}$: indeed, recall that $L'$ is defined as $L' = A_B \cap L'_{B'}$; since $L'_{B'} \subseteq V(G_{B'})$, it then suffices to show that $L_A' \cap V(G_{B'}) = A_B \cap V(G_{B'})$; this in turn holds because $A_B = \widehat{L_A}^{\complete} \setminus V^{\tcopy} = (L_A' \setminus \set{s}) \setminus V^{\tcopy} = L_A' \cap V(G_{B'})$.
		This concludes the proof of the second inequality from \Cref{cl:upper-bound-on-value-of-new-promised-cut:vertex-version}.
		
		Lastly, we now prove the third inequality from \Cref{cl:upper-bound-on-value-of-new-promised-cut:vertex-version}.
		Recall we are assuming that event $\event$ occurred, which means that $\widehat{\OPT} \leq \OPT_{V,y^*}$;
		therefore, since $\vol^+_G(L'_{B'}) \leq 2\nu$, it follows from \Cref{obs:relating-cut-value-after-modification:vertex-version} that $w_{B'}(\partial^+_{H_{B'}}(L'_{B'})) - w_{B'}(\partial^+_{G_{B'}}(L'_{B'})) \leq \frac{\epsilon\cdot\widehat{\OPT}}{z} \leq \frac{\epsilon\cdot\OPT_{V,y^*}}{z}$.
		So, by the choice of $(L'_{B'},S'_{B'},R'_{B'})$ as a vertex-cut in $G_{B'}$ whose value is at most $\left(1 + \frac{(i-1)\epsilon}{z}\right)\OPT_{V,y^*}$, it follows that
		\begin{equation*}
			w_{B'}(N^+_{H_{B'}}(L'_{B'})) \leq w_{B'}(N^+_{G_{B'}}(L'_{B'})) + \frac{\epsilon\cdot\OPT_{V,y^*}}{z} \leq
			w_{B'}(S'_{B'}) + \frac{\epsilon\cdot\OPT_{V,y^*}}{z}
			 \leq \left(1 + \frac{i\epsilon}{z}\right)\OPT_{V,y^*}.
		\end{equation*}
		This concludes the proof of \Cref{cl:upper-bound-on-value-of-new-promised-cut:vertex-version}.
	\end{proof}
	
	We now continue the proof of \Cref{cl:A_B-satisfies-property-P'2:vertex-version}. The remainder of this proof is exactly analogous to the final part of the proof of \Cref{cl:A_B-satisfies-property-P'2:edge-version}.
	Specifically, to prove the claim, it now remains to prove that $x^* \in L'$. As $L'= A_B \cap L'_{B'}$ and $x^* \in L'_{B'}$, it suffices to prove that $x^* \in A_B$.
	We do this next, by first proving that $\kappa_{H_{B'}}(x^*,y^*) < 4\widehat{\OPT}$, and then using \Cref{cor:minimum-cut-in-modified-graph-contains-distinguished-terminal:vertex-version}. Indeed since $x^* \in L'_{B'}$ and $y^* \notin L'_{B'} \cup N^+_{H_{B'}}(L'_{B'})$, it follows by the third inequality of \Cref{cl:upper-bound-on-value-of-new-promised-cut:vertex-version} that the $x^*$-$y^*$ vertex-connectivity in $H_{B'}$ is
	\begin{equation}\label{eq:edge-connectivity-of-distinguished-terminal-in-modified-graph:vertex-version}
		\kappa_{H_{B'}}(x^*,y^*) \leq w_{B'}(N^+_{H_{B'}}(L'_{B'})) \leq \left(1 + \frac{i\epsilon}{z}\right)\OPT_{V,y^*} \leq 2\OPT_{V,y^*}.
	\end{equation}
	Recall also that we are assuming that event $\event$ occurred, which, by the definition of this event, means that $\OPT_{V,y^*} < 2\widehat{\OPT}$.
	By plugging this inequality into \Cref{eq:edge-connectivity-of-distinguished-terminal-in-modified-graph:vertex-version}, it follows that $\kappa_{H_{B'}}(x^*,y^*) < 4\widehat{\OPT}$.
	Since $x^* \in B \cap L'_{B'} \subseteq B \cap A_{B'}$, it now follows from \Cref{cor:minimum-cut-in-modified-graph-contains-distinguished-terminal:vertex-version} that $x^* \in \widehat{L_A}^{\complete}$. Since $x^* \in V(G)$, it then follows that $x^* \in \widehat{L_A}^{\complete} \setminus V^{\tcopy} = A_B$.
	Thus, $x^* \in A_B \cap L'_{B'} = L'$, as we needed to prove.
	This concludes the proof of \Cref{cl:A_B-satisfies-property-P'2:vertex-version}, and thus concludes the analysis of \vertexapproxsparsify.
\end{proof}

\subsubsection{Step 2: Proof of \Cref{cl:reduction-to-computing-A_B:vertex-version}}\label{sec:reduction-to-computing-A_B:vertex-version}

Consider a set $A_B \subseteq V(G) \setminus (\{y^*\} \cup N^-_G(y^*))$ of vertices that satisfies Properties \ref{property:A_B-P'1:vertex-version}-\ref{property:A_B-P'2:vertex-version}.
We construct the graph $G_B$ using $A_B$ as follows: choose the vertex-set of $G_B$ to be $V(G_B) = A_B \cup N^+_G(A_B) \cup \set{y^*}$, and construct the edge-set of $G_B$ by beginning with the set $E_G(A_B,A_B \cup N^+_G(A_B))$ of all edges of $G$ that originate inside $A_B$, and then adding, for every vertex $v \in N^+_G(A_B)$, an edge $(v,y^*)$. Furthermore, set the weight $w_B(v)$ of each vertex $v \in V(G_B)$ to be exactly the same as the weight $w(v)$ of this vertex in $G$.

This kind of construction was also used in \cite{CMT25}, and they referred to the resulting graph as \emph{the graph derived from $G$ via the vertex-set $A_B$}.
It can easily verified that the above construction can be implemented in $\Tilde{O}(\vol^+_G(A_B))$ time.
The following observation is immediate from the construction of $G_B$, and was also formally proved in \cite{CMT25}.
\begin{observation}\label{obs:cuts-in-G_B-have-same-value-as-in-G:vertex-version}
	For every set $L' \subseteq A_B$, we have $w_B(N^+_{G_B}(L')) = w(N^+_{G}(L'))$.
\end{observation}
%The following observation is immediate from the construction of $G_B$.
%\begin{observation}[Special Case of Claim 2.3 from (CITE SODA VERSION OS CMT PAPER)]\label{obs:cuts-in-G_B-have-same-value-as-in-G:vertex-version}
%	For every set $L' \subseteq A_B$, we have $w_B(N^+_{G_B}(L')) = w(N^+_{G}(L'))$.
%\end{observation}

Next, we show that the obtained graph $G_B$ indeed satisfies Properties \ref{property:G_B-P1:vertex-version}-\ref{property:G_B-P4:vertex-version}.
Properties \ref{property:G_B-P2:vertex-version} and \ref{property:G_B-P3:vertex-version} follow directly from the construction of $G_B$ and from \Cref{obs:cuts-in-G_B-have-same-value-as-in-G:vertex-version}. (Indeed, for every set $X' \subseteq V(G_B) \setminus (\{y^*\} \cup N^-_{G_B}(y^*))$ as described in Property \ref{property:G_B-P3:vertex-version}, it can be easily verified that $X' \subseteq A_B$ must hold, so \Cref{obs:cuts-in-G_B-have-same-value-as-in-G:vertex-version} applies for this set.)
We now consider Properties \ref{property:G_B-P1:vertex-version} and \ref{property:G_B-P4:vertex-version}.
It can be easily verified from the construction of $G_B$ that $|E(G_B)| = \vol^+_G(A_B) + |N^+_G(A_B)| \leq 2\vol^+_G(A_B)$; so, the upper-bound on $|E(G_B)|$ required for Property \ref{property:G_B-P1:vertex-version} of $G_B$ follows from Property \ref{property:A_B-P'1:vertex-version} of $A_B$; furthermore, a similar upper-bound on $|V(G_B)|$ follows as $|V(G_B)|=|A_B| + |N^+_G(A_B)| +1 \leq |E(G_B)|+1$, where the last inequality is due to the fact that every vertex of $A_B$ must have at least one-outgoing edge in $G$, which in turn holds by the assumption $\OPT_{V,y^*}>0$.
It remains to show that $G_B$ satisfies Property \ref{property:G_B-P4:vertex-version}, which we do next.
Suppose that there exists a terminal $x^* \in B \cap L^*$ and the good event $\event$ occurred; we need to prove that in this case, $x^* \in V(G_B)$ holds, and there exists an $x^*$-$y^*$ vertex-cut $(L',S',R')$ in $G_B$ of value $w_B(S') \leq \left(1 + \frac{i\epsilon}{z}\right)\OPT_{V,y^*}$ that additionally satisfies the condition $\vol^+_G(L') \leq 2\nu$.
Indeed, in the aforementioned case, Property \ref{property:A_B-P'2:vertex-version} of the set $A_B$ guarantees that there exists a subset $L' \subseteq A_B$ with $x^* \in L'$ and $\vol^+_G(L') \leq 2\nu$ such that the vertex-cut $(L',N^+_G(L'),V(G) \setminus (L' \cup N^+_G(L')))$ in $G$ has value $w(N^+_G(L')) \leq \left(1 + \frac{i\epsilon}{z}\right)\OPT_{V,y^*}$.
Furthermore, by \Cref{obs:cuts-in-G_B-have-same-value-as-in-G:vertex-version}, this set $L'$ satisfies $w_B(N^+_{G_B}(L')) = w(N^+_{G}(L'))\leq \left(1 + \frac{i\epsilon}{z}\right)\OPT_{V,y^*}$. Thus, if we define $S'$ as $S' = N^+_{G_B}(L')$ and $R'$ as $R'=V(G_B) \setminus (L' \cup S')$, then the vertex-cut $(L',S',R')$ in $G_B$ satisfies all the conditions required by Property \ref{property:G_B-P4:vertex-version}.
Furthermore, since $x^* \in L'$ and $L' \subseteq A_B \subseteq V(G_B) \setminus N^-_{G_B}(y^*)$, it follows that $x^* \in V(G_B) \setminus N^-_{G_B}(y^*)$, as needed.
This concludes the proof of Property \ref{property:G_B-P4:vertex-version}, and thus concludes the proof of \Cref{cl:reduction-to-computing-A_B:vertex-version}.
\section{Reducing Dependence on $W$}\label{sec: low-dependence-on-W}

In this section, we describe a standard method that we use in our algorithms for both the Rooted Minimum Edge-Cut and the Rooted Minimum Vertex-Cut problem, in order to reduce the dependence of the running time on the maximum weight $W$ from polylogarithmic to logarithmic.
For the sake of completion, we also provide a formal proof.
%In this section, we describe a standard method that we use in our algorithms for both the Global and Rooted variants of the Minimum Edge-Cut and Minimum Vertex-Cut problems, in order to reduce the dependence of the running time on the maximum weight $W$ from polylogarithmic to logarithmic.
Specifically, this method allows us to transform any $(1+\epsilon)$-approximation algorithm for these problems that runs in time $O\left(\frac{m^{1+o(1)}}{\epsilon}\polylog W\right)$ on an $m$-edge graph and has a success probability of $\frac{1}{2}$, to one that runs in time $O\left(\left(\frac{m^{1+o(1)}}{\epsilon}\right)\polylog\left(\frac{m}{\epsilon}\right)\log W\right)$ and has the same success probability.
As long as $\frac{1}{\epsilon} \leq \poly(m)$, the latter running time is upper-bounded by $O\left(\frac{m^{1+o(1)}}{\epsilon}\log W\right)$, as required by both \Cref{thm : main-thm-rooted-edge-version} and \Cref{thm : main-thm-rooted-vertex-version}.
In the case where the inequality $\frac{1}{\epsilon} \leq \poly(m)$ does not hold, and in particular whenever $\frac{1}{\epsilon} > m$, we can instead solve the Rooted Minimum Edge-Cut (resp. Vertex-Cut) problem with root $y^*$ by explicitly computing the minimum $v$-$y^*$ edge-cut (resp. vertex-cut) for every vertex $v$, resulting in an algorithm with running time $O(m^{1+o(1)}\cdot n\cdot \log W) \leq O(m^{2+o(1)} \log W) \leq O(\frac{m^{1+o(1)}}{\epsilon} \log W)$.

The method that we describe in this section is applicable for a general class of problems, and we will describe it using a general formulation that captures both the Rooted Minimum Edge-Cut problem and the Rooted Minimum Vertex-Cut problem.
%The method that we describe in this section is applicable for a general class of problems, and we will describe it using a general formulation that captures both the Rooted and Global variants of both the Minimum Edge-Cut problem and the Minimum Vertex-Cut problem.
We now present this formulation:

Consider a fixed mapping $\varphi$ that, given a universe $U$ of elements and a bit-string $s$, maps these to a set $\varphi(U,s)$ of \emph{feasible solutions}, where each element $F$ of $\varphi(U,s)$ is a pair $F=(Q,s')$ consisting of a subset $Q \subseteq U$ of elements, and a bit-string $s'$.
Now, consider the following optimization problem $P_{\varphi}$.
In this problem, the input consists of a universe $U$, a bit-string $s$, a positive integer $W$, and weight-function $w:U \to \{1,\ldots,W\}$ that assigns a positive integral weight $w(e) \leq W$ to each element $e \in U$; the goal is finding a feasible solution $F = (Q,s') \in \varphi(U,s)$ that minimizes the sum of weights $\sum_{e \in Q}w(e)$.

Since a graph can be represented as a bit-string, it is possible to formulate both the Rooted Minimum Edge-Cut and the Rooted Minimum Vertex-Cut problems as a problem $P_{\varphi}$ for the appropriate mapping $\varphi$. For example, in the case of Rooted Minimum Edge-Cut, the string $s$ represents the graph $G$ and the root $y^*$; the universe $U$ is the set of edges $U=E(G)$; and the set $\varphi(U,s)$ of feasible solutions consists of all pairs $(Q,s')$ such that $s'$ is a bit-string representing an edge-cut $(X,Y)$ with $y^* \in Y$, and $Q$ is the set $E_G(X,Y)$ of edges crossing this cut.
The following theorem demonstrates that for any minimization problem that is captured by the above setup, it is possible to reduce the running time dependence of any algorithm on the parameter $W$ from polylogarithmic to logarithmic.

\begin{theorem}\label{thm:reducing-dependence-on-W}
	Consider an optimization problem $P_{\varphi}$ as defined above, and suppose there exists an algorithm $\aalg$ that, given any instance $(U,s,W,w)$ of this problem and any parameter $\epsilon \in (0,1)$, outputs a feasible solution $F \in \varphi(U,s)$ so that with probability at least $\frac{1}{2}$, $F$ is a $(1+\epsilon)$-approximate solution to the instance $(U,s,W,w)$ of the problem.
	Furthermore, suppose that there exists some function $T(\cdot,\cdot)$ such that the running time of algorithm $\aalg$ on any instance $(U,s,W,w)$ and parameter $\epsilon \in (0,1)$ is $O\left(T(s,\epsilon) \cdot \polylog(W)\right)$.
	Then, there exists an algorithm $\aalg'$ that, given the same input as $\aalg$, runs in time $O\left((T(s,\epsilon)+|U|) \cdot \polylog(\frac{|U|}{\epsilon}) \cdot \log (W|U|)\right)$ and computes a feasible solution $F \in \varphi(U,s)$ so that with probability at least $\frac{1}{2}$, $F$ is a $(1+O(\epsilon))$-approximate solution.
\end{theorem}

Essentially, the algorithm $\aalg'$ promised by \Cref{thm : main-thm-edge-version} works by reducing the input instance $(U,s,W,w)$ to $z\approx \log(W|U|)$ many instances $(U,s,W',w_1),\ldots,(U,s,W',w_z)$ with a maximum weight of $W'=\poly(\frac{|U|}{\epsilon})$, and then calling $\aalg$ to solve each of these instances.
Informally, for each $i \in [z]$, the instance $(U,s,W',w_i)$ is designed to "mimic" the original instance $(U,s,W,w)$ in the case where the optimal solution value $\OPT$ of the original instance falls in the range $[2^i,2^{i+1})$.
Specifically, the instance $(U,s,W',w_i)$ is designed so that, if $\OPT \in [2^i,2^{i+1})$, then any $(1+\epsilon)$-approximate solution to this instance is also a $(1+O(\epsilon))$-approximate solution to the original instance $(U,s,W,w)$.
Then, after using $\aalg$ on each of the instances $(U,s,W',w_1),\ldots,(U,s,W',w_z)$, there is a probability of $\frac{1}{2}$ that at least one of the obtained feasible solutions is also a $(1+O(\epsilon))$-approximate solution to the original instance $(U,s,W,w)$.
%Specifically, the instance $(U,s,W',w_i)$ is constructed such that if $\OPT \in [2^i,2^{i+1}]$, then any $(1+\epsilon)$-approximate solution to $(U,s,W',w_i)$ is also a $(1+O(\epsilon))$-approximate solution to the original instance; furthermore, if $\OPT > 2^{i+2}$ then
Next, we provide a formal proof of \Cref{thm:reducing-dependence-on-W}.

%\begin{claim}
%	Consider an instance $(U,s,W,w)$ for problem $P_{\varphi}$ as defined above, and let $W'=$\rnote{!} and $z=\ceil{\log(W|U|+1)}$.
%	There exists an algorithm that, given such an instance $(U,s,W,w)$, computes in time $O\left(|U| \cdot \polylog(\frac{|U|}{\epsilon}) \cdot \log (W|U|)\right)$ weight functions $w_1,\ldots,w_z:U \to [W']$ such that, for every $i \in [z]$, the instance $(U,s,W',w_i)$ satisfies:
%	\begin{itemize}
%		\item Every feasible solution $F \in \varphi(U,s)$ satisfies $\val_i(F) \leq \frac{\val(F) \cdot W'}{2^{i+2}}$, where $\val_i(F)$ denote the value $\val_i(F) = \sum_{e \in }$
%	\end{itemize}
%\end{claim}

\paragraph{Proof of \Cref{thm:reducing-dependence-on-W}.}
We begin by formally describing the algorithm $\aalg'$.
Given an instance $(U,s,W,w)$ for problem $P_{\varphi}$, the algorithm first computes and writes down the values $z=\ceil{\log(W|U|+1)}$ and $W'$, where $W'$ is the smallest integral power of $2$ that is larger than $\frac{4|U|}{\epsilon}$.
Then, for every $i \in [z]$, the algorithm explicitly computes and writes down the weight function $w_i:U \to [W']$ defined as follows:
\[
 w_i(e) = \begin{cases}
 	\floor{\frac{w(e) \cdot W'}{2^{i+2}}} \qquad &\text{if $w(e) < 2^{i+2}$,}\\
 	W' \qquad &\text{otherwise.}
 \end{cases}
\]
We discuss later how to efficiently compute these weight functions $w_1,\ldots,w_z$.
Now, we call algorithm $\aalg$ on each of the instances $(U,s,W',w_1),\ldots,(U,s,W',w_z)$.
For each $i \in [z]$, let $F_i=(Q_i,s'_i)$ denote the feasible solution that $\aalg$ returned when it was ran on instance $(U,s,W',w_i)$.
Observe that explicitly computing the value $\val(F_i) = \sum_{e \in Q_i}w(e)$ of all $z$ feasible solutions $F_1,\ldots,F_z$ w.r.t the weight function $w$ would require time that is quadratic in $\log(W)$, which we cannot afford. Thus, we take a slightly different approach, as described next.
For each $i \in [z]$, we compute the value $\val_i(F_i) = \sum_{e \in Q_i}w_i(e)$ of $F_i$ w.r.t the weight function $w_i$. Since the values $w_i(e)$ assigned by the weight function $w_i$ are written using just $O(\log(W'|U|))$ bits, these computations together take a total time of only $O(|U| \cdot \log(W'|U|) \cdot z) = O(|U| \cdot \log(\frac{|U|}{\epsilon}) \cdot \log(W|U|))$.
Let $i^*$ be the smallest $i \in [z]$ for which $\val_i(F_i) < W'$.
Then, for every $i \in \{i^*,i^*+1,i^*+2\}$, our algorithm calculates the value $\val(F_i)$ of $F_i$ w.r.t the original weight function $w$. Since we are computing $\val(F_i)$ for only a constant number of feasible solutions $F_i$, these computations can be performed in time $O(|U| \cdot \log(W|U|))$.
The algorithm then outputs the feasible solution that minimizes $\val(F_{i})$ among the three solutions $\{F_{i} \mid i \in \{i^*,i^*+1,i^*+2\}\}$, and outputs this feasible solution.

We now analyze the above algorithm.
For the analysis, let $\OPT$ denote the optimal solution value of the original input instance $(U,s,W,w)$, and let $i'$ be the unique integer for which $2^{i'} \leq \OPT < 2^{i'+1}$ holds.
Let $\event$ denote the good event in which the feasible solution $F_{i'}$ obtained from the call to $\aalg$ on the instance $(U,s,W',w_{i'})$ is a $(1+\epsilon)$-approximate solution to the instance $(U,s,W',w_{i'})$.
By the guarantees of the algorithm $\aalg$, event $\event$ happens with probability at least $\frac{1}{2}$.
We will show that whenever event $\event$ happens, our algorithm outputs a $(1+O(\epsilon))$-approximate solution to the original instance.
Since the output of the algorithm is the best feasible solution among $F_{i^*},F_{i^*+1},F_{i^*+2}$, it suffices to prove the following claim.
\begin{claim}\label{cl:inner-claim:reducing-dependence-on-W}
	Whenever event $\event$ happens, the feasible solution $F_{i'}$ must be a $(1+2\epsilon)$-approximate solution to the original instance $(U,s,W,w)$, and $i' \in \{i^*,i^*+1,i^*+2\}$ must hold.
\end{claim}
\begin{proof}
	We begin by proving that if event $\event$ happens, then $F_{i'}$ is a $(1+2\epsilon)$-approximate solution to the instance $(U,s,W,w)$.
	In other words, we show next that if event $\event$ happens, then $\val(F_{i'}) \leq (1+2\epsilon)\OPT$.
	To do this, we will first show that if $\event$ happens then $\val_{i'}(F_{i'})$ must be small, and then we will relate $\val(F_{i'})$ to $\val_{i'}(F_{i'})$.
	Indeed, recall the definition of the weight function $w_{i'}$, and observe that by this definition, every element $e \in U$ satisfies $w_{i'}(e) \leq \frac{w(e) \cdot W'}{2^{i'+2}}$.
	Thus, the value of the optimal solution to the instance $(U,s,W',w_{i'})$ must be at most $\frac{\OPT \cdot W'}{2^{i'+2}}$.
	Recall also that if event $\event$ happens, then $F_{i'}$ must be a $(1+\epsilon)$-approximate solution to the instance $(U,s,W',w_{i'})$, meaning that $\val_{i'}(F_{i'}) \leq (1+\epsilon)\frac{\OPT \cdot W'}{2^{i'+2}}$.
	Since $i'$ is defined such that $\OPT < 2^{i'+1}$, and since $\epsilon<1$, it follows in particular that $\val_{i'}(F_{i'}) < W'$, meaning that $Q_{i'}$ contains only elements $e \in U$ with $w(e) < 2^{i'+2}$.
	Thus, by the definition of $w_{i'}$, it follows that $w(e) \leq \frac{2^{i'+2}}{W'}(1 + w_{i'}(e))$ holds for all $e \in Q_{i'}$, so $\val(F_{i'}) \leq \frac{2^{i'+2}}{W'}(|U|+\val_{i'}(F_{i'}))$.
	Since we've seen that $\val_{i'}(F_{i'}) \leq (1+\epsilon)\frac{\OPT \cdot W'}{2^{i'+2}}$, it follows that $\val(F_{i'}) \leq \frac{2^{i'+2} |U|}{W'} + (1+\epsilon)\OPT$.
	Now, since $i'$ is defined such that $2^{i'} \leq \OPT$, it follows that $\val_{i'}(F_{i'}) \leq \frac{4 \OPT \cdot |U|}{W'} + (1+\epsilon)\OPT \leq (1+2\epsilon)\OPT$, as we needed. (In the last inequality, we used the definition of $W'$.)
	
	It remains to show that if event $\event$ happens, then $i' \in \{i^*,i^*+1,i^*+2\}$ must hold.
	Recall that $i^*$ is defined as the smallest integer $i \in [z]$ for which $\val_i(F_i) < W'$.
	Further recall that we've already shown that if event $\event$ happens, then $\val_{i'}(F_{i'}) < W'$ holds.
	Thus, it only remains to show that $\val_{i}(F_i) \geq W'$ holds for every $i < i' -2$.
	Indeed, consider any $i < i'-2$, and consider the solution $F_i=(Q_i,s_i')$.
	If $Q_i$ contains any element $e \in U$ with $w(e) \geq 2^{i+2}$, then $w_i(e) = W'$ and thus $\val_i(F_i) \geq W'$ holds.
	Otherwise, all elements $e \in Q_i$ satisfy $w_i(e) = \floor{\frac{w(e) \cdot W'}{2^{i+2}}}$, meaning that $\val_i(F_i) \geq \frac{\val(F_i) \cdot W'}{2^{i+2}} - |U| \geq \frac{\OPT \cdot W'}{2^{i+2}} - |U|$.
	It remains to show that $\frac{\OPT \cdot W'}{2^{i+2}} - |U| \geq W'$.
	Indeed, since $i < i'-2$, it must be that $i \leq i'-3$, and thus $2^{i+2} \leq 2^{i'-1} \leq \OPT/2$;
	it therefore follows that $\frac{\OPT \cdot W'}{2^{i+2}} - |U| \geq 2W' - |U| \geq W'$.
	This concludes the proof of \Cref{cl:inner-claim:reducing-dependence-on-W}.
\end{proof}

In order to conclude the proof of \Cref{thm:reducing-dependence-on-W}, it remains to show how to efficiently compute the weight functions $w_1,\ldots,w_z$. We show this next.

\paragraph{Efficient computation of the Weight Functions.}
We now describe how the algorithm computes the weight functions  $w_1,\ldots,w_z$ in time $O\left(|U| \cdot \polylog(\frac{|U|}{\epsilon}) \cdot \log (W|U|)\right)$.
To do this, the algorithm first writes the weight $w(e)$ of each element $e \in U$ using the standard binary representation of natural numbers in $\floor{1+\log W}$ bits, which can be done in time $O(|U| \cdot \log W)$.
Furthermore, for each $e \in U$, the algorithm locates the most-significant non-zero bit of this representation in order to determine the largest integer $i$ for which $w(e) \geq 2^i$ holds, and writes down this value, which we denote by $i_e$.
Observe that for every $e \in U$ and every $i \in [z]$, if $i_e \leq i+1$, then all but the first (least significant) $i+2$ bits in the binary representation of $w(e)$ are zero; in this case, the choice of $w_i(e) = \floor{\frac{w(e) \cdot W'}{2^{i+2}}}$ means that the binary representation of $w_i(e)$ is equal to the substring of the binary representation of $w(e)$, beginning at the $(i+3-\log W')$'th least significant bit and ending at the $(i+2)$'th least significant bit.
In the remaining case where $i_e \geq i+2$, it follows that $w(e) \geq 2^{i+2}$ and thus $w_i(e)=W'$, so $w_i(e)$ can be computed by copying the value $W'$ that the algorithm already wrote down beforehand.
Thus, for a given element $e$, we can compute the values $w_1(e),\ldots,w_z(e)$ in total time $O(z \cdot \log(W')) = O\left(\polylog(\frac{|U|}{\epsilon}) \cdot \log (W|U|)\right)$, as needed.

%We now analyze the above algorithm.
%For the analysis, let $\OPT$ denote the optimal solution value of the original input instance $(U,s,W,w)$, and let $i'$ be the unique integer for which $2^{i'} \leq \OPT < 2^{i'+1}$.
%By the guarantees of the algorithm $\aalg$, there is a probability of at least $\frac{1}{2}$ that the feasible solution $F_{i'}$ obtained from the call to $\aalg$ on the instance $(U,s,W',w_{i'})$ is a $(1+\epsilon)$-approximate solution to the instance $(U,s,W',w_{i'})$.
%We will show that whenever this happens, our algorithm outputs a $(1+O(\epsilon))$-approximate solution to the original instance.
%To do this, we must show that whenever $F_{i'}$ is a $(1+\epsilon)$-approximate solution to $(U,s,W',w_{i'})$, there must be some $\{F_{i} \mid i \in \{i^*,i^*+1,i^*+2\}\}$ such that $F_i$ is a $(1+O(\epsilon))$-approximate solution to the original instance $(U,s,W,w)$ -- specifically, we will show that in this case, $i' \in \{i^*,i^*+1,i^*+2\}$ must hold, and $F_{i'}$ must be a $(1+O(\epsilon))$-approximate solution to the original instance.

\section{Missing Proofs from \Cref{sec:prelims}}\label{sec:missing-proofs-in-prelims}

\paragraph{Proof of \Cref{lem:submodularity-of-minimum-cuts:edge-version}.}
Consider a graph $G$ with weights $w(e)$ on the edges $e \in E(G)$, as well as a pair $s,t$ of vertices in $G$, and a minimum $s$-$t$ edge-cut $(S,T)$ in $G$.
Furthermore, let $X \subseteq V(G)$ be any set of vertices.
We need to prove that $w(\partial^+_G(X \cap S)) \leq w(\partial^+_G(X))$.
Indeed, as $(S,T)$ is a minimum $s$-$t$ edge-cut in $G$, we have $w(\partial^+_G(S)) \leq w(\partial^+_G(S \cup X))$.
Thus, the desired inequality follows if we can prove that $w(\partial^+_G(X \cap S)) + w(\partial^+_G(S \cup X)) \leq w(\partial^+_G(X)) + w(\partial^+_G(S))$.
Indeed, it is not hard to verify that every edge that contributes its weight to the left side of this inequality must also contribute its weight to the right side of this inequality, and that every edge which contributes its weight twice to the left side of the inequality must also contribute its weight twice to the right side of the inequality.

\paragraph{Proof of \Cref{lem:submodularity-of-minimum-cuts:vertex-version}}
The proof of this lemma is essentially identical to that of \Cref{lem:submodularity-of-minimum-cuts:edge-version}, except it uses the operator $w(N^+_G(\cdot))$ instead of $w(\partial^+_G(\cdot))$.

\section{Missing Proofs from \Cref{sec:rooted-minimum-edge-cut}}\label{sec:missing-proofs-in-main-edge-version}

\paragraph{Proof of \Cref{cl:reduction-to-computing-A_B:edge-version}.}
Consider a set $A_B \subseteq V(G) \setminus \set{y^*}$ of vertices that satisfies properties \ref{property:A_B-P'1:edge-version}-\ref{property:A_B-P'2:edge-version}, and define $G_B$ as stated in \Cref{cl:reduction-to-computing-A_B:edge-version}. Specifically, $G_B$ is constructed from $G$ by merging all vertices outside of $A_B$ into $y^*$ and then removing any resulting self-loops, as well as any edges outgoing from $y^*$ in the resulting graph.
We need to prove that $G_B$ satisfies Properties \ref{property:G_B-P1:edge-version}-\ref{property:G_B-P4:edge-version}.

We begin with Properties \ref{property:G_B-P2:edge-version} and \ref{property:G_B-P3:edge-version}, which follow easily from the definition of the graph $G_B$: indeed, Property \ref{property:G_B-P2:edge-version} holds because $V(G_B)=A_B \cup \set{y^*}$, and Property \ref{property:G_B-P3:edge-version} is a special case of the next observation.
\begin{observation}\label{obs:cuts-in-G_B-have-same-value-as-in-G:edge-version}
	For every set $X' \subseteq A_B$, we have $w_B(\partial^+_{G_B}(X')) = w(\partial^+_{G}(X'))$.
\end{observation}
\begin{proof}
	Indeed, each edge exiting $X'$ in $G$ induces a corresponding edge in $G_B$ that also exits $X'$, and has the same weight;
	furthermore, each edge exiting $X'$ in $G_B$ is induced by such an edge. This concludes the proof of the observation.
\end{proof}

Next, we prove Property \ref{property:G_B-P1:edge-version}, which requires upper-bounding the number of edges and vertices in $G_B$:
it is easy to verify that $|E(G_B)| = \vol^+_G(A_B)$ holds; so, as the set $A_B$ satisfies Property \ref{property:A_B-P'1:edge-version}, it follows that $|E(G_B)| = \Tilde{O}\left(\frac{m}{2^i \cdot \epsilon}\right)$, as needed for Property \ref{property:G_B-P1:edge-version}. To bound the number of vertices in $G_B$, we use the fact that $V(G_B) = A_B \cup \set{y^*}$, as well as the assumption that $\OPT_{E,y^*} > 0$; based on this assumption, every vertex $v \in V(G) \setminus \set{y^*}$ must have at least one outgoing edge in $G$, and so, by the construction of $G_B$, every vertex of $A_B$ must have at least one outgoing edge in $G_B$, meaning that $|E(G_B)| \geq |A_B| = |V(G_B)| - 1$, and thus that $|V(G_B)| \leq |E(G_B)|+1 = \Tilde{O}\left(\frac{m}{2^i \cdot \epsilon}\right)$. It is also easily verified that $G_B$ does not have more edges or vertices than $G$. This concludes the proof that $G_B$ satisfies Property \ref{property:G_B-P1:edge-version}.

It remains to prove that $G_B$ satisfies Property \ref{property:G_B-P4:edge-version}, which we do next.
Suppose that there exists a terminal $x^* \in B \cap X^*$ and the good event $\event$ occurred; we need to prove that in this case, $x^* \in V(G_B)$ holds, and there exists an $x^*$-$y^*$ edge-cut $(X',Y')$ in $G_B$ of value $w_B(E_{G_B}(X',Y')) \leq \left(1 + \frac{i\epsilon}{z}\right)\OPT_{E,y^*}$ that additionally satisfies the condition $\vol^+_G(X') \leq 2\nu$.
Indeed, in the aforementioned case, Property \ref{property:A_B-P'2:edge-version} of the set $A_B$ guarantees that there exists a subset $X' \subseteq A_B$ with $x^* \in X'$ and $\vol^+_G(X') \leq 2\nu$ such that the edge-cut $(X',V(G) \setminus X')$ in $G$ has value $w(E_{G}(X',V(G) \setminus X')) \leq \left(1 + \frac{i\epsilon}{z}\right)\OPT_{E,y^*}$.
Furthermore, by \Cref{obs:cuts-in-G_B-have-same-value-as-in-G:edge-version}, this set $X'$ satisfies $w_B(\partial^+_{G_B}(X')) = w(\partial^+_{G}(X')) = w(E_{G}(X',V(G) \setminus X'))\leq \left(1 + \frac{i\epsilon}{z}\right)\OPT_{E,y^*}$. Thus, if we define $Y'$ as $Y' = V(G_B) \setminus X'$, then the edge-cut $(X',Y')$ in $G_B$ satisfies all the conditions required by Property \ref{property:G_B-P4:edge-version}.
Furthermore, since $x^* \in X'$ and $X' \subseteq A_B \subseteq V(G_B)$, it follows that $x^* \in V(G_B)$, as needed.
This concludes the proof of Property \ref{property:G_B-P4:edge-version}, and thus concludes the proof of \Cref{cl:reduction-to-computing-A_B:edge-version}.

\section{Removing the Assumption $w(L^*) \leq w(R^*)$ in the Proof of \Cref{thm:main-reduction}}\label{sec:removing-assumption-in-reduction:reduction}
In \Cref{sec:main-reduction}, we designed a procedure \algglobalvertexcut that implements \Cref{thm:main-reduction} under the following additional assumption:
\begin{assumption}\label{assump:side-weight-assumption-for-reduction}
	The graph $G$ in the input to \Cref{thm:main-reduction} has some global minimum vertex-cut $(L^*,S^*,R^*)$ with $w(L^*) \leq w(R^*)$.
\end{assumption}
This assumption can be removed using standard techniques: informally, we implement \Cref{thm:main-reduction} by running the procedure \algglobalvertexcut once on the input graph $G$, and once on the reversed graph $\overline{G}$, and selecting the best of the two obtained solutions.

For the sake of completion, we now explain this in detail.
Observe that, by the guarantees of the procedure \algglobalvertexcut, there exists some absolute constant $C$ and some function $T:\mathbb{N} \to \mathbb{N}$ satisfying $T(m)=\Tilde{O}(m)$, so that, when the procedure is given an input satisfying \Cref{assump:side-weight-assumption-for-reduction} and oracle-access as described in \Cref{thm:main-reduction}, it runs in time at most $T(m)\cdot \log W$, and makes at most $C \cdot \frac{m}{n}$ queries to the oracle, such that the total number of edges in these queries is at most $C \cdot m$.
Furthermore, in this case, the output of the procedure is a vertex-cut $(L,S,R)$ in the input graph $G$ so that, with probability at least $\frac{1}{2}$, its weight is at most $(1+\epsilon)\OPT_G$, where $\OPT_G$ is the value of the global minimum vertex-cut in $G$.

Then, to implement \Cref{thm:main-reduction}, we run two instances of this procedure \algglobalvertexcut; one instance is given the input graph $G$, and the other is given the reversed graph $\overline{G}$ -- this is the graph obtained from $G$ by reversing the direction of each edge. Both instances are given the same vertex-weights, parameters, and oracle access.
For each of these instances, we terminate it early if it: runs for longer than $T(m) \cdot \log W$ time; makes more than $C \cdot \frac{m}{n}$ oracle queries; or, makes oracle queries whose total edge-count is larger than $C \cdot m$. (If an instance is about to make a single query with more than $C \cdot m$ edges, then we terminate it before it makes this query.)
If an instance was not terminated early, and if its output is a valid vertex-cut in $G$, then we say that it "terminated naturally". (Observe that we can check in time $O(m)$ whether the output is a valid vertex-cut in $G$.)
Next, we continue the description of the algorithm.

If the instance on $G$ terminated naturally, then let $(L',S',R')$ be the vertex-cut obtained by it; otherwise, let $(L',S',R')$ be an arbitrary vertex-cut in $G$.
Similarly, if the instance on $\overline{G}$ terminated naturally, then let $(L'',S'',R'')$ be the vertex-cut obtained by it; otherwise, let $(L'',S'',R'')$ be an arbitrary vertex-cut in $\overline{G}$.
Observe that $(R'',S'',L'')$ must be a valid vertex-cut in $G$, and its weight, $w(S'')$, is the same as the weight of the vertex-cut $(L'',S'',R'')$ in $\overline{G}$.
The algorithm outputs the smaller-weight vertex-cut among the two cuts $(L',S',R')$ and $(R'',S'',L'')$ in $G$.

For the analysis, it is not hard to verify that the value of the global minimum vertex-cut of the reversed graph $\overline{G}$ is the same as of the original graph $G$, and that at least one of these two instances of the Global Minimum Vertex-Cut problem must satisfy \Cref{assump:side-weight-assumption-for-reduction}; thus, by the guarantees of \algglobalvertexcut, at least one of the calls to \algglobalvertexcut must terminate naturally, and, with probability at least $\frac{1}{2}$, the vertex-cut obtained by this call has weight at most $(1+\epsilon)\OPT_{G}$.
Therefore, with probability at least $\frac{1}{2}$, at least one of the two vertex-cuts $(L',S',R')$ and $(L'',S'',R'')$ has weight at most $(1+\epsilon)\OPT_{G}$, in which case the output of the algorithm is a vertex-cut in $G$ with weight at most $(1+\epsilon)\OPT_{G}$, as required.
\section{Extension to Non-Negative Weights}\label{sec:zero-weights}
In this section, we briefly explain how our main theorems can be extended to the setting in which edges and vertices are allowed to have weight $0$.
We first note that this extension is trivial to achieve for the Minimum Edge-Cut problem, as zero-weight edges do not affect the solution value or the set of feasible solutions, and can thus be ignored.
In order to extend the algorithms for the Minimum \emph{Vertex-Cut} problem to the setting of non-negative vertex-weights, we use a simple and standard reduction from this setting to the positive-weight setting.
The same reduction was also used in, for example, \cite{CMT25}; but we analyze the reduction again here in order to convince the reader that it also works for approximate algorithms, rather than just for exact algorithms as was shown in \cite{CMT25}.
For the sake of brevity, we will analyze only the \emph{Global} Minimum Vertex-Cut problem, though the same reduction also works for the Rooted Minimum Vertex-Cut problem.

Consider an instance $(V,E,w)$ of the Global Minimum Vertex-Cut problem in the non-negative weight setting.
To solve this instance, we transform it into an instance $(V,E,w')$ with only positive weights, by setting the new weight of each vertex $v$ to $w'(v) \leftarrow 4n^2 w(v) + 1$.

Now, it is immediate to verify that any optimal solution of instance $(V,E,w')$ must also be an optimal solution of instance $(V,E,w)$, meaning that solving the instance $(V,E,w)$ exactly can be reduced to solving $(V,E,w')$ exactly.
Demonstrating a similar relation for approximate solutions is only a little more complicated: specifically, we show next that if one wishes to find a $(1+\epsilon)$-approximate solution to the original instance $(V,E,w)$, then it suffices to find a $(1+\epsilon')$-approximate solution for the new instance $(V,E,w')$, where $\epsilon'=\epsilon/2$.
\begin{proof}
	Indeed, consider any optimal solution $(L,S,R)$ of the original instance $(V,E,w)$, and observe that $4n^2 w(S) \leq w'(S) \leq 4n^2 w(S) + n$ must hold.
	Suppose first that the weight of $(L,S,R)$ in the original instance $(V,E,w)$ satisfies $w(S) < \frac{1}{\epsilon}$ -- it is then easy to verify that every solution $(L',S',R')$ of the original instance with weight $w(S') \geq w(S)+1$ must have its weight in the new instance satisfy $w'(S') \geq 4n^2 (w(S) + 1) > (1+\epsilon')\cdot(4n^2 w(S) + n) \geq (1+\epsilon') w'(S)$, meaning that every non-optimal solution of the original instance is also not a $(1+\epsilon')$-approximate solution of the new instance.
	We now consider the remaining case, in which $w(S) \geq \frac{1}{\epsilon}$.
	In this case, $w'(S) \leq (1+\frac{\epsilon}{4n}) \cdot 4n^2 w(S)$ must hold, meaning that $(1 - \frac{\epsilon}{4n}) w'(S) \leq 4n^2 w(S)$.
	It then follows that for every vertex-cut $(L',S',R')$ in $G$ with $w(S') > (1+\epsilon)w(S)$, the inequality
	\[
	 w'(S') > 4n^2 \cdot (1+\epsilon)w(S) \geq (1+\epsilon) \cdot (1 - \frac{\epsilon}{4n}) w'(S) \geq (1+\epsilon')w'(S)
	\]
	holds, meaning that $(L',S',R')$ is not a $(1+\epsilon')$-approximate solution of the new instance.
\end{proof}
\newpage
\bibliographystyle{alpha}
\bibliography{approx-global-mincut}

\end{document}